%% file: main.tex
\documentclass[acmsmall]{acmart}

\AtBeginDocument{%
}

\setcopyright{acmlicensed}
\acmDOI{10.1145/3729303}
\acmYear{2025}
\acmJournal{PACMPL}
\acmVolume{9}
\acmNumber{PLDI}
\acmArticle{200}
\acmMonth{6}
\received{2024-11-15}
\received[accepted]{2025-03-06}

\input{preamble}
\input{macros}

\begin{document}

\title{$\tool$: Symbolic Reasoning for SQL using Conflict-Driven Under-Approximation Search}

\author{Pinhan Zhao}
\orcid{0009-0002-1149-0706}
\affiliation{%
  \institution{University of Michigan}
  \city{Ann Arbor}
  \country{USA}
}
\email{pinhan@umich.edu}

\author{Yuepeng Wang}
\orcid{0000-0003-3370-2431}
\affiliation{%
  \institution{Simon Fraser University}
  \city{Burnaby}
  \country{Canada}
}
\email{yuepeng@sfu.ca}

\author{Xinyu Wang}
\orcid{0000-0002-1836-0202}
\affiliation{%
  \institution{University of Michigan}
  \city{Ann Arbor}
  \country{USA}
}
\email{xwangsd@umich.edu}

\renewcommand{\shortauthors}{Pinhan Zhao, Yuepeng Wang, and Xinyu Wang}

\input{sections/abstract}



\maketitle


\input{sections/intro}

\input{sections/overview}

\input{sections/alg}

\input{sections/eval}

\input{sections/related}
\input{sections/conc}
\input{sections/ack}

\input{sections/data-availability-statement}

\bibliography{main}

\newpage
\appendix
\input{sections/appendix/ua-others}
\input{sections/appendix/encoding}
\input{sections/appendix/proof}
\input{sections/appendix/disambiguation-cond}

\end{document}

%% file: preamble.tex
\usepackage{enumitem}
\usepackage[noend]{algpseudocode}
\usepackage{algorithmicx}
\usepackage{algorithm}
\usepackage{amsmath}
\usepackage{tikz-qtree}
\usepackage[linguistics]{forest}
\usetikzlibrary{fit}
\usepackage{xspace}
\usepackage{marginnote}

\usepackage{listings}

\usepackage{multirow}
\usepackage{array}

%% file: macros.tex
\newcommand{\xinyurevision}[1]{{#1}}

\newcommand{\tableattrs}{\textnormal{Attrs}}

\newcommand{\truepredicate}{\textnormal{True}}

\newcommand{\ruleleadsto}{\twoheadrightarrow}

\newcommand{\truevalue}{\texttt{T}}
\newcommand{\falsevalue}{\texttt{F}}
\newcommand{\unknownvalue}{\star}
\newcommand{\unknown}{\unknownvalue}

\newcommand{\XWnote}[1]{\marginnote{\scriptsize {\color{red}#1}}}

\newcommand{\evalfinding}[1]{
\vspace{1pt}
\begin{center}
\fcolorbox{black}{white}{\parbox{.96\linewidth}{
{#1}
}}
\vspace{1pt}
\end{center}
}

\newcommand{\tool}{\xspace\textsc{Polygon}\xspace}
\newcommand{\xdata}{\xspace\textsc{XData}\xspace}
\newcommand{\verieql}{\xspace\textsc{VeriEQL}\xspace}
\newcommand{\cubes}{\xspace\textsc{Cubes}\xspace}
\newcommand{\qex}{\xspace\textsc{Qex}\xspace}
\newcommand{\cosette}{\xspace\textsc{Cosette}\xspace}
\newcommand{\datafiller}{\xspace\textsc{DataFiller}\xspace}
\newcommand{\evosql}{\xspace\textsc{EvoSQL}\xspace}

\algnewcommand\Input{\hspace{-10pt}\textbf{input: }}
\algnewcommand\Output{\hspace{-10pt}\textbf{output: }}
\newcommand{\myprocedure}[1]{\hspace{-10pt}\textbf{procedure }{#1}}

\newcommand{\geninput}{\textsc{GenInput}}
\newcommand{\buildinputdatabase}{\textit{ExtractInputDB}}
\newcommand{\conflictdrivenUAsearch}{\textsc{ConflictDrivenUASearch}}
\newcommand{\fixconflict}{\textsc{ResolveConflict}}
\newcommand{\getmodel}{\textit{GetModel}}
\newcommand{\extractconflict}{\textit{ExtractConflict}}
\newcommand{\splitUAspace}{\textit{CoverUAs}}

\newcommand{\getTopUA}{\textit{GetTopUA}}

\newcommand{\getastnodes}{\textit{ASTNodes}}
\newcommand{\getop}{\textit{op}}
\newcommand{\getqueryop}{\getop}
\newcommand{\encode}{\encodesemantics}
\newcommand{\encodeuasemantics}{\encodesemantics}
\newcommand{\encodefullsemantics}{\textit{EncFullSemantics}}
\newcommand{\encodesemantics}{\textit{EncUASemantics}}
\newcommand{\encodechoice}{\textit{EncUAChoice}}
\newcommand{\encodeconflicts}{\textit{EncConflicts}}

\newcommand{\encodeelemchoice}{\textit{EncUAElemChoice}}
\newcommand{\encodeuavalue}{\textit{EncUAChoiceValue}}

\newcommand{\topUA}{\textit{TopUA}}

\newcommand{\encodingall}{\Phi}
\newcommand{\queryopencoding}{\Psi}

\newcommand{\queryoutput}{y}

\newcommand{\refines}{\sqsubseteq}
\newcommand{\subsumes}{\sqsupseteq}

\newcommand{\newpara}[1]{\vspace{5pt}\noindent\emph{\textbf{#1}}\hspace{4pt}}
\newcommand{\mydots}{\cdots}
\newcommand{\assign}{:=}
\newcommand{\doubleprime}{'\mkern-2mu'}

\newcommand{\worklist}{W}
\newcommand{\dom}{\textnormal{dom}}

\newcommand{\smtmodel}{\sigma}

\newcommand{\sqlquery}{\program}
\newcommand{\program}{P}
\newcommand{\query}{\sqlquery}

\newcommand{\queryop}{F}
\newcommand{\queryopinput}{x}
\newcommand{\queryopoutput}{y}
\newcommand{\astnodeinput}{x}
\newcommand{\astnodeoutput}{y}
\newcommand{\varUAvector}{z}
\newcommand{\labelappcond}{\textit{lc}}
\newcommand{\labelastnode}{\textit{l}\astnode}

\newcommand{\underapproxfamily}{\mathcal{U}}
\newcommand{\underapprox}{u}
\newcommand{\underapproxvar}{z}
\newcommand{\mapastnodetounderapprox}{M}
\newcommand{\astnode}{v}

\newcommand{\conflictASTnodes}{V}
\newcommand{\conflict}{\omega}
\newcommand{\conflicts}{\Omega}

\newcommand{\appcond}{C}

\newcommand{\inputdatabase}{I}

\newcommand{\mapprojection}{\downarrow}

\newcommand{\xinyu}[1]{{\color{red}{XW: #1}}}
\newcommand{\pinhan}[1]{{\color{blue}{#1}}}

\newcommand{\bigland}{\bigwedge}
\newcommand{\biglor}{\bigvee}

\newcommand{\irule}[2]%
   {\mkern-2mu\displaystyle\frac{#1}{\vphantom{,}#2}\mkern-2mu}

\newcommand{\expression}{E}
\newcommand{\attrlist}{L}
\newcommand{\proj}{\text{Project}}
\newcommand{\filter}{\text{Filter}}
\newcommand{\rename}{\text{Rename}}
\newcommand{\product}{\text{Product}}
\newcommand{\ijoin}{\text{IJoin}}
\newcommand{\ljoin}{\text{LJoin}}
\newcommand{\rjoin}{\text{RJoin}}
\newcommand{\fjoin}{\text{FJoin}}
\newcommand{\unionall}{\text{UnionAll}}
\newcommand{\distinct}{\text{Distinct}}
\newcommand{\groupby}{\text{GroupBy}}
\newcommand{\orderby}{\text{OrderBy}}
\newcommand{\with}{\text{With}}
\newcommand{\alljoin}{\otimes}
\newcommand{\allarith}{\diamond}
\newcommand{\alllogic}{\odot}
\newcommand{\nullv}{\text{Null}\xspace}

\newcommand{\mylist}[1]{\{#1\}}
\newcommand{\set}[1]{\{#1\}}
\newcommand{\formula}{\varphi}

\newcommand{\indicator}{\mathbb{I}}

\newcommand{\del}{\text{Del}\xspace}
\newcommand{\querybelongstocluster}{\text{b}\xspace}

\newcommand{\relation}{R}
\newcommand{\tuple}{t}
\newcommand{\predicate}{\phi}

\newcommand{\doublebrackets}[1]{\llbracket #1 \rrbracket}
\newcommand{\denot}[1]{\doublebrackets{#1}}

\newcommand{\cmdCopy}{\textsf{Copy}}
\newcommand{\cmdExist}{\textsf{Exist}}
\newcommand{\cmdCondDel}{\textsf{NotExist}}
\newcommand{\cmdNullTuples}{\textsf{NullTuples}}
\newcommand{\group}{\textsf{g}}

\newcommand{\vectorarray}[1]{%
[#1]
    
}

%% file: sections/abstract.tex
\begin{abstract}

We present a novel symbolic reasoning engine for SQL which can efficiently generate an input $\inputdatabase$ for $n$ queries $\sqlquery_1, \mydots, \sqlquery_n$, such that their outputs on $\inputdatabase$ satisfy a given property (expressed in SMT). 
This is useful in different contexts, such as disproving equivalence of two SQL queries and disambiguating a set of queries. 
Our first idea is to reason about an \emph{under-approximation} of each $\sqlquery_i$---that is, a subset of $\sqlquery_i$'s input-output behaviors. 
While it makes our approach both semantics-aware and lightweight, this idea alone is incomplete (as a fixed under-approximation might miss some behaviors of interest). 
Therefore, our second idea is to perform \emph{search} over an \emph{expressive family} of under-approximations (which collectively cover all program behaviors of interest), thereby making our approach complete. 
We have implemented these ideas in a tool, $\tool$, and evaluated it on over 30,000 benchmarks across two tasks (namely, SQL equivalence refutation and query disambiguation). 
Our evaluation results show that $\tool$ significantly outperforms all prior techniques. 

\end{abstract}

%% file: sections/intro.tex
\section{Introduction}\label{sec:intro}

The general problem we study in this paper is the following.

\vspace{5pt}
\begin{center}
\parbox{.95\linewidth}{
Given $n$ programs $P_1, \mydots, P_n$, how to generate an input $I$ such that their outputs on $I$ satisfy a given property $\appcond$ (which is expressed as an SMT formula over variables $\queryoutput_1, \mydots, \queryoutput_n$). That is, $\appcond \big[ \queryoutput_1 \mapsto \sqlquery_1(I), \mydots, \queryoutput_n \mapsto \sqlquery_n(I) \big]$ is true. 
}
\end{center}
\vspace{5pt}

This problem can be viewed as a form of ``test input generation'' task, which can be instantiated to different applications with different \emph{application conditions} $\appcond$.
One example application is to disprove equivalence of two programs, with $\appcond$ being simply $\queryoutput_1 \neq \queryoutput_2$. 
Another is program disambiguation---e.g., find $I$ that can divide $n$ programs into disjoint groups, such that: 
(i) programs within the same group return the same output given $I$, while 
(ii) outputs of those from different groups are distinct. 
Finding such distinguishing inputs is crucial for example-based program synthesis~\cite{jha2010oracle,mayer2015user,ji2020question}

\newpara{Instantiation to SQL.}
We focus on one instantiation of this general problem, where $\sqlquery_i$'s are written in SQL (which is a widely used domain-specific language).
The aforementioned applications still hold---generating counterexamples to refute SQL equivalence is useful in many ways~\cite{chandra2015data,chandra2019automated,he2024verieql,chu2017cosette}, and query disambiguation is critical for example-based SQL query synthesis ~\cite{brancas2022cubes,manquinho2024towards,wang2017interactive,wang2017synthesizing}.

\newpara{State-of-the-art.}
To the best of our knowledge, no existing work can \emph{efficiently} generate such inputs for an \emph{expressive} subset of SQL. 
While conventional testing-based approaches---e.g., those based on fuzzing~\cite{datafiller-website} and evolutionary search~\cite{castelein2018search}---can generate many inputs quickly, they do not consider the (highly complex) semantics of SQL.
As a result, they fall short of capturing subtle differences across queries, which is crucial for generating distinguishing inputs. 
While some works (e.g., $\xdata$~\cite{chandra2015data,chandra2019automated}) consider basic semantic information (such as join and selection conditions), they support a very limited subset of SQL and are specialized in equivalence checking.
On the other end of the spectrum, techniques based on formal methods~\cite{qex,chu2017cosette,wang2018speeding,he2024verieql,zhaodemonstration} perform symbolic reasoning---typically by encoding \emph{complete} query semantics in SMT. 
While boiling the problem down to SMT solving, these methods often create large SMT formulas that are computationally expensive to solve, especially for problems that involve many large queries with complex semantics. For instance, for operators like group-by and aggregation, fully encoding all grouping possibilities would cause an exponential blow-up in the resulting formula's size.

\newpara{Key challenge.}
While clearly critical to consider \emph{some} semantic information, it would significantly slow down the reasoning process if we consider the \emph{full} semantics. 
The core challenge hence is how to take into account SQL semantics in a way that is \emph{lightweight without hindering completeness.}
Prior works fall into two extremes of the spectrum: 
either (1) fully encoding semantics for all inputs thus heavyweight and not scalable, or 
(2) fast by considering no or very little semantic information but at the cost of missing inputs of interest frequently (i.e., incomplete).

\newpara{Our key insight.}
Our key insight is to reason about 
an \emph{under-approximation (or UA)} of the program.
This, while much faster than analyzing the full program semantics, is incomplete. 
Therefore, we also perform search over an \emph{expressive family of UAs} (which collectively cover all program behaviors of interest), thereby making our entire approach complete.

Specifically, given $\sqlquery_1, \mydots, \sqlquery_n$ and application condition $\appcond$, we begin with a UA for each $\sqlquery_i$ which encodes a subset $\mathbb{O}_i$ of \emph{reachable} outputs for $\sqlquery_i$ (together with their corresponding inputs). 
Here, an output is reachable, if it  can indeed be returned by the program on an input~\cite{o2019incorrectness}.
We then check if there exist $O_1 \in \mathbb{O}_1, \mydots, O_n \in \mathbb{O}_n$ such that $\appcond[\queryoutput_1 \mapsto O_1, \mydots, \queryoutput_n \mapsto O_n]$ is true.
If so, we can easily solve our problem by deriving the corresponding input given the under-approximate semantics (or UA semantics) of $\sqlquery_1, \mydots, \sqlquery_n$.
Otherwise, we try again but use a different choice of UA. 
This process terminates, when a desired UA is found or no such UA can be found for the given family of UAs.

\vspace{6pt}
\noindent
Below we briefly summarize the key challenges and our solutions. 
Section~\ref{sec:overview} will further illustrate our approach using a concrete example.

\xinyurevision{

\newpara{Under-approximating SQL query semantics.}
The first challenge is how to under-approximate a query $\sqlquery$. 
Our idea is to encode $\sqlquery$'s UA semantics in an SMT formula $\Psi$ whose models correspond to \emph{genuine} input-output pairs of $\sqlquery$.  
In other words, $\Psi$ always encodes reachable outputs.
On the other hand, not all reachable outputs are necessarily encoded by $\Psi$; therefore $\Psi$ is an under-approximation. 
Our first novelty lies in a \emph{compositional} method to build $\Psi$ for $\sqlquery$, by conjoining $\Psi_i$ for each AST node $\astnode_i$ of $\sqlquery$. 
Here, each $\Psi_i$ under-approximates the semantics of the query operator $\queryop$ at $\astnode_i$. 

\newpara{Defining a family of under-approximations.}
This compositional encoding method lends itself well to addressing our second challenge---how to define a family of UAs for a query $\sqlquery$? 
Our idea is to first define a family $\underapproxfamily_{\queryop}$ of UAs for each query operator $\queryop$ in the query language, which collectively covers all inputs of interest to $\queryop$. 
Then, we define a \emph{UA map} $\mapastnodetounderapprox$ for $\sqlquery$, which maps each AST node $\astnode_i$ in $\sqlquery$ to some UA in $\underapproxfamily_{\getqueryop(\astnode_i)}$. 
Each UA map $\mapastnodetounderapprox$ can be encoded into an SMT formula $\Psi \assign \textit{Encode}(\mapastnodetounderapprox)$, by taking the conjunction of the encodings of $\mapastnodetounderapprox$'s entries. 
We can define a family of such UA maps, where each $\Psi$ under-approximates $\sqlquery$ in a different way. 
This approach can be further generalized to under-approximate $n$ queries, by extending the UA map $\mapastnodetounderapprox$ to include \emph{all} AST nodes \emph{across $n$ queries}. 
Conceptually, if $\encodingall \assign \Psi \land \appcond$ is satisfiable, we can derive a satisfying input $\inputdatabase$ from a satisfying assignment of $\encodingall$ that solves our reasoning task for $n$ queries.

\newpara{Conflict-driven under-approximation search.}
Our third research question is: how to efficiently find a \emph{satisfying UA map} $\mapastnodetounderapprox$ (i.e., $\textit{Encode}(\mapastnodetounderapprox) \land \appcond$ is satisfiable)?
We propose a novel search algorithm that explores a sequence of $\mapastnodetounderapprox_i$'s, until reaching a satisfying one. 
Each iteration is fast, and the number of iterations is typically small. 
Our novelty lies in how we generate $\mapastnodetounderapprox_{i+1}$ from an unsatisfiable $\mapastnodetounderapprox_i$. 
In particular, we first obtain a subset of $\mapastnodetounderapprox_i$'s entries for a subset $\conflictASTnodes$ of AST nodes in $\mapastnodetounderapprox_i$---i.e., $\mapastnodetounderapprox_{i} \mapprojection \conflictASTnodes$---such that $\textit{Encode}(\mapastnodetounderapprox_{i} \mapprojection \conflictASTnodes) \land \appcond$ is unsat. 
In other words, $\mapastnodetounderapprox_i \mapprojection \conflictASTnodes$ is a \emph{conflict}. 
Then, we resolve this conflict by mapping some nodes in $\conflictASTnodes$ to new UAs, such that $\mapastnodetounderapprox_{i+1} \mapprojection \conflictASTnodes$ is satisfiable. 
During this process, we might also need to adjust mappings for nodes outside $\conflictASTnodes$, but we do not analyze their semantics. 
Our novelty lies in the development of a lattice structure of UAs and an algorithm that exploits this structure for efficient conflict resolution. 

}

\newpara{Evaluation.}
We have implemented these ideas in a tool called $\tool$\footnote{When you under-approximate circle (rhymes with SQL), you get $\tool$.}, and evaluated it on two applications. 
The first one is SQL query equivalence refutation. Our evaluation result on 24,455 benchmarks reveals that $\tool$ can disprove a large number of query pairs using a median of 0.1 seconds---significantly outperforming all prior techniques.
Our evaluation result on a total of 6,720 query disambiguation benchmarks also shows $\tool$ significantly beats all existing approaches both in terms of the number of benchmarks solved and the solving time.

\newpara{Contributions.}
This paper makes the following contributions. 
\begin{itemize}[leftmargin=*]
\item 
Develop a new symbolic reasoning engine for SQL based on under-approximate reasoning. 
\item 
Formulate the reasoning task as an under-approximation search problem.
\item 
Propose a compositional method to define the search space of under-approximations. 
\item 
Design an efficient conflict-driven algorithm for searching under-approximations. 
\item 
Evaluate an implementation, $\tool$, of these ideas on more than 30,000 benchmarks. 
\end{itemize}

%% file: sections/overview.tex
\section{Overview} \label{sec:overview}

\lstdefinestyle{sqlstyle}{
    language=SQL,
    basicstyle={\scriptsize\ttfamily\linespread{1.1}\selectfont},  
    numbers=none,
    numberstyle=\scriptsize,
    commentstyle=\color{gray},
    numbersep=5pt,
    showstringspaces=false,
    belowskip=3mm,
    breakatwhitespace=true,
    breaklines=true,
    classoffset=0,
    columns=flexible,
    framexleftmargin=0em,
    keywordstyle=\textbf,  
    tabsize=3,
    xleftmargin=2em
}

\forestset{
  ast/.style={
    for tree={
      align=center,
      inner sep=2.2pt,
      s sep+=8pt,
      before computing xy={l=20pt},  
      edge={color=black},  
      line width=.5pt,
      draw=black,
      font=\sffamily\scriptsize,
      node options={align=center},
    }
  }
}


\xinyurevision{
\newpara{Background on SQL equivalence checking.}
In this section, we will present a simple example to illustrate how our approach works for checking the equivalence of two SQL queries. 
This is an important problem with applications in various downstream tasks. 
One such task is automated grading  of SQL queries (consider LeetCode\footnote{\url{https://leetcode.com/} is an online platform that provides coding problems in different languages (including SQL).}): a user-submitted query needs to be checked against a predefined ``ground-truth'' query; in case of non-equivalence, provide a counterexample to users. 
Another task is to validate query rewriting, where a slow query $\sqlquery_1$ is transformed to a faster query $\sqlquery_2$ using rewrite rules. We can perform translation validation by checking $\sqlquery_1$ is equivalent to $\sqlquery_2$. A counterexample in this case can help developers fix the incorrect rewrite rule. 
We also refer readers to recent work on SQL equivalence checking~\cite{he2024verieql,chu2017cosette} for more details. 
}

\newpara{Equivalence refutation example.}
Consider a database with the following three relations. 
\[
\begin{array}{rcl}
\textit{Customers} & : & \text{[ customer\_id, customer\_name, email ]} \\ 
\textit{Contacts} & : & \text{[ user\_id, contact\_name, contact\_email ]} \\ 
\textit{Invoice} & : & \text{[ invoice\_id, price, user\_id ]} \\ 
\end{array}
\]
Here, \emph{Customers} table contains customer information, \emph{Contacts} stores contact information for each customer, and \emph{Invoice} tracks price and customer information for invoices.

Let us consider the following simplified task from a LeetCode problem\footnote{\url{https://leetcode.com/problems/number-of-trusted-contacts-of-a-customer/description/}}: for each invoice, find its corresponding customer's name and the number of contacts for this customer. 
Figures~\ref{fig:example:correct-query} and~\ref{fig:example:wrong-query} show two SQL queries (which are simplified from real-life queries submitted by LeetCode users), where $\sqlquery$ is a correct solution but $\sqlquery'$ is not. 
$\sqlquery$ first counts the number of contacts for each customer using a subquery (rooted at AST node $\astnode_3$; see Figure~\ref{fig:example:correct-query}), then joins it with the \emph{Invoices} table, and finally selects the desired columns. 
On the other hand,  $\sqlquery'$ first joins \emph{Invoices}, \emph{Customers}, and \emph{Contacts} to find the customer name and the count of contacts for each invoice (see $\astnode'_3$ in Figure~\ref{fig:example:wrong-query}), and then joins the \emph{Customers} table again. 
$\sqlquery$ and $\sqlquery'$ yield different outputs when multiple rows in \emph{Customers} share the same email address (which is possible): in this case, {a customer's contact for an invoice would be counted multiple times in the output of $\sqlquery'$}, leading to incorrect aggregation results.

Figure~\ref{fig:example:cex} shows a counterexample input, witnessing the non-equivalence of $\sqlquery$ and $\sqlquery'$. 
Note that Alice and Bob share the same email address in \emph{Customers}.
The goal of equivalence refutation is to find an input  $\inputdatabase$ (like the one in Figure~\ref{fig:example:cex}) such that $\sqlquery(\inputdatabase) \neq \sqlquery'(\inputdatabase)$.

\begin{figure}[!t]
\centering
\begin{minipage}{0.5\textwidth}
\begin{lstlisting}[mathescape,style=sqlstyle]
$v_1$:  SELECT invoice_id, T.customer_name, cnt    
    FROM Invoices I 
$v_2$:  LEFT JOIN (                 
$L_2$:       SELECT A.customer_id, customer_name,
                COUNT(contact_name) AS cnt
        FROM Customers A
$v_4$:             LEFT JOIN Contacts B           
$\phi_2$:             ON A.customer_id = B.user_id
$v_3$:       GROUP BY A.customer_id, customer_name) T
$\phi_1$:    ON I.user_id = T.customer_id
\end{lstlisting}
\end{minipage}
\hspace{5pt}
\begin{minipage}{0.4\textwidth}
\centering
\begin{forest}
ast
[$\astnode_1: \proj_{\attrlist_1}$
    [$\astnode_2: \ljoin_{\predicate_1}$
        [$\astnode_5$: \texttt{Invoices}]
        [$\astnode_3: \groupby_{\vec{\expression}, {\attrlist_2}}$
            [$\astnode_4: \ljoin_{\predicate_2}$
                [$\astnode_6$: \texttt{Customers}]
                [$\astnode_7$: \texttt{Contacts}]
            ]
        ]
    ]
]
\end{forest}
\end{minipage}
\vspace{-10pt}
\caption{SQL query $\sqlquery$ (left) and its corresponding AST (right). While $\sqlquery$ is written in standard SQL syntax, we express its AST following our grammar in Figure~\ref{fig:sql-syntax}. In particular, an AST node is labeled with a query operator, with non-query parameters (such as column lists and predicates) being part of the label. $\proj$ in the AST means \texttt{\textbf{SELECT}}, $\ljoin$ is \texttt{\textbf{LEFT\;JOIN}}, $\attrlist_1$ corresponds to ``\texttt{invoice\_id,\,T.customer\_name,\,cnt}'' on the left, etc. Each AST node is annotated with a unique id $\astnode_i$. We also annotate some parts of $\sqlquery$ to illustrate this mapping.}
\label{fig:example:correct-query}
\vspace{-6pt}
\end{figure}
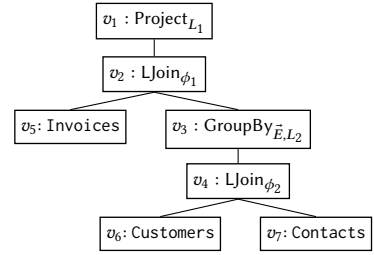

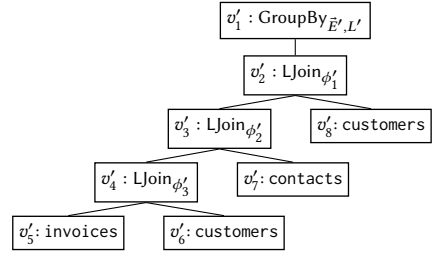
\begin{figure}[!t]
\centering
\begin{minipage}{0.44\textwidth}
\begin{lstlisting}[mathescape,style=sqlstyle]
$L'$:  SELECT invoice_id, c.customer_name,
            COUNT(t.contact_name) cnt
    FROM invoices i
$v'_4$: LEFT JOIN customers c                   
$\phi'_3$: ON i.user_id = c.customer_id
$v'_3$: LEFT JOIN contacts t             
$\phi'_2$: ON c.customer_id = t.user_id
$v'_2$: LEFT JOIN customers c1                
$\phi'_1$: ON t.contact_email = c1.email
$v'_1$: GROUP BY invoice_id, c.customer_name  
\end{lstlisting}
\end{minipage}
\hspace{6pt}
\begin{minipage}{0.4\textwidth}
\centering
\begin{forest}
ast
[$\astnode'_1: \groupby_{\vec{\expression}', \attrlist'}$
    [$\astnode'_2: \ljoin_{\predicate'_1}$
        [$\astnode'_3: \ljoin_{\predicate'_2}$
            [$\astnode'_4: \ljoin_{\predicate'_3}$
                [$\astnode'_5$: \texttt{invoices}]
                [$\astnode'_6$: \texttt{customers}]
            ]
            [$\astnode'_7$: \texttt{contacts}]
        ]
        [$\astnode'_8$: \texttt{customers}]
    ]
]
\end{forest}
\end{minipage}
\vspace{-10pt}
\caption{SQL query $\sqlquery'$ (left) and its corresponding AST (right), following the same protocol as in Figure~\ref{fig:example:correct-query}.}
\label{fig:example:wrong-query}
\end{figure}

\newpara{Prior work.}
{Despite the abundance of work on SQL equivalence checking, 
existing techniques---such as
$\evosql$~\cite{castelein2018search},
$\datafiller$~\cite{datafiller-website}, 
$\xdata$~\cite{chandra2015data,chandra2019automated}, 
$\cosette$~\cite{cosette-website,wang2018speeding},  
and 
$\qex$~\cite{qex,wang2018speeding}---all fail to refute the equivalence of $\sqlquery$ and $\sqlquery'$.}
$\verieql$~\cite{he2024verieql} is the only tool that succeeds, though taking {132 seconds} to solve our (original, unsimplified) example. 
The reason it takes this long is because $\verieql$ \emph{precisely} encodes the semantics of $\sqlquery$ and $\sqlquery'$ for \emph{all inputs}. 
In particular, it first considers all input tables with at most 1 tuple: under this bound, $\sqlquery$ and $\sqlquery'$ are equivalent. Then, it bumps up the bound  to 2 tuples per input table. This ends up creating a complex formula (due to various complex features, including nested joins and group-by, among others not shown in our simplified example). 
It takes a state-of-the-art SMT solver---in particular, z3~\cite{z3-tacas08}---{121 seconds} to solve.

\begin{figure}[!t]
    \centering
    \scriptsize
    \begin{minipage}[b][][b]{0.38\textwidth}
        \centering
        \begin{tabular}{|c|c|c|}
            \hline
            customer\_id & customer\_name & email \\ \hline
            1 & Alice   & a@g.com \\ \hline
            2 & Bob  & a@g.com \\ \hline
        \end{tabular}
\captionsetup{skip=3pt} 
        \caption*{\emph{Customers}}
    \end{minipage} 
    \begin{minipage}[b][][b]{0.32\textwidth}
        \centering
        \begin{tabular}{|c|c|c|}
            \hline
            user\_id & contact\_name & contact\_email \\ \hline
            2 & A & a@g.com \\ \hline
        \end{tabular}
\captionsetup{skip=3pt} 
        \caption*{\emph{Contacts}}
    \end{minipage} \hfill
    \begin{minipage}[b][][b]{0.29\textwidth}
        \centering
        \begin{tabular}{|c|c|c|}
            \hline
            invoice\_id & user\_id & price\\ \hline
            3 & 2  & 10 \\ \hline
        \end{tabular}
\captionsetup{skip=3pt} 
        \caption*{\emph{Invoices}}
    \end{minipage}
\vspace{-15pt}
\caption{A counterexample input database on which $\sqlquery$ and $\sqlquery'$ return different outputs.}
\label{fig:example:cex}
\vspace{-4pt}
\end{figure}

\xinyurevision{
\newpara{Insight 1: under-approximating queries.}
Our key idea is to reason about an under-approximation (UA), which is represented as a UA map $\mapastnodetounderapprox$. 
In particular, given $n$ queries represented as ASTs, $\mapastnodetounderapprox$ maps each AST node $\astnode_i$ to a UA that under-approximates the query operator at $\astnode_i$. 
For instance, consider the $\mapastnodetounderapprox$ in Figure~\ref{fig:ex:M}, which under-approximates $\sqlquery$ and $\sqlquery'$ from Figures~\ref{fig:example:correct-query} and~\ref{fig:example:wrong-query}.
}

\begin{figure}[!t]
\vspace{-5pt}
\centering
\[\small 
\mapastnodetounderapprox = 
\left\{
\begin{array}{llll}
\astnode_1 \mapsto [\truevalue, \falsevalue] & 
\astnode_2 \mapsto \big[ [ \truevalue, \falsevalue ], [ \falsevalue, \falsevalue ] \big] & 
\astnode_3 \mapsto [\truevalue_{\truepredicate}, \truevalue_{\truepredicate}] & 
\astnode_4 \mapsto \big[ [ \falsevalue, \truevalue ], [ \falsevalue, \falsevalue ] \big] \\[5pt]
\astnode'_1 \mapsto [ \truevalue_{\truepredicate}, \falsevalue ] & 
\astnode'_2 \mapsto \big[ [ \truevalue, \truevalue ], [ \falsevalue, \falsevalue ] \big] & 
\astnode'_3 \mapsto \big[ [\falsevalue, \truevalue ], [ \falsevalue, \falsevalue ] \big] & 
\astnode'_4 \mapsto \big[ [ \falsevalue, \truevalue ], [ \falsevalue, \falsevalue ] \big] 
\end{array}
\right\}
\]
\vspace{-10pt}
\caption{An example UA map $\mapastnodetounderapprox$, which can be used to generate the counterexample input in Figure~\ref{fig:example:cex} to refute the equivalence of $\sqlquery$ (in Figure~\ref{fig:example:correct-query}) and $\sqlquery'$ (in Figure~\ref{fig:example:wrong-query}).}
\label{fig:ex:M}
\vspace{-5pt}
\end{figure}

\xinyurevision{
Let us explain some of the entries in $\mapastnodetounderapprox$. 
At a high level, $\mapastnodetounderapprox$ maps each AST node to a so-called ``UA choice'' (or, simply UA) that is represented by an array (potentially more than one-dimensional). 
For example, $\astnode_4$'s UA $\underapprox_4$---denoted by a $2 \times 2$ matrix $\big[ [ \falsevalue, \truevalue ], [ \falsevalue, \falsevalue ] \big]$---considers \emph{Customers} tables with up to 2 tuples and \emph{Contacts} tables with up to 2 tuples (recall from Figure~\ref{fig:example:correct-query} that $\astnode_4$ is a $\ljoin$ operator). 
Importantly, the $\truevalue$ value at position (1,2) constrains that only the first tuple in \emph{Customers} and the second tuple in \emph{Contacts} meet the join predicate $\predicate_2$: this additional constraint allows us to focus on a small set of input-output behaviors.
The UAs for $\astnode_2, \astnode_4, \astnode'_2, \astnode'_3$ are defined in the same way (see Example~\ref{ex:ua-map} with a more detailed explanation), and the UAs for $\astnode_1$ and $\astnode'_1$ follow a similar rationale. 
$\mapastnodetounderapprox$ can be encoded into $\queryopencoding \assign \queryopencoding_1 \land \mydots \land \queryopencoding_4 \land \queryopencoding'_1 \land \mydots \land \queryopencoding'_4$, where each $\queryopencoding_i$ (resp. $\queryopencoding'_i$) encodes the input-output behaviors for $\astnode_i$ (resp. $\astnode'_i$). 
As a whole, $\queryopencoding$ encodes a subset of behaviors of $\sqlquery$ and $\sqlquery'$. 
While we do not show any actual SMT formulas in this example, Section~\ref{sec:alg:ua-semantics} will describe how to build such formulas in detail. 

Given this $\queryopencoding$ and application condition $\appcond$ (which encodes ``the outputs of $\sqlquery$ and $\sqlquery'$ are distinct''), we obtain $\encodingall \assign \queryopencoding \land \appcond$, which in this example is satisfiable. 
A satisfying assignment of $\encodingall$ corresponds to a counterexample input that witnesses the non-equivalence of $\sqlquery$ and $\sqlquery'$ (e.g., $\inputdatabase$ from Figure~\ref{fig:example:cex}). 
Notably, checking the satisfiability of such $\encodingall$ is typically cheap. 
For instance, z3 gives a model in {0.02 seconds}, for the corresponding $\encodingall$ in our original (unsimplified) example.
}

\xinyurevision{
\newpara{Insight 2: under-approximation search for completeness.}
As mentioned earlier, reasoning over a fixed UA may miss some behaviors of interest (i.e., $\encodingall$ may be unsat) and therefore is incomplete. 
Our second insight is to construct a family of UA maps and search over this space for one that is satisfiable. 
This search space is defined compositionally by defining a family of UAs for each query operator. 
Let us take $\astnode_4$ from  the above Figure~\ref{fig:ex:M} as an example: its family contains $2^{4} = 16$ UAs (i.e., each of the four elements can be either $\truevalue$ or $\falsevalue$), which collectively cover all inputs of interest (i.e., input tables with up to two tuples).
Each of these 16 UAs can be encoded in an SMT formula whose models correspond to genuine input-output behaviors for the $\ljoin$ operator at $\astnode_4$. 
We define UA families for the other AST nodes similarly. 
These operator-level UAs induce the search space of UA maps for queries.
The search problem then is how to find a \emph{satisfying UA map} $\mapastnodetounderapprox$ (mapping all AST nodes in $\sqlquery$ and $\sqlquery'$ to UAs), such that $\encodingall \assign \textit{Encode}(\mapastnodetounderapprox) \land \appcond$ is satisfiable.

Our search algorithm explores a sequence of UA maps $\mapastnodetounderapprox_i$'s until reaching a satisfying one. 
In general, $\mapastnodetounderapprox_i$ is partial (i.e., containing a subset of AST nodes). Given $\mapastnodetounderapprox_i$ at each step, we produce the next $\mapastnodetounderapprox_{i+1}$ by either adding more nodes or by adjusting UAs of existing nodes. 
Figure~\ref{fig:ex:M-sequence} shows one such sequence $\mapastnodetounderapprox_1 \rightarrow \mapastnodetounderapprox_2 \rightarrow \mapastnodetounderapprox_3 \rightarrow \mapastnodetounderapprox$, where $\mapastnodetounderapprox$ is the satisfying UA map from Figure~\ref{fig:ex:M}.
}

\begin{figure}[!t]
\centering
\[\small 
\begin{array}{ll}

\mapastnodetounderapprox_1 = 
\left\{
\begin{array}{llll}
\astnode_1 \mapsto [ \unknownvalue, \unknownvalue ] 
\\[3pt]
\astnode'_1 \mapsto [ \unknownvalue, \unknownvalue ] 
\end{array}
\right\}

\hspace{10pt}

\mapastnodetounderapprox_2 =
\left\{
\begin{array}{llll}
\astnode_1 \mapsto [\truevalue, \falsevalue] & 
\astnode_2 \mapsto \big[ [ \unknownvalue, \unknownvalue ], [ \unknownvalue, \unknownvalue ] \big] 

\\[3pt] 

\astnode'_1 \mapsto [ \falsevalue, \falsevalue] & 
\astnode'_2 \mapsto \big[ [ \unknownvalue, \unknownvalue ], [ \unknownvalue, \unknownvalue ] \big] 

\end{array}
\right\}

\\ \\ 

\mapastnodetounderapprox_3 = 
\left\{
\begin{array}{llll}
\astnode_1 \mapsto [\truevalue, \falsevalue] & 
\astnode_2 \mapsto \big[ [ \falsevalue, \falsevalue ], [ \falsevalue, \truevalue ] \big] & 
\astnode_3 \mapsto [ \unknownvalue, \unknownvalue ] & 
\astnode_4 \mapsto \big[ [ \unknownvalue, \unknownvalue ], [ \unknownvalue, \unknownvalue ] \big] \\[5pt]
\astnode'_1 \mapsto [ \falsevalue, \falsevalue ] & 
\astnode'_2 \mapsto \big[ [ \falsevalue, \falsevalue ], [ \falsevalue, \falsevalue ] \big] & 
\astnode'_3 \mapsto \big[ [\unknownvalue, \unknownvalue ], [ \unknownvalue, \unknownvalue ] \big] & 
\astnode'_4 \mapsto \big[ [ \unknownvalue, \unknownvalue ], [ \unknownvalue, \unknownvalue ] \big] 
\end{array}
\right\}

\end{array}
\]
\vspace{-5pt}
\caption{An example sequence of UA maps  $\mapastnodetounderapprox_1 \rightarrow \mapastnodetounderapprox_2 \rightarrow \mapastnodetounderapprox_3 \rightarrow \mapastnodetounderapprox$  that our algorithm may explore, to refute the equivalence of $\sqlquery$ and $\sqlquery'$.  Here, $\mapastnodetounderapprox$ is the satisfying UA map in Figure~\ref{fig:ex:M}.}
\label{fig:ex:M-sequence}
\vspace{-5pt}
\end{figure}

\xinyurevision{
While each $\mapastnodetounderapprox_i$ is structurally similar to $\mapastnodetounderapprox$, they include a special ``top'' value $\unknownvalue$, which intuitively means ``I don't know.'' 
For instance, $\astnode_4$'s UA in $\mapastnodetounderapprox_3$ considers all \emph{Customers} and \emph{Contacts} tables both with up to 2 tuples, with no additional constraints. Intuitively, this UA ``subsumes'' the UA for $\astnode_4$ in $\mapastnodetounderapprox$, because the latter considers a specific way of joining the two input tables. 
Allowing $\unknownvalue$ essentially creates a lattice of UAs, which we exploit to perform efficient under-approximation search. 
Let us explain how this works in detail, still using Figure~\ref{fig:ex:M-sequence} as an example.

The algorithm begins with $\mapastnodetounderapprox_1$, obtains its encoding $\queryopencoding_1$, and confirms $\encodingall_1 \assign \queryopencoding_1 \land \appcond$ is satisfiable. 
In this case, we map $\astnode_1$ and $\astnode'_1$ to UAs derived from a model of $\encodingall_1$, and add 2 new nodes $\astnode_2, \astnode'_2$ (both mapped to $\unknownvalue$)---this yields $\mapastnodetounderapprox_2$. 
$\encodingall_2$ for $\mapastnodetounderapprox_2$ is also sat; therefore, we update $\mapastnodetounderapprox_2$ in the same way (except adding 4 nodes this time) and obtain $\mapastnodetounderapprox_3$.

$\encodingall_3$ for $\mapastnodetounderapprox_3$, however, is unsat. 
In this case, we first obtain a subset $\conflictASTnodes$ (namely, $\{ \astnode'_1, \astnode'_2, \astnode'_3, \astnode'_4, \astnode_1, \astnode_2 \}$) of nodes, whose UAs in $\mapastnodetounderapprox_3$ are in \emph{conflict}; that is, $\textit{Encode}(\mapastnodetounderapprox_3 \mapprojection \conflictASTnodes) \land \appcond$ is unsat ($\mapastnodetounderapprox_3 \mapprojection \conflictASTnodes$ gives the ``sub-map'' of $\mapastnodetounderapprox_3$ for nodes in $\conflictASTnodes$).
Then, we aim to build a new UA map $\mapastnodetounderapprox_4$ which (i) maps nodes in $\conflictASTnodes$ to new UAs and (ii) maps nodes outside $\conflictASTnodes$ to UAs, such that (a)  $\textit{Encode}(\mapastnodetounderapprox_4 \mapprojection \conflictASTnodes) \land \appcond$ is sat and (b) $\mapastnodetounderapprox_4$ does not contain any previously discovered conflicts. 
In this example, we have $\mapastnodetounderapprox_4 = \mapastnodetounderapprox$, which is a satisfying UA map. 
Note that step (i) requires searching, among all combinations of UAs for nodes in $\conflictASTnodes$, for one that meets (a). 
To do this efficiently, we take advantage of the aforementioned lattice structure of UAs to partition this large space into subspaces shown below in Figure~\ref{fig:ex:partition}, which enables efficient search over this combinatorial space.

}

\begin{figure}[!h]
\vspace{-10pt}
\[\footnotesize
\arraycolsep=1pt\def\arraystretch{1.3}
\begin{array}{rl}
(1) & 
\astnode_1 \mapsto [\unknownvalue, \unknownvalue], 
\astnode'_1 \mapsto [\unknownvalue, \unknownvalue], 
\astnode_2 \mapsto \big[ [\unknownvalue, \unknownvalue], [\unknownvalue, \unknownvalue] \big], 
\astnode'_2 \mapsto \big[ [\unknownvalue, \unknownvalue], [\unknownvalue, \unknownvalue] \big], 
\astnode'_3 \mapsto \big[ [\truevalue, \truevalue], [\unknownvalue, \unknownvalue] \big], 
\astnode'_4 \mapsto \big[ [\truevalue, \truevalue], [\unknownvalue, \unknownvalue] \big]
\\ 
(2) & 
\astnode_1 \mapsto [\unknownvalue, \unknownvalue], 
\astnode'_1 \mapsto [\unknownvalue, \unknownvalue], 
\astnode_2 \mapsto \big[ [\unknownvalue, \unknownvalue], [\unknownvalue, \unknownvalue] \big], 
\astnode'_2 \mapsto \big[ [\unknownvalue, \unknownvalue], [\unknownvalue, \unknownvalue] \big], 
\astnode'_3 \mapsto \big[ [\falsevalue, \truevalue], [\unknownvalue, \unknownvalue] \big], 
\astnode'_4 \mapsto \big[ [\truevalue, \truevalue], [\unknownvalue, \unknownvalue] \big]
\\
(3) & 
\astnode_1 \mapsto [\unknownvalue, \unknownvalue], 
\astnode'_1 \mapsto [\unknownvalue, \unknownvalue], 
\astnode_2 \mapsto \big[ [\unknownvalue, \unknownvalue], [\unknownvalue, \unknownvalue] \big], 
\astnode'_2 \mapsto \big[ [\unknownvalue, \unknownvalue], [\unknownvalue, \unknownvalue] \big], 
\astnode'_3 \mapsto \big[ [\truevalue, \truevalue], [\unknownvalue, \unknownvalue] \big], 
\astnode'_4 \mapsto \big[ [\falsevalue, \truevalue], [\unknownvalue, \unknownvalue] \big]
\\ 
(4) & 
\astnode_1 \mapsto [\unknownvalue, \unknownvalue], 
\astnode'_1 \mapsto [\unknownvalue, \unknownvalue], 
\astnode_2 \mapsto \big[ [\unknownvalue, \unknownvalue], [\unknownvalue, \unknownvalue] \big], 
\astnode'_2 \mapsto \big[ [\unknownvalue, \unknownvalue], [\unknownvalue, \unknownvalue] \big], 
\astnode'_3 \mapsto \big[ [\falsevalue, \truevalue], [\unknownvalue, \unknownvalue] \big], 
\astnode'_4 \mapsto \big[ [\falsevalue, \truevalue], [\unknownvalue, \unknownvalue] \big]
\\ 
& \mydots 
\\ 
(16) & 
\astnode_1 \mapsto [\unknownvalue, \unknownvalue], 
\astnode'_1 \mapsto [\unknownvalue, \unknownvalue], 
\astnode_2 \mapsto \big[ [\unknownvalue, \unknownvalue], [\unknownvalue, \unknownvalue] \big], 
\astnode'_2 \mapsto \big[ [\unknownvalue, \unknownvalue], [\unknownvalue, \unknownvalue] \big], 
\astnode'_3 \mapsto \big[ [\falsevalue, \falsevalue], [\unknownvalue, \unknownvalue] \big], 
\astnode'_4 \mapsto \big[ [\falsevalue, \falsevalue], [\unknownvalue, \unknownvalue] \big]
\end{array}
\]
\vspace{-10pt}
\caption{An example partition of the UA space for AST nodes $\astnode'_1, \astnode'_2, \astnode'_3, \astnode'_4, \astnode_1, \astnode_2$ from $\sqlquery$ and $\sqlquery'$.}
\label{fig:ex:conflict}
\vspace{-5pt}
\label{fig:ex:partition}
\end{figure}

\xinyurevision{
Note that each subspace is a UA map over nodes in $\conflictASTnodes$. Our algorithm processes them one by one. 
It first encodes the UA map $\mapastnodetounderapprox^{(1)}_{\conflictASTnodes}$ in (1) from Figure~\ref{fig:ex:partition}---that is, $\textit{Encode}\big(\mapastnodetounderapprox^{(1)}_{\conflictASTnodes}\big) \land \appcond$---and finds it to be unsat. In other words, this $\mapastnodetounderapprox^{(1)}_{\conflictASTnodes}$ is still a conflict. 
{Similarly, $\mapastnodetounderapprox^{(2)}_{\conflictASTnodes}$ and $\mapastnodetounderapprox^{(3)}_{\conflictASTnodes}$ are also unsat.}
On the other hand, $\mapastnodetounderapprox^{(4)}_{\conflictASTnodes}$ is sat. We therefore obtain a model for $\mapastnodetounderapprox^{(4)}_{\conflictASTnodes}$, from which we can generate $\mapastnodetounderapprox_4$, which is identical to $\mapastnodetounderapprox$. 
This concludes the search process. 
{In total, it takes 0.3 seconds for $\tool$ to solve the original (unsimplified) example. }
}

%% file: sections/alg.tex
\section{Conflict-Driven Under-Approximation Search for SQL}
\label{sec:alg}

We begin with the syntax of our query language (Section~\ref{sec:alg:language}). 
Then, we define under-approximations for this language and formalize its under-approximate semantics (Sections~\ref{sec:alg:unders}-\ref{sec:alg:ua-semantics}). 
Finally,  we give our under-approximation search algorithm (Sections~\ref{sec:alg:top-level}-\ref{sec:alg:discussion}).

\subsection{Query Language}
\label{sec:alg:language}

\newpara{Syntax.}
Figure~\ref{fig:sql-syntax} shows the syntax of our language, which is a highly expressive subset of SQL. 
A query $\query$ can be a relation $\relation$, a projection $\proj_\attrlist(\query)$ that selects columns $\attrlist$ from $\query$'s result, a filter $\filter_{\predicate}(\query)$ that retains from $\query$ the rows satisfying predicate $\predicate$, a renaming operator $\rename_\relation(\query)$ that renames the output of $\query$ to relation $\relation$.
It also supports bag union $\unionall$ and different types of joins (including Cartesian product $\product$, inner join $\ijoin_{\predicate}$, left join $\ljoin_{\predicate}$, right join $\rjoin_{\predicate}$, and full outer join $\fjoin_{\predicate}$).
$\text{Distinct}(\query)$ removes duplicate rows. 
$\text{GroupBy}_{\vec{\expression}, \attrlist, \predicate}(\query)$ first groups all rows from $\query$ based on $\vec{\expression}$, computes expressions $\attrlist$ for each group, and returns the result for groups satisfying $\predicate$.
$\text{OrderBy}_{\expression}(\query)$ sorts (in ascending order) $\query$'s output by expression list $\expression$. 
{$\text{With}(\vec{\query}, \vec{\relation}, \query)$ creates intermediate relations $\vec{\relation}$ with the result of queries $\vec{\query}$, and returns the result of $\query$ (which potentially uses $\vec{\relation}$).}
It is worth noting that we support if-then-else $\text{ITE}(\predicate, \expression, \expression)$, case-when $\text{Case}(\vec{\predicate}, \vec{\expression}, \expression)$, and predicates like $\vec{\expression} \in \query$ and $\text{Exists}(\query)$.

\input{figures/fig-sql-syntax}

\newpara{Semantics.}
Our semantics is based on prior work~\cite{he2024verieql}, which will be formalized in Section~\ref{sec:alg:full-semantics}.

\subsection{Representing Under-Approximations}
\label{sec:alg:unders}

This section presents a method to define a family of under-approximations for  each \emph{Query} operator $\queryop$ from Figure~\ref{fig:sql-syntax}. 
In a nutshell, our under-approximation (UA) describes some subset of $\queryop$'s \emph{reachable} outputs (i.e., outputs that can indeed be produced by $\queryop$ for some input).
This notion is consistent with the under-approximation idea from O'Hearn's seminal paper on incorrectness logic~\cite{o2019incorrectness}. 
Our work, however, realizes UAs using a two-step approach:
(i) first describe a subset of inputs for $\queryop$, 
(ii) then define their corresponding input-output behaviors (per $\queryop$'s semantics). 
Here, (i) is simply a ``pointer'' that refers to a subset of inputs, whereas (ii) is the actual encoding of the UA. 
This section focuses on step (i), while Sections~\ref{sec:alg:full-semantics}-\ref{sec:alg:ua-semantics} will explain step (ii).

In this work, we call a subset of inputs in step (i) ``a UA choice''---or, with a slight abuse of notation, simply ``a UA''---which is denoted by   $\underapprox$. 
More specifically, $\underapprox$ for each $\queryop$ is always represented by an array (potentially more than one-dimensional) of values. 
The interpretation of $\underapprox$ is specific to $\queryop$, which we will explain below. As mentioned earlier, $\underapprox$ is just used to refer to a subset of inputs.

\newpara{Filter.}
The UA $\underapprox$ for $\filter_{\predicate}$ is always of the form $\mylist{\truevalue, \falsevalue, \unknown}^n$. That is, $\underapprox$ is always a vector of $n$ values, where each value is $\truevalue$, $\falsevalue$ or $\unknown$.
Here, $n$ is the maximum input table size; that is, we consider input tables $\relation$ with at most $n$ tuples. 
If the $i$th value $\underapprox_i$ is $\truevalue$, it means the $i$th tuple $\tuple_i$ in $\relation$ satisfies predicate $\predicate$. 
On the other hand, $\underapprox_i = \falsevalue$ means either $\tuple_i$ is not present\footnote{The formal encoding of a table will be described in Section~\ref{sec:alg:full-semantics}. In brief, while a table always has exactly $n$ tuples, we allow some of the tuples to be deleted (i.e., not present). This allows us to encode all tables with \emph{at most} (not just exactly) $n$ tuples.} in $\relation$, or $\tuple_i$ does not satisfy $\predicate$. 
Finally, $\unknownvalue$ is the ``top'' value (i.e., ``I don't know''), meaning $\tuple_i$ can be either $\truevalue$ or $\falsevalue$. 
As we can see, any table (of size up to $n$) can be covered by some $\underapprox$, no matter how its tuples satisfy the filtering condition $\predicate$.

\begin{example}
Let us consider $\filter_{\texttt{id} > 3}$, and a UA choice $\underapprox = [\truevalue, \falsevalue]$ for it. 
Here, $\underapprox$ refers to those (input) relations $\relation$ whose \emph{first} tuple has \texttt{id} greater than 3---this is what the first value $\underapprox_1 = \truevalue$ in $\underapprox$ means. 
The second value $\underapprox_2 = \falsevalue$ means that the second tuple in $\relation$ is either not present (meaning $\relation$ has exactly one row), or its \texttt{id} is \emph{not} greater than 3 (in this case, $\relation$ has exactly two rows). Note that $\relation$ has at most 2 tuples, since $\underapprox$ has length 2. 
As another example, let us consider $\underapprox' = [\truevalue, \unknownvalue]$ for the same filter operator. 
Different from $\underapprox_2$, $\underapprox'_2$ puts no restrictions on the second tuple. So $\underapprox'$ refers to all tables with up to two tuples, whose first tuple must satisfy predicate $\texttt{id} > 3$. In other words, $\underapprox'$ ``subsumes'' $\underapprox$, or $\underapprox$ ``refines'' $\underapprox'$.
\label{ex:filter-ua}
\end{example}

\newpara{Projection and UnionAll.}
The UAs for projection are defined in the same way as filter, but their interpretation is slightly different. 
In particular, $\underapprox_i$ being $\truevalue$ or $\falsevalue$ indicates whether or not $\tuple_i$ from the input table is present. 
On the other hand, $\unknownvalue$ is still the top value that means either $\truevalue$ or $\falsevalue$. 
The UA for $\unionall$ is essentially the same as for projection, except that it is a vector of $n_1 + n_2$ values, where the first $n_1$ (resp. the last $n_2$) values correspond to tuples from the first (resp. second) table.

\newpara{Joins.}
Now, let us consider the join operators. Take the inner join $\ijoin_{\predicate}$ as an example. Its UA $\underapprox$ is an $n_1 \times n_2$  matrix, where $\underapprox_{i, j}$ is $\truevalue, \falsevalue$ or $\unknownvalue$. 
Here, $n_1$ (resp. $n_2$) is the maximum size of the first (resp. second) table.
For each $\underapprox_{i, j}$, $\truevalue$  means the $i$th tuple $\tuple_i$ from the first table and the $j$th tuple $\tuple_j$ from the second are both present and satisfy the join condition $\predicate$; whereas $\falsevalue$ means at least one of $\tuple_i$, $\tuple_j$ is deleted, or they are both present but cannot be joined. 
The UA definitions for the other joins are quite similar; 
we refer readers to Appendix~\ref{sec:ua-others} for more details.

\newpara{GroupBy and Distinct.}
The UA definition for $\groupby_{\vec{\expression}, \attrlist, \predicate}$ is different from all operators above. 
{Recall that $\groupby$ first groups all input tuples based on $\vec{\expression}$, then evaluates $\attrlist$ for each group, and finally returns an output table with groups  satisfying predicate $\predicate$.}
A UA choice $\underapprox$ for $\groupby$ is always of the form $\mylist{\truevalue, \truevalue_{\predicate}, \falsevalue, \unknown}^n$. 
Similar to before, $\underapprox_i$ corresponds to the $i$th tuple $\tuple_i$ in the input. 
The interpretation, however, is different. 
{If $\underapprox_i = \truevalue$ or $\truevalue_{\predicate}$, it means $\tuple_i$ is present and distinct from all tuples $\tuple_j (j < i)$ before it, with respect to columns in $\vec{\expression}$.}
In other words, $\tuple_i$ will form a new group. 
The difference between $\truevalue$ and $\truevalue_{\predicate}$ is that $\truevalue_{\predicate}$ means this new group further satisfies $\predicate$, whereas $\truevalue$ means it does not. 
On the other hand, $\falsevalue$ means $\tuple_i$ will not lead to a new group (either it is deleted, or it belongs to an existing group). 
Finally, $\unknownvalue$ is still our top value, which means $\truevalue, \truevalue_{\predicate}$, or $\falsevalue$.
The UA $\underapprox$ for $\distinct$ is a vector of $n$ values chosen from $\{ \truevalue, \falsevalue, \unknownvalue \}$, where $\underapprox_i = \truevalue$ means $\tuple_i$ is distinct from all prior tuples $\tuple_j (j < i)$, and $\falsevalue$ means there exists some $j < i$ such that $\tuple_i = \tuple_j$ (i.e., $\tuple_i$ is a duplicate).

\begin{example}
Consider $\groupby_{\vec{\expression}, \attrlist_2}$ at node $\astnode_3$ from Figure~\ref{fig:example:correct-query}(b), where $\attrlist_2$ is shown in Figure~\ref{fig:example:correct-query}(a) and $\vec{\expression} = [{ \texttt{A.customer\_id}, \texttt{customer\_name}}]$. Note that the having predicate $\predicate$ here is $\truepredicate$ (so our notation omitted it).
Consider UA $\underapprox = [\truevalue_{\predicate}, \falsevalue]$, which refers to tables with at most two tuples $\tuple_1, \tuple_2$. It further states: 
(i) $\tuple_1$ exists and forms a new group, and this group satisfies $\predicate$ (which is always true as we have $\predicate = \truepredicate$ in this example); (ii) $\tuple_2$ either does not exist, or $\vec{\expression}$ evaluates on $\tuple_2$ to the same result as a previous tuple (in this example, $\tuple_1$). 
\label{ex:groupby-ua}
\end{example}

\newpara{Lattice of UAs.}
As we can see, UAs for each operator form a (semi-)lattice, due to the top value $\unknownvalue$.  
In particular, given two UAs $\underapprox$ and $\underapprox'$ for the same operator, we define 
$\underapprox \subsumes \underapprox'$ if we have $\underapprox_i \subsumes \underapprox'_i$ for all $i$. 
That is, $\underapprox$ subsumes $\underapprox'$ (or, $\underapprox'$ refines $\underapprox)$ if the $i$th value in $\underapprox$ subsumes the $i$th value in $\underapprox'$. 
The value subsumption relation is straightforward: the top value $\unknownvalue$ subsumes any non-top value, and non-top values do not subsume each other. 
A UA whose values are all $\unknownvalue$ is called a ``top UA'', while we call a UA with only non-top values a ``minimal UA''.
We use $\underapproxfamily_{\queryop}$ to denote the set of UAs for $\queryop$. 

\newpara{Other query operators.}
We refer readers to Appendix~\ref{sec:ua-others} for $\orderby$'s UA definition.
For operators $\relation$, $\rename$, and $\with$, we have one UA that always considers \emph{all} inputs up to a given size (i.e., effectively no under-approximation).

\subsection{Encoding Full Semantics}
\label{sec:alg:full-semantics}

This section formalizes the \emph{full} semantics of query operators from Figure~\ref{fig:sql-syntax}, which Section~\ref{sec:alg:ua-semantics} will use to build \emph{under-approximate} semantics. 
Our formalization is based on prior work~\cite{he2024verieql}: 
a table is encoded symbolically using $n$ symbolic tuples, a tuple can be deleted (this allows us to consider all tables \emph{up to} size $n$), and operator semantics is encoded in SMT following SQL standard. 
The \emph{key distinction} from prior work is that we embed the UA choices into the semantics encoding---we will further elaborate on this as we go over the encoding rules below. 
We note that, our ``full semantics'' is still bounded, in that it considers tables up to a finite size bound. We use ``full'' just to differentiate it from under-approximate semantics (which will be presented in Section~\ref{sec:alg:ua-semantics}).

Figure~\ref{fig:full-semantics} gives the semantics encoding rules for a representative subset of query operators. 
Please find the full set of rules in Appendix~\ref{sec:all-encoding}.
Here, $\encodefullsemantics$ generates an SMT formula $\queryopencoding$ that encodes a given operator $\queryop$'s semantics for all input tables up to a certain bound.\footnote{Note that the actual values of $\queryop$'s non-table arguments (such as predicates and expressions) are all given.}
$\queryopencoding$ is always over free variables $\queryopinput$ (denoting an input table) and $\queryopoutput$ (denoting $\queryop$'s output table). Any satisfying assignment of $\queryopencoding$ corresponds to a genuine input-output behavior of $\queryop$; that is, all outputs encoded by $\queryopencoding$ are reachable. 
$\queryopencoding$ has another free variable $\underapproxvar$, called ``UA variable'', which denotes the UA choice---as we will see below, this UA variable $\underapproxvar$ plays a central role in our encoding.

\newpara{Filter.}
Rule (1) encodes $\filter_{\predicate}$'s reachable outputs for all input tables with up to $n$ tuples. 
Each $\tuple_i$ is a symbolic tuple from the input table $\queryopinput$, and $\tuple'_i$ is the corresponding symbolic tuple in the output table $\queryopoutput$. 
We use an uninterpreted predicate $\del$ to denote if a tuple is deleted. Some $\tuple_i$'s may be deleted, so $\queryopinput$ may have fewer than $n$ tuples. For a non-deleted $\tuple_i$ that does not meet the filtering condition $\predicate$, we use $\del(\tuple'_i)$ to mean the corresponding $\tuple'_i$ is deleted in $\queryopoutput$. 

Our key novelty is to embed UA variable $\underapproxvar$ into the semantics encoding. 
For $\filter_{\predicate}$, $\underapproxvar$ is a vector of variables, where each $\underapproxvar_i$ denotes the $i$th value  $\underapprox_i$ in the UA. 
Recall from Section~\ref{sec:alg:unders} that $\underapprox_i = \truevalue$ means $\tuple_i$ is present and satisfies  $\predicate$.
For such $\tuple_i$, $\queryopencoding_{i, \truevalue}$ encodes the corresponding $\tuple'_i$. 
$\queryopencoding_{i, \falsevalue}$ handles the $\falsevalue$ case. 
The final formula $\queryopencoding$ is built \emph{compositionally}, by conjoining $\queryopencoding_{i, \truevalue} \land \queryopencoding_{i, \falsevalue}$ across all $i$'s from $1$ to $n$.
As we will show shortly, all of our semantics encoding rules are compositional. 
To simplify our presentation, Rule (1) uses some auxiliary functions. 
$\cmdExist(\tuple_i, \tuple'_i, \predicate) 
\assign 
\big(
\neg \del(\tuple_i) 
\land \denot{\predicate}_{\tuple_i}
\big)
\land
\neg \del(\tuple'_i)
$
encodes the case where input tuple $\tuple_i$ is present and satisfies $\predicate$, and (therefore) the output tuple $\tuple'_i$ exists. 
In this case, we also have
$
\cmdCopy(\tuple_i, \tuple'_i, \attrlist) 
\assign
\bigland\limits_{a \in \attrlist} 
    \denot{\tuple'_i.a} = \denot{\tuple_i.a}
$,
which creates the content of $\tuple'_i$ by copying values of attributes in $\attrlist$\footnote{We assume attributes are given, to simplify our presentation. In general, they can be easily inferred from the input schema.} from $\tuple_i$ to $\tuple'_i$. 
$\queryopencoding_{i, \falsevalue}$ encodes the situation where $\tuple'_i$ is deleted using 
$
\cmdCondDel(\tuple_i, \tuple'_i, \predicate) 
\assign
\Big(
\del(\tuple_i) 
\lor 
\big(
\neg \del(\tuple_i) \land \neg \denot{\predicate}_{\tuple_i} 
\big) 
\Big)
\land 
\del(\tuple'_i)
$.
Here, the corresponding $\tuple_i$ is either not present or does not meet $\predicate$. 
Finally, we note that $\underapproxvar_i$ only considers non-top values---i.e., $\truevalue$ and $\falsevalue$ for filter---because considering $\unknownvalue$ is not necessary when encoding the full semantics.

\newpara{Projection.}
Rule (2) is almost the same as Rule (1), except that the attribute list $\attrlist$ is given and the ``filtering condition'' is always $\truepredicate$. 
\xinyurevision{Note that aggregate functions are always precisely encoded (i.e., no UAs) as also shown in the rule, although under-approximation happens when encoding the semantics of projection. This is also the case for other operators that use aggregate functions.}

\newpara{Inner join.}
Rule (3) follows the same principle, but is slightly more involved. 
First, each tuple $\tuple\doubleprime_{i, j}$ in $\queryopoutput$ corresponds to the join result of $\tuple_i$ from $\queryopinput_1$ and $\tuple'_j$ from $\queryopinput_2$. 
Variable $\underapproxvar_{i, j}$ denotes UA value $\underapprox_{i, j}$. 
$\queryopencoding_{i, j, \truevalue}$ encodes the input-output behavior when $\underapproxvar_{i, j} = \truevalue$. 
Here, we first assert $\tuple\doubleprime_{i, j}$ exists in the output, and then use $\cmdCopy$ to construct the content in $\tuple\doubleprime_{i, j}$. 
Note that $\cmdExist$ is extended here to accept a pair of input tuples: 
$\cmdExist \big( (\tuple_i, \tuple'_j), \tuple\doubleprime_{i, j}, \predicate \big) 
\assign
\big(
\neg \del(\tuple_i) 
\land 
\neg \del(\tuple'_j)
\land 
\denot{\predicate}_{\tuple_i, \tuple_j}
\big)
\land 
\neg \del(\tuple\doubleprime_{i, j})
$. 
$\queryopencoding_{i, j, \falsevalue}$ encodes the $\falsevalue$ case, with an extended $\cmdCondDel$ function. 
\[
\cmdCondDel \big( ( \tuple_i, \tuple'_j), \tuple\doubleprime_{i, j}, \predicate \big) 
\assign
\Big(
\del(\tuple_i) 
\lor 
\del(\tuple'_j) 
\lor 
\big(
\neg \del(\tuple_i) 
\land 
\neg \del(\tuple'_j) 
\land 
\neg \denot{\predicate}_{\tuple\doubleprime_{i, j}} 
\big) 
\Big)
\land 
\del(\tuple\doubleprime_{i, j})
\]

\input{figures/fig-under-semantics}

\newpara{GroupBy.}
Rule (4) also constructs the final encoding compositionally, by conjoining the formulas for $\truevalue, \truevalue_{\predicate}$ and $\falsevalue$ across all $i$'s from $1$ to $n$. 
Recall from Section~\ref{sec:alg:unders} that $\underapprox_i = \falsevalue$ means $\tuple_i$ does not form a new group. 
Therefore, $\queryopencoding_{i, \falsevalue}$ encodes the constraints for $\tuple_i$ and asserts  $\tuple'_i$ is not present.
There are two possibilities: $\tuple_i$ is deleted, or there exists some $\tuple_j (j < i)$ in the same group as $\tuple_i$. 
Here, $\group$ is an uninterpreted function to record a tuple's belonging group, which will be useful later. 
$\queryopencoding_{i, \truevalue}$ and $\queryopencoding_{i, \truevalue_{\predicate}}$ both consider the case where $\tuple_i$ forms a new group; so they share $\queryopencoding_{i, \neg \falsevalue}$ as a common piece, which simply says $\tuple_i$ is present and there is no prior tuple $\tuple_j$ that belongs to the same group. 
$\queryopencoding_{i, \truevalue}$ further states that those tuples in $\tuple_i$'s group do not satisfy $\predicate$ and therefore we have $\tuple'_i$ deleted. 
On the other hand, $\queryopencoding_{i, \truevalue_{\predicate}}$ encodes the case where $\predicate$ is met and output tuple $\tuple'_i$ is present. 
Here, $\cmdCopy$ creates the content in $\tuple'_i$, based on a set $\vec{\tuple}$ of input tuples belonging to a group (given by $\group^{-1}$). 
\[
\cmdCopy \big( \vec{\tuple} , \tuple'_i, \attrlist \big) 
\assign
\bigwedge\limits_{a \in \attrlist}  \denot{\tuple'_i.a} = \denot{a}_{\vec{\tuple}}
\]

\subsection{Encoding Under-Approximate Semantics}
\label{sec:alg:ua-semantics}

To under-approximate a query, we first under-approximate its operators.

\newpara{Encoding UA semantics for query operators.}
Given a UA choice $\underapprox \in \underapproxfamily_{\queryop}$ for a query operator $\queryop$, we encode the UA semantics of $\queryop$ against $\underapprox$ as follows. 
\[
\encodeuasemantics( \queryop, \underapprox )
\assign
\encodechoice(\underapprox) \land \encodefullsemantics(\queryop)
\]
$\encodefullsemantics$ is described in Figure~\ref{fig:full-semantics}, and $\encodechoice$ is defined below. 
\[
\begin{array}{rll}
\encodechoice(\underapprox) & 
\assign 
& 
\bigland_{\underapprox_i \in \underapprox} \encodeuavalue( \underapprox_i ) 
\\[5pt]
\encodeuavalue(\underapprox_i) & 
\assign 
& 
\begin{cases}
\ \ \underapproxvar_i = \underapprox_i 
& 
\text{if} \ \underapprox_i \neq \unknown 
\\
\ \ \bigvee_{c \in \dom(\underapprox_i) \setminus \set{\unknown}} \underapproxvar_i = c 
& 
\text{if} \  \underapprox_i = \unknown 
\\
\end{cases}
\end{array}
\]
Here, $\dom(\underapprox_i)$ gives all possible values for $\underapprox_i$---including top $\unknownvalue$---but $\encodeuavalue$ removes $\unknownvalue$; hence the final encoding considers only minimal UAs. 
While $\encodefullsemantics$ constructs a fixed formula for $\queryop$ that encodes reachable outputs for \emph{all} inputs (up to a bound), 
$\encodeuasemantics$ only encodes reachable outputs for the \emph{subset} of inputs specified by $\underapprox$. 
This (further) under-approximates $\queryop$, and enables efficient symbolic reasoning.

\newpara{Encoding UA semantics for queries.}
We can further encode the UA semantics for query $\sqlquery$, given a \emph{UA map} $\mapastnodetounderapprox$ that maps each of $\sqlquery$'s AST node $\astnode$ to a UA choice $\underapprox \in \underapproxfamily_{\astnode}$, as follows. 
\[
\begin{array}{ll}
\encode( \mapastnodetounderapprox ) 
\assign 
\bigland\limits_{
\substack{
( \astnode \; \mapsto \; \underapprox ) \in \mapastnodetounderapprox \\ 
\textit{children}( \astnode ) = [ \astnode_1, \mydots, \astnode_l ]
}
}
\Big(
\encode ( \astnode, \underapprox )
\land 
(
\astnodeinput^{\astnode}_1 = \astnodeoutput^{\astnode_1} 
\land 
\mydots 
\land 
\astnodeinput^{\astnode}_l = \astnodeoutput^{\astnode_l} 
)
\Big)_{\labelastnode}
\end{array}
\]
Here, $\astnode$ is an AST node in $\sqlquery$ with $l$ children, and $\underapproxfamily_{\astnode} = \underapproxfamily_{\getqueryop(\astnode)}$.
{$\encode( \astnode, \underapprox )$} is defined as:
\[\footnotesize
\begin{array}{ll}
\encode( \astnode, \underapprox ) 
\\[5pt] 
\assign
\begin{cases}
\Big(
\encodefullsemantics \big( \getqueryop ( \astnode) \big) 
\Big)
\big[ 
\queryopinput_1 \mapsto \astnodeinput^{\relation}_{1},  
\mydots, 
\queryopinput_l \mapsto \astnodeinput^{\relation}_l, 
\queryopoutput \mapsto \astnodeoutput^{\astnode}
\big] 
& 
\hspace{-16pt}
\text{if } \getqueryop(\astnode) = \relation

\\[5pt]

\Big(
\encodefullsemantics \big( \getqueryop ( \astnode) \big) 
\Big)
\big[ 
\queryopinput_1 \mapsto \astnodeinput^{\astnode}_{1}, 
\mydots, 
\queryopinput_l \mapsto \astnodeinput^{\astnode}_l, 
\queryopoutput \mapsto \astnodeoutput^{\astnode}
\big]
& 
\hspace{-16pt}
\text{if }  \getqueryop(\astnode) = \with \text{ or } \rename

\\[5pt]

\Big(
\encodechoice( \underapprox ) 
\land 
\encodefullsemantics \big( \getqueryop ( \astnode) \big) 
\Big)
\big[ 
\queryopinput_1 \mapsto \astnodeinput^{\astnode}_{1}, \mydots, \queryopinput_l \mapsto \astnodeinput^{\astnode}_l, \queryopoutput \mapsto \astnodeoutput^{\astnode}, 
\varUAvector \mapsto \varUAvector^{\astnode}
\big]
& 
\text{otherwise}

\end{cases}
\end{array}
\]
which encodes the UA semantics for query operator $\getqueryop(\astnode)$ at $\astnode$ against $\underapprox$. 
We highlight the ``otherwise'' case, where we rename each variable $\queryopinput_i$ to $\queryopinput_{i}^{\astnode}$ (that denotes the $i$th input to $\astnode$), and similarly $\queryopoutput$ to $\queryopoutput_i^{\astnode}$ (which denotes $\astnode$'s output), as well as $\varUAvector$ to $\underapproxvar^{\astnode}$ (which denotes the UA choice at $\astnode$). 
Operators in the other cases always have one UA, so we do not need to encode it.
For operator $\relation$, we use $\queryopinput^{\relation}$ (instead of $\queryopinput^{\astnode}$), since multiple AST nodes may be labeled with the same relation. 
$\encodeuasemantics(\mapastnodetounderapprox)$ conjoins encodings across all $\astnode \mapsto \underapprox$ entries in $\mapastnodetounderapprox$. Each clause is labeled with $\labelastnode$, which (as we will see in later sections) is used for conflict extraction. 
As a final minor note, $\encodeuasemantics$ assumes access to $\sqlquery$'s schema, which allows $\encodefullsemantics$ to obtain attributes of $\sqlquery$'s intermediate tables.

$\encode(\mapastnodetounderapprox)$ can be generalized to encode UA semantics for multiple queries $\sqlquery_1, \mydots, \sqlquery_n$, by extending $\mapastnodetounderapprox$ to include all $\astnode \mapsto \underapprox$ entries across all $\sqlquery_i$'s.

\begin{example}
Consider the UA map $\mapastnodetounderapprox$ from Figure~\ref{fig:ex:M}.
Let us briefly explain some of its entries. 
Consider $\astnode_2 \mapsto \big[ [ \truevalue, \falsevalue ], [ \falsevalue, \falsevalue ] \big]$, where $\astnode_2$ is a left join operator with two input tables $\queryopinput^{\astnode_2}_1$ (i.e., \emph{Invoices}) and $\queryopinput^{\astnode_2}_2$ (i.e., output of $\astnode_3$). 
$\mapastnodetounderapprox(\astnode_2)$ states: 
$\queryopinput^{\astnode_2}_1$ and $\queryopinput^{\astnode_2}_2$ have up to 2 tuples; 
the first tuple in $\queryopinput^{\astnode_2}_1$ and the first tuple in $\queryopinput^{\astnode_2}_2$ meets the join condition; the other 3 tuple pairs do not join. 
Consider $\astnode_1 \mapsto [\truevalue, \falsevalue]$, where $\astnode_1$ is a projection.
$\mapastnodetounderapprox(\astnode_1)$ states that only the first tuple in $\queryopinput^{\astnode_1}$ exists. 
In other words, although $\queryopoutput^{\astnode_2}$ has up to 2 tuples (according to $\mapastnodetounderapprox(\astnode_2)$), $\mapastnodetounderapprox(\astnode_1)$ considers only those tables of size one\footnote{$\tool$ sets a node's UA size heuristically, ranging from 2 to 16 and 2x2 to 16x16 for unary and binary operators.}.  
Consider $\astnode_3 \mapsto [\truevalue_{\truepredicate}, \truevalue_{\truepredicate}]$, where $\astnode_3$ is $\groupby$. 
$\mapastnodetounderapprox(\astnode_3)$ describes tables $\queryopinput^{\astnode_3}$ with 2 tuples, where each tuple leads to a new group. 
$\astnode'_1 \mapsto [ \truevalue_{\truepredicate}, \falsevalue ]$ also concerns $\groupby$ but uses a different UA. It considers $\queryopinput^{\astnode'_1}$ with up to 2 tuples, where the first tuple forms a new group but the second does not.

\label{ex:ua-map}
\end{example}

\subsection{Top-Level Algorithm and Problem Statement}
\label{sec:alg:top-level}

\begin{figure}[!t]
\vspace{-5pt}
\small 
\begin{algorithm}[H]
\caption{Top-level algorithm.}
\label{alg:top-level}
\begin{algorithmic}[1]

\Statex\myprocedure{\geninput($\sqlquery_1, \mydots, \sqlquery_n, \appcond$)}
\vspace{1pt}

\Statex\Input{Each $\sqlquery_i$ is a query. $\appcond$ is an application condition (expressed as an SMT formula).}

\Statex\Output{A database $\inputdatabase$ that satisfies $\appcond$, or $null$ indicating no such $\inputdatabase$ is found.}
\vspace{2pt}

\State
$\mapastnodetounderapprox \assign \conflictdrivenUAsearch( \sqlquery_1, \mydots, \sqlquery_n, \appcond )$; 
\label{alg:top-level:call-search}

\vspace{1pt}

\If{$\mapastnodetounderapprox = null$}
\label{alg:top-level:null}
\Return $null$;
\EndIf

\State
\Return 
$\buildinputdatabase(\mapastnodetounderapprox, \appcond)$;
\label{alg:top-level:return-input}
\end{algorithmic}
\end{algorithm}
\vspace{-15pt}
\end{figure}

Let us switch gears and present our under-approximation search algorithm. 
Algorithm~\ref{alg:top-level} shows our top-level algorithm. 
Given $\sqlquery_1, \mydots, \sqlquery_n$ and an SMT formula $\appcond$ (encoding the application condition) over variables $\queryoutput_1, \mydots, \queryoutput_n$ (each $\queryoutput_i$ denoting $\sqlquery_i$'s output), $\geninput$ returns a \emph{satisfying input} $\inputdatabase$ such that $\appcond[ \queryoutput_1 \mapsto \sqlquery_1(\inputdatabase), \mydots, \queryoutput_n \mapsto \sqlquery_n(\inputdatabase)]$ is true, or returns $null$ if no such $\inputdatabase$ is found.

Our key novelty lies in $\conflictdrivenUAsearch$ (line~\ref{alg:top-level:call-search}), which searches for a \emph{satisfying} UA map $\mapastnodetounderapprox$. 
Given such an $\mapastnodetounderapprox$, we invoke $\buildinputdatabase$ (line~\ref{alg:top-level:return-input}) to derive a satisfying input. 
Below, we first formulate the under-approximation search problem, and then explain $\buildinputdatabase$. 

\begin{definition}
\textbf{(Under-Approximation Search Problem).}
Given $n$ queries $\sqlquery_1, \mydots, \sqlquery_n$ from the language in Figure~\ref{fig:sql-syntax}, 
given $\underapproxfamily_{\queryop}$ that defines a family of UAs for each query operator $\queryop$, 
and given an SMT formula $\appcond$ over variables $\queryoutput_1, \mydots, \queryoutput_n$, 
find a UA map $\mapastnodetounderapprox$ which maps \emph{each} AST node $\astnode$ in $\sqlquery_i$ (for all $i \in [1,n]$) to a \emph{minimal} UA $\underapprox \in \underapproxfamily_{\getqueryop(\astnode)}$ such that 
\begin{equation}
\begin{array}{c}
\encode(\mapastnodetounderapprox) 
\land 
\Big(
\appcond[ \queryoutput_1 \mapsto \astnodeoutput^{\astnode_1}, \mydots, \queryoutput_n \mapsto \astnodeoutput^{\astnode_n}] 
\Big)_{\labelappcond}
\label{eq:UAmap-condition}
\end{array}
\end{equation}
is satisfiable. 
Here, $\astnode_i$ is $\sqlquery_i$'s root AST node (and recall that we use variable $\astnodeoutput^{\astnode}$ to denote  $\astnode$'s output), and $\labelappcond$ is simply a label for the application condition (which will later be used for conflict extraction). 
We call such $\mapastnodetounderapprox$ a \emph{satisfying under-approximation map}, or \emph{satisfying UA map} for short.  
\label{def:UA-search}
\end{definition}

In fact, given any $\mapastnodetounderapprox$ that maps AST nodes to UAs (which can be non-minimal), we can check the satisfiability of the above formula~(\ref{eq:UAmap-condition}): if it is satisfiable, we can obtain a satisfying UA map from $\mapastnodetounderapprox$, by refining all non-minimal UAs in $\mapastnodetounderapprox$ to minimal ones with the help of a satisfying assignment of~(\ref{eq:UAmap-condition}). The next section will explain how this works in detail.

\vspace{5pt}

{Given a satisfying UA map $\mapastnodetounderapprox$, $\buildinputdatabase$ first obtains a model $\smtmodel$ for the formula in~(\ref{eq:UAmap-condition}). Then, $\smtmodel(\queryopinput^{\relation})$ gives the content for each input relation $\relation$.}

\begin{example}
Figure~\ref{fig:ex:M} shows a satisfying UA map $\mapastnodetounderapprox$ for the equivalence refutation example from Section~\ref{sec:overview}. 
Given this $\mapastnodetounderapprox$, $\buildinputdatabase$ can generate the database in Figure~\ref{fig:example:cex}. 
\label{ex:sat-ua-map}
\end{example}

\subsection{Conflict-Driven Under-Approximation Search}
\label{sec:alg:search}

\begin{algorithm}[!t]
\small 
\caption{Algorithm for $\conflictdrivenUAsearch$.}
\label{alg:conflict-driven-UA-search}
\begin{algorithmic}[1]

\Statex\myprocedure{$\conflictdrivenUAsearch(\sqlquery_1, \mydots, \sqlquery_n, \appcond)$}
\vspace{1pt}

\Statex\Input{Each $\sqlquery_i$ is a query. $\appcond$ is an application condition.}

\Statex\Output{A satisfying UA map $\mapastnodetounderapprox$, or $null$ indicating no such $\mapastnodetounderapprox$ is found.}
\vspace{2pt}

\State
$\mapastnodetounderapprox \assign \emptyset$; \ \ \ \ 
$\conflicts \assign \emptyset$; \ \ \ \ 
$\worklist \assign \bigcup_{i \in [1,n]} \getastnodes( \sqlquery_i )$; 
\label{alg:conflict-driven-UA-search:init}

\While{true}
\label{alg:conflict-driven-UA-search:while-begin}

\State 
$\encodingall \assign \encode(\mapastnodetounderapprox) \land \appcond$; 
\label{alg:conflict-driven-UA-search:condition}

\If{$\encodingall$ is satisfiable}
\label{alg:conflict-driven-UA-search:sat-case}

\State
$\smtmodel \assign \getmodel(\encodingall)$;
\ \ \ \ 
$\mapastnodetounderapprox \assign \mapastnodetounderapprox \big[ \astnode \mapsto \smtmodel( \varUAvector^{\astnode} ) \ | \ \astnode \in \dom(\mapastnodetounderapprox) \big]$;
\label{alg:conflict-driven-UA-search:obtain-model-and-update}

\vspace{1pt}

\If{$\worklist = \emptyset$}
\Return $\mapastnodetounderapprox$;
\EndIf
\label{alg:conflict-driven-UA-search:return-M}

\State
$( \astnode_1, \mydots, \astnode_k ) \assign \worklist.\textit{remove}()$; \ \ \ \ 
$\mapastnodetounderapprox \assign \mapastnodetounderapprox[ \astnode_1 \mapsto \getTopUA(\underapproxfamily_{\astnode_1}), \mydots, \astnode_k \mapsto \getTopUA(\underapproxfamily_{\astnode_k}) ]$;
\label{alg:conflict-driven-UA-search:add-more-nodes}

\Else 
\label{alg:conflict-driven-UA-search:unsat-case}

\State
$\conflictASTnodes \assign \extractconflict( \encodingall )$; 
\ \ \ \ 
$( \mapastnodetounderapprox, \conflicts ) \assign \fixconflict( \mapastnodetounderapprox, \conflictASTnodes, \appcond, \conflicts)$;
\label{alg:conflict-driven-UA-search:resolve-conflict}

\If{$\mapastnodetounderapprox = null$} 
\Return $null$;
\label{alg:conflict-driven-UA-search:return-null}

\EndIf

\EndIf

\EndWhile
\label{alg:conflict-driven-UA-search:while-end}

\end{algorithmic}
\end{algorithm}

Now, let us unpack the $\conflictdrivenUAsearch$ procedure---see Algorithm~\ref{alg:conflict-driven-UA-search}.
At a high level, it iteratively generates a sequence of \emph{candidate} UA maps: 
$\mapastnodetounderapprox$ is initially empty (line~\ref{alg:conflict-driven-UA-search:init}), and is iteratively updated in two ways depending on if it satisfies $\encodingall$ at line~\ref{alg:conflict-driven-UA-search:condition}. 
This $\encodingall$ is essentially the same condition as in~(\ref{eq:UAmap-condition}), but we omit the variable renaming part and the label for $\appcond$, to simplify the presentation. 
Note that before termination, $\mapastnodetounderapprox$ is \emph{partial}; that is, $\mapastnodetounderapprox$ does not contain all AST nodes from all $\sqlquery_i$'s. 
To update $\mapastnodetounderapprox$, we either 
(i) add more entries with new AST nodes (line~\ref{alg:conflict-driven-UA-search:add-more-nodes}), or 
(ii) modify existing entries by remapping some of the existing nodes to new UAs (done by $\fixconflict$ at line~\ref{alg:conflict-driven-UA-search:resolve-conflict}). 
In what follows, we explain in more detail how (i) and (ii) work, beginning with (i).

\newpara{Lines~\ref{alg:conflict-driven-UA-search:sat-case}-\ref{alg:conflict-driven-UA-search:add-more-nodes}.}
If $\encodingall$ is satisfiable (line~\ref{alg:conflict-driven-UA-search:sat-case}), we first update $\mapastnodetounderapprox$ using a satisfying assignment $\smtmodel$ of $\encodingall$ (line~\ref{alg:conflict-driven-UA-search:obtain-model-and-update}). 
In particular, $\smtmodel$ maps $\varUAvector^{\astnode}$ to a \emph{minimal} UA, for every AST node  in the domain of $\mapastnodetounderapprox$. 
This updated $\mapastnodetounderapprox$ is guaranteed to satisfy $\encodingall$ and contains only minimal UAs. But it may still be partial. 
Therefore, line~\ref{alg:conflict-driven-UA-search:return-M} checks if $\worklist$ (which initially includes all AST nodes) contains any additional nodes. 
If not (meaning $\mapastnodetounderapprox$ contains all nodes), $\mapastnodetounderapprox$ must be a satisfying UA map and hence the algorithm terminates (line~\ref{alg:conflict-driven-UA-search:return-M}). 
Otherwise (i.e., $\mapastnodetounderapprox$ is partial), line~\ref{alg:conflict-driven-UA-search:add-more-nodes} adds $k$ entries to $\mapastnodetounderapprox$, each mapping a new node $\astnode_i$ (removed from $\worklist$) to a Top UA (from $\underapproxfamily_{\astnode_i}$). 
An implication of line~\ref{alg:conflict-driven-UA-search:add-more-nodes} is that $\mapastnodetounderapprox$ at line~\ref{alg:conflict-driven-UA-search:condition} may map some nodes to Top UAs---this is exactly why we need line~\ref{alg:conflict-driven-UA-search:obtain-model-and-update}. 
This way, we also maintain an invariant that $\mapastnodetounderapprox$ at line~\ref{alg:conflict-driven-UA-search:condition} has at most $k$ entries with non-minimal UAs ($k$ is a hyperparameter that can be tuned heuristically). 
Maintaining such a \emph{lightweight} $\mapastnodetounderapprox$ helps make the satisfiability check (at line~\ref{alg:conflict-driven-UA-search:sat-case}) fast. 
We will next explain how the else branch (lines~\ref{alg:conflict-driven-UA-search:unsat-case}-\ref{alg:conflict-driven-UA-search:return-null}) maintains this invariant. 

\begin{example}
Consider $\sqlquery$ and $\sqlquery'$ from Section~\ref{sec:overview}. 
Suppose $\mapastnodetounderapprox_1$ in Figure~\ref{fig:ex:M-sequence} is the (partial) UA map at line~\ref{alg:conflict-driven-UA-search:condition}. 
$\mapastnodetounderapprox_1$'s corresponding $\encodingall_1$ is satisfiable, so line~\ref{alg:conflict-driven-UA-search:obtain-model-and-update} will obtain a satisfying assignment $\smtmodel$.
Suppose $\smtmodel(\varUAvector^{\astnode_1}) = [\truevalue, \falsevalue]$ and $\smtmodel(\varUAvector^{\astnode'_1}) = [ \falsevalue, \falsevalue]$. 
This leads to $\mapastnodetounderapprox_1 = \big\{ \astnode_1 \mapsto [\truevalue, \falsevalue], \astnode'_1 \mapsto [ \falsevalue, \falsevalue] \big\}$ after line~\ref{alg:conflict-driven-UA-search:obtain-model-and-update}.
$\worklist$ is not empty at line~\ref{alg:conflict-driven-UA-search:return-M}, so line~\ref{alg:conflict-driven-UA-search:add-more-nodes} loads in more nodes---say $\astnode_2, \astnode'_2$---from $\worklist$, and maps them to Top UAs.
At this point, we obtain the UA map $\mapastnodetounderapprox_2$ shown in Figure~\ref{fig:ex:M-sequence}.
\label{ex:alg-add-nodes}
\end{example}

\newpara{Lines~\ref{alg:conflict-driven-UA-search:unsat-case}-\ref{alg:conflict-driven-UA-search:return-null}.}
Let us now examine the else branch where $\encodingall$ is unsatisfiable. Intuitively, this means some entries $\mapastnodetounderapprox$ must be adjusted, in order for $\mapastnodetounderapprox$ to be compatible with $\appcond$. 
Recall that $\encodingall$ is a conjunction of  $|\mapastnodetounderapprox|+1$ clauses, one for $\appcond$ and each of the entries in $\mapastnodetounderapprox$. 
So if $\encodingall$ is unsat, we know a subset of clauses must be in conflict (i.e., their conjunction is unsat). 
In this case, we first use $\extractconflict$ to obtain the set $\conflictASTnodes$ of nodes corresponding to these conflicting clauses (line~\ref{alg:conflict-driven-UA-search:resolve-conflict}). 
More formally:
\[
\conflictASTnodes = \{  \astnode_i \ | \  \labelastnode_i \in \textsc{UnsatCore}(\encodingall)  \}
\]
Here, $\textsc{UnsatCore}$ extracts an unsatisfiability core of $\encodingall$. Then, $\conflictASTnodes$ includes an AST node $\astnode_i$ only if this unsat core has a clause labeled $\labelastnode_i$ (recall that a different label $\labelappcond$ is used for $\appcond$). 
Given $\conflictASTnodes$, we define the following ``projection'' operation which gives us the subset of entries from $\mapastnodetounderapprox$ at $\conflictASTnodes$.
\[
\mapastnodetounderapprox \mapprojection \conflictASTnodes
= \{  \astnode \mapsto  \mapastnodetounderapprox(\astnode) \ | \  \astnode \in \conflictASTnodes \}
\]
We call this set of entries a \emph{conflict}. 
$\mapastnodetounderapprox$ must be updated to map some nodes in $\conflictASTnodes$ to different UAs, because otherwise $\mapastnodetounderapprox \mapprojection \conflictASTnodes$ will always be a conflict, no matter how the other entries are modified or what new entries are added. 
This update is done by $\fixconflict$ (line~\ref{alg:conflict-driven-UA-search:resolve-conflict}), which Section~\ref{sec:alg:resolve-conflict} will describe in detail. 
At a high level, $\fixconflict$ returns a new $\mapastnodetounderapprox$ with no conflicts at $\conflictASTnodes$, and also maintains the aforementioned invariant that $\mapastnodetounderapprox$ is lightweight. 
It also takes as input $\conflicts$---which is a set of currently discovered conflicts---and returns a new one.
For instance, the previous conflict $\mapastnodetounderapprox \mapprojection \conflictASTnodes$ will be added to $\conflicts$. 
$\fixconflict$ guarantees the returned $\mapastnodetounderapprox$ does not manifest any of these known conflicts; this is critical for termination. 
If $\mapastnodetounderapprox$ is $null$ (line~\ref{alg:conflict-driven-UA-search:return-null}), it means no UA map exists---given the family of UAs---that can avoid the current conflict at $\conflictASTnodes$.
We note that, while the update of  $\mapastnodetounderapprox$ is driven by a conflict of $\encodingall$, the updated $\mapastnodetounderapprox$ may not satisfy $\encodingall$; however, future iterations will keep fixing new conflicts until satisfaction.

\begin{example}
Consider our example in Section~\ref{sec:overview}, and suppose $\mapastnodetounderapprox_3$ in Figure~\ref{fig:ex:M-sequence} is the UA map at line~\ref{alg:conflict-driven-UA-search:condition}.
Its corresponding $\encodingall_3$ is unsat. 
$\extractconflict$ at line~\ref{alg:conflict-driven-UA-search:resolve-conflict} returns $\conflictASTnodes = \{ \astnode'_1, \astnode'_2, \astnode'_3, \astnode'_4, \astnode_1, \astnode_2 \}$. 
That is, $\mapastnodetounderapprox_3 \mapprojection \conflictASTnodes$ is a conflict---to understand why, focus on \emph{Invoices} which is shared by $\sqlquery$ and $\sqlquery'$. 
Let us assume $\encodingall_3$ is sat. 
Given $\mapastnodetounderapprox_3(\astnode'_1), \mapastnodetounderapprox_3(\astnode'_2), \mapastnodetounderapprox_3(\astnode'_3), \mapastnodetounderapprox_3(\astnode'_4)$, we know \emph{Invoices} must be empty, because $\mapastnodetounderapprox_3(\astnode'_1)$ constrains its input to be empty and $\astnode'_2, \astnode'_3, \astnode'_4$ are all left join operators. 
But, $\mapastnodetounderapprox_3(\astnode_1), \mapastnodetounderapprox_3(\astnode_2)$ suggest otherwise (i.e., \emph{Invoices} is non-empty). Contradiction. 
Note that $\astnode_3$ and $\astnode_4$ are not part of the conflict (although they were added by the previous iteration). 
$\fixconflict$ takes this conflict as input, and yields the $\mapastnodetounderapprox$ in Figure~\ref{fig:ex:M}. In particular, UAs at $\conflictASTnodes$ for $\mapastnodetounderapprox$ are not  in conflict anymore.

\label{ex:else-branch}
\end{example}

\subsection{Conflict Resolution and Accumulation}
\label{sec:alg:resolve-conflict}

Now let us proceed to the $\fixconflict$ procedure, which is presented in Algorithm~\ref{alg:resolveconflict}. 
It takes four inputs. $\appcond$ is the application condition, $\mapastnodetounderapprox$ is a UA map, $\conflictASTnodes \subseteq \dom(\mapastnodetounderapprox)$ is a subset of AST nodes from $\mapastnodetounderapprox$ where $\mapastnodetounderapprox \mapprojection \conflictASTnodes$ is a conflict, and $\conflicts$ is a set of conflicts. 
It returns a pair $(\mapastnodetounderapprox', \conflicts')$. 
In particular, $\mapastnodetounderapprox'$ is guaranteed to satisfy the following three properties, if it is not $null$. 
\begin{enumerate}[leftmargin=*]
\item 
$\mapastnodetounderapprox'$ is conflict-free at $\conflictASTnodes$. That is, $\encode(\mapastnodetounderapprox' \mapprojection \conflictASTnodes) \land \appcond$ is satisfiable. 
\item 
$\mapastnodetounderapprox'$ doesn't exhibit any of the conflicts in $\conflicts' \supseteq \conflicts$. That is, no subset of entries from $\mapastnodetounderapprox'$ is in $\conflicts'$. 
\item 
$\mapastnodetounderapprox'$ is lightweight. 
In fact, all entries in $\mapastnodetounderapprox$ use minimal UAs. 
\end{enumerate}
If no such $\mapastnodetounderapprox'$ exists, for the given  family of UAs, $null$ will be returned. 

Let us now dive into the internals of $\fixconflict$. 
Line~\ref{alg:resolveconflict:add-conflict} first creates $\conflicts'$ that additionally includes conflict $\mapastnodetounderapprox \mapprojection \conflictASTnodes$. 
Then, lines~\ref{alg:resolveconflict:loop-begin}-\ref{alg:resolveconflict:loop-end} aim to generate $\mapastnodetounderapprox'$ that satisfies the above three properties. 
The basic idea is simple. 
We first try to fix the conflict at $\conflictASTnodes$, by considering ``each''\footnote{We actually only consider those $\underapprox_i$'s from a covering set. We will expand on this later.} UA $\underapprox_i$ from $\underapproxfamily_{\astnode_i}$ for each node $\astnode_i$ in $\conflictASTnodes$ (line~\ref{alg:resolveconflict:loop-begin}). 
Given $( \underapprox_1, \mydots, \underapprox_m )$, we encode the UA semantics of $\mapastnodetounderapprox_{\conflictASTnodes}$ (line~\ref{alg:resolveconflict:encode-semantics-at-V}). 
If the conflict persists (line~\ref{alg:resolveconflict:check-conflict}), we add a new conflict $\mapastnodetounderapprox_{\conflictASTnodes}$ to $\conflicts'$ and continue. 
On the other hand, if $\encodingall_{\conflictASTnodes}$ is satisfiable---meaning property (1) holds at this point---we further make sure property (2) also holds. 
This is done by line~\ref{alg:resolveconflict:encode-conflicts} that further encodes $\conflicts'$ and all UAs for $\mapastnodetounderapprox$'s nodes outside $\conflictASTnodes$.
Specifically, 
\[
\encodeconflicts(\conflicts') 
\assign
\biglor_{\{ \astnode_1 \mapsto \underapprox_1, \mydots, \astnode_r \mapsto \underapprox_r \}\in \conflicts'} 
\ 
\bigland_{i = 1, \mydots, r}
\encodechoice( \underapprox_i ) [ \underapproxvar \mapsto \underapproxvar^{\astnode_i} ]
\]
If this $\encodingall'$ is satisfiable (line~\ref{alg:resolveconflict:check-conflicts}), we can easily construct an $\mapastnodetounderapprox'$ from a satisfying assignment $\smtmodel$ of $\encodingall'$ (line~\ref{alg:resolveconflict:create-M'}) that is guaranteed to satisfy both properties (1) and (2). 
Here, $\smtmodel$ maps every $\varUAvector^{\astnode}$ to a minimal UA; therefore, $\mapastnodetounderapprox'$ also meets property (3). 
We finally return at line~\ref{alg:resolveconflict:return-M'}.
If $\encodingall'$ is unsat for all iterations (meaning all combinations of UAs at $\conflictASTnodes$ are exhausted), we return $null$ (line~\ref{alg:resolveconflict:return-null}).

\newpara{Remarks.}
A few things are worth noting. 
First, line~\ref{alg:resolveconflict:loop-begin} uses $\splitUAspace$ which does not return all UAs in $\underapproxfamily_{\astnode_i}$, but returns a \emph{covering} set $S \subseteq \underapproxfamily_{\astnode_i}$ of UAs.
That is, for every minimal UA $\underapprox \in \underapproxfamily_{\astnode_i}$, there exists a UA $\underapprox' \in S$ such that $\underapprox' \subsumes \underapprox$.  
This allows us to search the entire space of UAs \emph{symbolically}, without explicitly enumerating all (minimal) UAs, thereby speeding up the search. 
The implementation of $\splitUAspace$ can be tuned heuristically. 
$\tool$ chooses to set a fixed number of values in the UA for each $\astnode_i \in \conflictASTnodes$ to top, and then enumerate all permutations of non-top values for the rest. 
Second, line~\ref{alg:resolveconflict:check-conflict} accumulates additional conflicts, with the same goal of accelerating the search. 
Finally, $\encodingall'$ (at line~\ref{alg:resolveconflict:encode-conflicts}) encodes all  UAs choices for AST nodes outside $\conflictASTnodes$. This is necessary for completeness. 
In other words, our algorithm would become incomplete, if it were to only modify UAs for $\conflictASTnodes$.

\begin{figure}
\vspace{-5pt}
\small 
\begin{algorithm}[H]
\caption{Algorithm for $\fixconflict$.}
\label{alg:resolveconflict}
\begin{algorithmic}[1]

\Statex\myprocedure{$\fixconflict( \mapastnodetounderapprox, \conflictASTnodes, \appcond, \conflicts )$}
\vspace{1pt}

\Statex\Input{$\conflictASTnodes$ is a subset of AST nodes from $\mapastnodetounderapprox$. $\appcond$ is the application condition. $\conflicts$ is a set of known conflicts. In particular, the entries from $\mapastnodetounderapprox$ at $\conflictASTnodes$ (i.e., $\mapastnodetounderapprox \mapprojection \conflictASTnodes$) are known to cause a conflict.}

\Statex\Output{A new UA map $\mapastnodetounderapprox'$ with no conflict at $\conflictASTnodes$. $\conflicts' \supseteq \conflicts$ is a new set of conflicts.}

\vspace{2pt}

\State
$\conflicts' \assign \conflicts \cup \{  \mapastnodetounderapprox \mapprojection \conflictASTnodes  \}$;
\label{alg:resolveconflict:add-conflict}

\vspace{1pt}

\For{$\underapprox_1 \in \splitUAspace(\underapproxfamily_{\astnode_1}), \mydots, \underapprox_m \in \splitUAspace( \underapproxfamily_{\astnode_m} )$ where $\conflictASTnodes = \{ \astnode_1, \mydots, \astnode_m \}$}
\label{alg:resolveconflict:loop-begin}

\vspace{1pt}

\State
$\mapastnodetounderapprox_{\conflictASTnodes} \assign \{ \astnode_1 \mapsto \underapprox_1, \mydots, \astnode_m \mapsto \underapprox_m \}$;
\ \ \ \ 
$\encodingall_{\conflictASTnodes} \assign  \encode(\mapastnodetounderapprox_{\conflictASTnodes}) \land  \appcond$;
\label{alg:resolveconflict:encode-semantics-at-V}

\vspace{2pt}

\If{$\encodingall_{\conflictASTnodes}$ is not satisfiable}
$\conflicts' \assign \conflicts' \cup \{ \mapastnodetounderapprox_{\conflictASTnodes} \}$; 
\textbf{continue};
\EndIf
\label{alg:resolveconflict:check-conflict}

\State
$\encodingall' \assign \encodingall_{\conflictASTnodes} \land \encodechoice \Big(  \big\{  \astnode \mapsto \topUA(\underapproxfamily_{\astnode}) \ | \  \astnode \in \dom(\mapastnodetounderapprox ) \setminus \conflictASTnodes \big\}  \Big)  \land \neg \encodeconflicts(\conflicts')$;
\label{alg:resolveconflict:encode-conflicts}

\If{$\encodingall'$ is satisfiable}
\label{alg:resolveconflict:check-conflicts}

\vspace{1pt}

\State 
$\smtmodel \assign \getmodel(\encodingall')$;
\ \ \ \ 
$\mapastnodetounderapprox' \assign \mapastnodetounderapprox \big[ \astnode \mapsto \smtmodel( \varUAvector^{\astnode} ) \ | \  \astnode \in \dom(\mapastnodetounderapprox) \big]$;
\label{alg:resolveconflict:create-M'}

\vspace{1pt}

\State
\Return $( \mapastnodetounderapprox', \conflicts' )$;
\label{alg:resolveconflict:return-M'}

\EndIf
\label{alg:resolveconflict:loop-end}

\EndFor

\State\Return $( null, \conflicts' )$;
\label{alg:resolveconflict:return-null}

\end{algorithmic}
\end{algorithm}
\vspace{-10pt}
\end{figure}

\begin{example}
Consider $\mapastnodetounderapprox_3$ from Figure~\ref{fig:ex:M-sequence}, which has a conflict at $\conflictASTnodes = \{ \astnode'_1, \astnode'_2, \astnode'_3, \astnode'_4, \astnode_1, \astnode_2 \}$. 
Suppose Figure~\ref{fig:ex:partition} corresponds to  line~\ref{alg:resolveconflict:loop-begin}. That is, 
\[\small 
\begin{array}{ll}
\splitUAspace(\astnode_1) = \{ [ \unknownvalue, \unknownvalue ] \} 
& 
\splitUAspace(\astnode'_1) = \{ [ \unknownvalue, \unknownvalue ] \}  

\\[5pt]

\splitUAspace(\astnode_2) = 
\{ \big[ [ \unknownvalue, \unknownvalue ], [ \unknownvalue, \unknownvalue ] \big] \} 
& 
\splitUAspace(\astnode'_2) = \{ \big[ [ \unknownvalue, \unknownvalue ], [ \unknownvalue, \unknownvalue ] \big] \} 

\\[5pt]

\splitUAspace(\astnode'_3) = 
\left\{
\begin{array}{l}
\big[ [ \truevalue, \truevalue ], [ \unknownvalue, \unknownvalue ] \big],
\big[ [ \truevalue, \falsevalue ], [ \unknownvalue, \unknownvalue ] \big] 
\\[3pt] 
\big[ [ \falsevalue, \truevalue ], [ \unknownvalue, \unknownvalue ] \big], 
\big[ [ \falsevalue, \falsevalue ], [ \unknownvalue, \unknownvalue ] \big] 
\end{array}
\right\}

& 

\splitUAspace(\astnode'_4) = 
\left\{
\begin{array}{l}
\big[ [ \truevalue, \truevalue ], [ \unknownvalue, \unknownvalue ] \big],
\big[ [ \truevalue, \falsevalue ], [ \unknownvalue, \unknownvalue ] \big] 
\\[3pt] 
\big[ [ \falsevalue, \truevalue ], [ \unknownvalue, \unknownvalue ] \big], 
\big[ [ \falsevalue, \falsevalue ], [ \unknownvalue, \unknownvalue ] \big] 
\end{array}
\right\}
\end{array}
\]
In particular, during the first iteration of the loop (lines~\ref{alg:resolveconflict:loop-begin}-\ref{alg:resolveconflict:loop-end}), $\mapastnodetounderapprox_{\conflictASTnodes}$ at line~\ref{alg:resolveconflict:encode-semantics-at-V} is the first UA map $\mapastnodetounderapprox^{(1)}_{\astnode}$ from Figure~\ref{fig:ex:partition}. $\encodingall_{\conflictASTnodes}$ is unsat; therefore, $\mapastnodetounderapprox_{\conflictASTnodes}$ is added to $\conflicts'$ (line~\ref{alg:resolveconflict:check-conflict}).
The next two iterations add $\mapastnodetounderapprox^{(2)}_{\conflictASTnodes}$ and $\mapastnodetounderapprox^{(3)}_{\conflictASTnodes}$ to $\conflicts'$ (line~\ref{alg:resolveconflict:check-conflict}), since they are both unsat.
At this point, $\conflicts'$ contains a total of four conflicts, including the initial $\mapastnodetounderapprox_3 \mapprojection \conflictASTnodes$.
Now, consider the fourth iteration. $\mapastnodetounderapprox^{(4)}_{\conflictASTnodes}$ from Figure~\ref{fig:ex:partition} corresponds to a satisfiable $\encodingall_{\conflictASTnodes}$ at line~\ref{alg:resolveconflict:encode-semantics-at-V}. 
In this case, line~\ref{alg:resolveconflict:encode-conflicts} encodes the UA choices for nodes outside $\conflictASTnodes$---namely, $\astnode_3$ and $\astnode_4$---and also encodes the current conflicts $\conflicts'$, which are conjoined with $\encodingall_{\conflictASTnodes}$ to form $\encodingall'$. 
Note that $\encodingall'$ only considers UA semantics for those AST nodes in $\conflictASTnodes$. In other words, if $\conflictASTnodes$ is small, the satisfiability check at line~\ref{alg:resolveconflict:check-conflicts} should be pretty fast. 
It turns out $\encodingall'$ is indeed sat. 
Therefore, we obtain a model $\smtmodel$ at line~\ref{alg:resolveconflict:create-M'}, from which we construct $\mapastnodetounderapprox'$---this is the $\mapastnodetounderapprox$ from Figure~\ref{fig:ex:M}. 
$\fixconflict$ terminates at this point, returning $\mapastnodetounderapprox$ and $\conflicts$ (with four conflicts). 

\label{ex:resolve-conflict}
\end{example}

\subsection{Theorems}
\label{sec:alg:discussion}

This section presents key theorems, whose proofs can be found in Appendix~\ref{sec:proof}.

\begin{theorem}[Correctness of UA semantics] \label{thm:semantics-correctness}
Suppose $\encodeuasemantics(\queryop, \underapprox)$ yields an SMT formula $\formula$, for query operator $\queryop$ and UA $\underapprox \in \underapproxfamily_{\queryop}$.
For any model $\smtmodel$ of $\formula$, the corresponding inputs $\smtmodel(\vec{x})$ and output $\smtmodel(y)$ are consistent with the precise semantics of $\queryop$; that is, $\denot{\queryop}_{\smtmodel(\vec{x})} = \smtmodel(y)$.
\xinyurevision{
Intuitively, this theorem states that any under-approximation $\underapprox$ of $\queryop$ is always encoded into an SMT formula whose satisfying assignments correspond to genuine input-output behaviors of $\queryop$. 
Therefore, our approach is consistent with incorrectness logic~\cite{o2019incorrectness} in that both consider reachable outputs (no false positives). 
}
\end{theorem}


\begin{theorem}[Refinement of UA Semantics] \label{thm:semantics-refinement}
Given query operator $\queryop$, 
and two UAs $\underapprox \in \underapproxfamily_{\queryop}$ and $\underapprox' \in \underapproxfamily_{\queryop}$ where  $\underapprox'$ refines $\underapprox$ (i.e., $\underapprox' \refines \underapprox$), we have  $\encodeuasemantics(\queryop, \underapprox) \Rightarrow \encodeuasemantics(\queryop, \underapprox')$.
\xinyurevision{Intuitively, this means the set of $\queryop$'s reachable outputs for $\underapprox'$ should be a subset of that for $\underapprox$; therefore, the UA semantics encoding for $\underapprox'$ is entailed by that for $\underapprox$.}
\end{theorem}

\begin{theorem}[Soundness] \label{thm:soundness}
Given queries $\sqlquery_1, \mydots, \sqlquery_n$ and application condition $\appcond$, if $\geninput$ returns an input database $\inputdatabase$, then we have $\appcond[ \queryoutput_1 \mapsto \sqlquery_1(\inputdatabase), \mydots, \queryoutput_n \mapsto \sqlquery_n(\inputdatabase) ]$ is true. 
\xinyurevision{Intuitively, this is because our UAs always encode genuine input-output behaviors of $\sqlquery_i$.
}

\end{theorem}


\begin{theorem}[Completeness] \label{thm:completeness}
Given queries $\sqlquery_1, \mydots, \sqlquery_n$ and application condition $\appcond$, if $\geninput$ returns $null$, then there does not exist an input $\inputdatabase$ (with respect to the semantics presented in Figure~\ref{fig:full-semantics}) for which $\appcond[ \queryoutput_1 \mapsto \sqlquery_1(\inputdatabase), \mydots, \queryoutput_n \mapsto \sqlquery_n(\inputdatabase)]$ is true. 
\xinyurevision{Intuitively, this is because our approach essentially performs exhaustive search while using the lattice structure of UAs to soundly prune the search space.}
\end{theorem}


%% file: figures/fig-sql-syntax.tex
\begin{figure}[!t]
\small 
\centering
\[
\begin{array}{rccl}

\textit{Query} & 
\query & 
::= & 
\relation ~|~ \proj_\attrlist(\query) ~|~ \filter_\predicate(\query) ~|~ \rename_\relation(\query) \ | \ \unionall(\query_1, \query_2) ~| \alljoin(\query_1, \query_2)   \\
& & 
| & 
\distinct(\query) 
~|~ 
\groupby_{\vec{\expression}, \attrlist, \predicate}(\query) 
~|~ 
\orderby_\expression(\query) 
~|~ 
\with(\vec{\query}, \vec{\relation}, \query) 
\\

\textit{Join Op} & 
\alljoin & 
::= & 
\product ~|~ \ijoin_\predicate ~|~ \ljoin_\predicate ~|~ \rjoin_\predicate ~|~ \fjoin_\predicate 
\\

\textit{Attr List} & 
\attrlist & 
::= 
& [\expression, \mydots, \expression] \\

\textit{Pred} & 
\predicate & 
::= & 
b ~|~ \nullv ~|~ \expression \alllogic \expression ~|~ \vec{\expression} \in \vec{v} ~|~ \vec{\expression} \in Q ~|~ \text{IsNull}(\expression) ~|~ \text{Exists}(\query)
                 ~|~  \predicate \land \predicate ~|~ \predicate \lor \predicate ~|~ \neg \predicate 
\\

\textit{Expr} & 
\expression & 
::= & 
a ~|~ v ~|~ \expression \allarith \expression ~|~ \mathcal{G}(b, \expression) ~|~ \text{ITE}(\predicate, \expression, \expression) ~|~ \text{Case}(\vec{\predicate}, \vec{\expression}, \expression) ~|~ \rename_a(\expression)
\\ 

\textit{Arith Op} & 
\allarith & 
::= & 
+ ~|~ - ~|~ \times ~|~ / ~|~ \% 
\\

\textit{Logic Op} & 
\alllogic & 
::= & 
\leq ~|~ < ~|~ = ~|~ \neq ~|~ > ~|~ \geq 
\\

\end{array}
\]
\[
\begin{array}{c}
\relation \in \textbf{Relation Names} \quad a \in \textbf{Attribute Names} \quad v \in \textbf{Values} \quad b \in \textbf{Bools} \\[3pt]
\mathcal{G} \in \{\text{Count}, \text{Min}, \text{Max}, \text{Sum}, \text{Avg}\}
\end{array}
\]
\vspace{-5pt}
\caption{Syntax of our query language.}
\label{fig:sql-syntax}
\vspace{-5pt}
\end{figure}

%% file: figures/fig-under-semantics.tex
\begin{figure}[!t]
\scriptsize
\centering

\[
\arraycolsep=1pt\def\arraystretch{1}
\begin{array}{ll}

(1) & 
\irule
{
\begin{array}{l}

\queryopinput = [\tuple_1, \mydots, \tuple_n] 

\quad

\queryopoutput = [\tuple'_1, \mydots, \tuple'_n]

\quad

\underapproxvar = [ \underapproxvar_1, \mydots, \underapproxvar_n ]

\quad

\attrlist = \tableattrs(\queryopinput)

\\

\queryopencoding_{i, \truevalue} = 
( \underapproxvar_i = \truevalue ) 
\to 
\big(
\cmdExist(\tuple_i, \tuple'_i, \denot{\predicate}_{\tuple'_i}) \land \cmdCopy(\tuple_i, \tuple'_i, \attrlist) 
\big)

\quad

\queryopencoding_{i, \falsevalue} = 
( 
\underapproxvar_i = \falsevalue ) 
\to 
\cmdCondDel(\tuple_i, \tuple'_i, \denot{\predicate}_{\tuple'_i}
)

\end{array}
}
{
\encodefullsemantics 
\big( 
\filter_{\predicate}
\big) 
\ruleleadsto 
\bigland_{i = 1, \mydots, n}
\queryopencoding_{i, \truevalue} 
\land 
\queryopencoding_{i, \falsevalue}
}

\\ \\ 

(2) & 
\irule
{
\begin{array}{l}

\queryopinput = [\tuple_1, \mydots, \tuple_n] 
\quad
\queryopoutput = [\tuple'_{1}, \mydots, \tuple'_n] 
\quad
\underapproxvar = [ \underapproxvar_1, \mydots, \underapproxvar_n ] 

\\

\queryopencoding_{i, \truevalue} = 
( \underapproxvar_i = \truevalue ) 
\to 
\big( 
\cmdExist(\tuple_i, \tuple'_i, \truepredicate) 
\land
\cmdCopy(\tuple_i, \tuple'_i, \attrlist) 
\big)

\quad 

\queryopencoding_{i, \falsevalue} = 
( \underapproxvar_i = \falsevalue ) 
\to 
\cmdCondDel(\tuple_i, \tuple'_i, \truepredicate)

\end{array}
}
{
\encodefullsemantics \big( \proj_\attrlist \big)  
\ruleleadsto 
\bigland_{i = 1, \mydots, n}
\queryopencoding_{i, \truevalue} 
\land 
\queryopencoding_{i, \falsevalue}
}

\\ \\

(3) & 
\irule
{
\begin{array}{l}

\queryopinput_1 = [ \tuple_1, \mydots, \tuple_{n_1} ] 

\ \ 

\queryopinput_2 = [ \tuple'_1, \mydots, \tuple'_{n_2} ] 

\ \ 

\queryopoutput = [ \tuple\doubleprime_{1,1}, \mydots, \tuple\doubleprime_{n_1, n_2} ] 

\ \ 

\underapproxvar = \big[ [ \underapproxvar_{1, 1}, \mydots \underapproxvar_{1, n_2} ], \mydots, [ \underapproxvar_{n_1, 1 }, \mydots, \underapproxvar_{n_1, n_2}] \big]

\ \

\attrlist_i =  \tableattrs(\queryopinput_i)

\\[2pt]

\queryopencoding_{i, j, \truevalue} = 
( \underapproxvar_{i,j} = \truevalue ) 
\to 
\cmdExist 
\Big( 
( \tuple_i, \tuple'_j ), \tuple\doubleprime_{i, j}, \denot{\predicate}_{\tuple_i, \tuple'_j} 
\Big) 
\land 
\cmdCopy( \tuple_i, \tuple\doubleprime_{i,j}, \attrlist_1 ) 
\land 
\cmdCopy( \tuple'_j, \tuple\doubleprime_{i,j}, \attrlist_2 )

\\ 

\queryopencoding_{i, j, \falsevalue} = 
( \underapproxvar_{i,j} = \falsevalue )  
\to 
\cmdCondDel 
\Big( 
( \tuple_i, \tuple'_j ),
\tuple\doubleprime_{i,j}, 
\denot{\predicate}_{\tuple_i, \tuple'_j} 
\Big)

\end{array}
}
{
\encodefullsemantics \big( \ijoin_\predicate \big) 
\ruleleadsto 
\bigland_{i = 1, \mydots, n_1} 
\bigland_{ j = 1, \mydots, n_2}
\queryopencoding_{i, j, \truevalue}
\land 
\queryopencoding_{i, j, \falsevalue}
}

\\ \\

(4) & 
\irule
{
\begin{array}{l}

\queryopinput = [\tuple_1, \mydots, \tuple_n] 

\quad

\queryopoutput = [\tuple'_{1}, \mydots, \tuple'_n] 

\quad

\underapproxvar = [ \underapproxvar_1, \mydots, \underapproxvar_n ]

\\[2pt]

\queryopencoding_{i, \falsevalue} =

( \underapproxvar_i = \falsevalue ) 

\to

\Big(

\big( 
\del( \tuple_i ) 
\lor 

\Big( 
\neg \del( \tuple_i )

\land

\biglor_{j=1, \mydots, i-1} 
\big( 
\neg \del(\tuple_j) 
\land 
\bigland_{a \in \vec{\expression}} \denot{\tuple_i.a} = \denot{\tuple_j.a} 
\land 
\group(\tuple_i) = j
\big) 

\Big) 
\big)

\land 

\del( \tuple'_i )

\Big)

\\

\queryopencoding_{i, \neg \falsevalue} = 
\neg \del(\tuple_i) 
\land 
\neg \biglor_{j=1, \mydots, i-1} 
\Big( 
\neg \del(\tuple_j) \land \bigland_{a \in \vec{\expression}}\denot{\tuple_i.a} = \denot{\tuple_j.a}
\Big) 
\land 
\group(\tuple_i) = i

\\

\queryopencoding_{i, \truevalue} = 
( \underapproxvar_i = \truevalue ) 
\to 
\Big( 
\queryopencoding_{i, \neg \falsevalue}
\land 
\neg \denot{\predicate}_{\group^{-1}(i)} 
\land 
\del(\tuple'_i)
\Big)

\\ 

\queryopencoding_{i, \truevalue_{\predicate}} = 
( \underapproxvar_i = \truevalue_{\predicate} ) 
\to 
\Big( 
\queryopencoding_{i, \neg \falsevalue}
\land 
\denot{\phi}_{\group^{-1}(i)} \land \neg \del(\tuple'_i) \land \cmdCopy(\group^{-1}(i), \tuple'_i, \attrlist)
\Big)

\end{array}
}
{
\encodefullsemantics \big( \groupby_{\vec{\expression}, \attrlist, \predicate} \big) 
\ruleleadsto 
\bigland_{i = 1, \mydots, n}
\queryopencoding_{i, \truevalue} 
\land 
\queryopencoding_{i, \truevalue_{\predicate}} 
\land 
\queryopencoding_{i, \falsevalue} 
}

\end{array}
\]

\vspace{-10pt}
\caption{Full semantics for a representative subset of query operators from Figure~\ref{fig:sql-syntax}.}
\label{fig:full-semantics}
\vspace{-5pt}
\end{figure}

%% file: sections/eval.tex
\section{Evaluation}
\label{sec:eval}

This section describes a series of experiments designed to answer the following questions:
\begin{itemize}[leftmargin=*]
\item 
\textbf{RQ1}: 
Can $\tool$ effectively solve real-world benchmarks? 
\item 
\textbf{RQ2}: 
How does $\tool$ compare against state-of-the-art techniques? 
\item 
\textbf{RQ3}:
How useful are various ideas in $\tool$?
\end{itemize}

\newpara{Two applications.}
The first one is the long-standing problem of SQL equivalence checking~\cite{he2024verieql,wang2018speeding,chu2017cosette,chu2017demonstration}. 
Given $\sqlquery_1$ and $\sqlquery_2$, the goal is to generate  $\inputdatabase$ such that $\queryoutput_1 \neq \queryoutput_2$, where $\queryoutput_i$ is $\sqlquery_i$'s output $\sqlquery_i(\inputdatabase)$. 
The other is query disambiguation~\cite{brancas2022cubes,wang2018speeding}. 
Given $\sqlquery_1, \mydots, \sqlquery_n$, find an input $\inputdatabase$ as well as an even split of them into two disjoint sets, such that for all $\sqlquery_i$ and $\sqlquery_j$: if they belong to the same set, $O_i = O_j$; otherwise $O_i \neq O_j$.\footnote{The encoding of this application condition can be found in Appendix~\ref{sec:disambiguation-encoding}.}
In what follows, we describe how we collect benchmarks for both applications.

\xinyurevision{
\newpara{Equivalence refutation benchmarks.}
We reuse the 24,455 benchmarks from a recent equivalence-checking work $\verieql$~\cite{he2024verieql}, where each benchmark is a pair of SQL queries.  These query pairs are obtained from a wide range of downstream tasks, including auto-grading, query rewriting and mutation testing. 
For instance, more than 20,000 query pairs correspond to auto-grading, where one query in the pair is the ground-truth and the other is a user submission accepted by LeetCode (i.e., passing LeetCode's test cases). In other words, non-equivalence in this case indicates inadequate testing, and counterexamples can help LeetCode developers further strengthen their test suites. 
We believe this is a comprehensive set of benchmarks for evaluating $\tool$.
Finally, we note that, while $\verieql$~\cite{he2024verieql} was able to refute more than 3,000 of these benchmarks, the solvability for the rest is unknown (and manually checking these benchmarks is a non-starter).
}

\newpara{Disambiguation benchmarks.}
We curate an extensive set of query disambiguation benchmarks, based on query synthesis tasks from the $\cubes$ work~\cite{brancas2022cubes} (which is a state-of-the-art SQL synthesizer).
In particular, given input-output examples from each $\cubes$ task, we generate one disambiguation \emph{task}, containing all satisfying queries that can be synthesized within 2 minutes. 
This yields {2,861} disambiguation tasks. 
The number of queries per task ranges from {2} to {3,434}, with an average of {272} and a median of {191}. 
While $\tool$ is directly applicable, a challenge for conducting our evaluation here is how to interpret the results: if a task was not solved, is that because it is not solvable at all, or is it due to $\tool$'s inability? 
Manually solving these tasks (even a subset) is nearly impossible. 
To address this, we develop a procedure to curate disambiguation \emph{benchmarks} (based on our {2,816} tasks) which are by construction solvable. 
\xinyurevision{
In particular, given a set $S$ of synthesized queries from the task, our first step is to partition $S$ into a set of equivalence classes $G$ as follows. 
\[\small 
\fbox{
$\begin{array}{ll}
\textbf{procedure} \textsc{ Partition}(S) \\ 
\ \ \ \ 
G \assign \emptyset;  \\ 
\ \ \ \ 
\textbf{while } S \text{ is not empty} \textbf{ do} \\ 
\ \ \ \ \ \ \ \ 
\sqlquery \assign S.\textit{remove}();  
\ \ \ \  
C = [ \sqlquery ];
\\ 
\ \ \ \  \ \ \ \ 
\textbf{foreach } \sqlquery_i \in S \textbf{ do} \\ 
\ \ \ \  \ \ \ \ 
\ \ \ \ 
\textbf{if } \sqlquery_i \text{ is equivalent to } \sqlquery \text{ up to bound } b \textbf{ then  } 
C.\textit{add}  ( \sqlquery_i ); 
\ \ S.\textit{remove}(\sqlquery_i); 
\\ 
\ \ \ \  \ \ \ \ 
\ \ \ \ 
\textbf{else if } \sqlquery_i \text{ is not equivalent to } \sqlquery \textbf{ then  } 
\textbf{continue}; 
\\ 
\ \ \ \  \ \ \ \ 
\ \ \ \ 
\textbf{else } S.\textit{remove}(\sqlquery_i); 
\ \ \ \ \triangleright \text{ bounded verifier timed out}
\\ 
\ \ \ \ \ \ \ \  
G \assign G \cup \{ C \}; \\ 
\ \ \ \  
\textbf{return } G; 
\end{array}$
}
\]
The key idea is to leverage a (bounded) equivalence verifier (in particular, $\verieql$~\cite{he2024verieql}) to partition $S$: queries in the same class $C$ are equivalent to each other (up to a certain bound $b$), whereas queries from different classes are not. 
In particular, the first query $\sqlquery$ in each class $C$ is a representative. We have a counterexample for any two representatives from two classes. 

Then, given $G$, we create a disambiguation benchmark with $2n$ queries, by selecting the first $n$ queries from each of any two classes. 
We consider $n=25, 50$ in our evaluation, yielding 4,245 and 2,475 disambiguation benchmarks, respectively. 
We name them D-50 and D-100.
}

\subsection{RQ1: Can $\tool$ Solve Real-World Benchmarks?}
\label{sec:eval:RQ1}

\evalfinding{$\tool$ solved {5,497} equivalence refutation benchmarks (out of 24,455 in total), with a median running time of {0.1} seconds per benchmark. 
Among {4,245} disambiguation benchmarks (each with {50} queries), $\tool$ solved {94\%} of them using a median of {3.7} seconds per benchmark.}

\newpara{Setup.}
Given a benchmark (either equivalence refutation or disambiguation), we run $\tool$ and record:
(1) if the benchmark is solved before timeout (1 minute), 
and (2) if so, the running time. 
We also log detailed statistics, which we will summarize and report below.

\newpara{Results.}
Table~\ref{tab:eval:RQ1-results} summarizes our key results. $\tool$ can solve {22.5\%} of the equivalence refutation benchmarks. This is a surprisingly high ratio, given that {98.1\% of our ER benchmarks are queries accepted by LeetCode.}
For disambiguation, $\tool$ consistently solves {over 90\% of the benchmarks in all settings.}
The solving time in general is pretty fast. 
$\tool$ slows down when disambiguating more queries, which is expected; but still, {within 40 seconds}, it can solve over {83\%} of the benchmarks. 

The next columns report some key internal statistics. 
In general, it takes $\tool$ a small number of iterations to solve a benchmark. This is because---while the number of AST nodes in a benchmark is large---line~\ref{alg:conflict-driven-UA-search:add-more-nodes} in Algorithm~\ref{alg:conflict-driven-UA-search} loads in multiple nodes in one iteration; therefore, we observe far fewer iterations. 
Some iterations are spent on fixing conflicts, by invoking $\fixconflict$; we report the number of such iterations in ``\#$\fixconflict$''. 
While the maximum can be over 100, the number of such calls is quite small in general.
But the impact of each call is much larger---look at the ``\#Nodes adjusted'' column. On average, over a quarter of a benchmark's nodes are remapped to different UAs. 
In other words, there are quite some backtrackings happening, but they are packed into a small number of $\fixconflict$ calls. 
We should note that the number of nodes in conflict (i.e., $|\conflictASTnodes|$ at line~\ref{alg:conflict-driven-UA-search:resolve-conflict} of Algorithm~\ref{alg:conflict-driven-UA-search}) is typically small. In particular, the median and average for ER are 4 and 5.1. For D-50 and D-100, median/average are 2/49.2 and 2/64.3, respectively. 
In other words, line~\ref{alg:resolveconflict:encode-conflicts} of Algorithm~\ref{alg:resolveconflict} encodes UA semantics for very few nodes (i.e., those in $\conflictASTnodes$), whereas it encodes UA choices for far more nodes (i.e., those outside $\conflictASTnodes$). 
The latter is typically much cheaper to solve.

\newcolumntype{R}[1]{>{\raggedleft\arraybackslash}p{#1}}

\newcolumntype{?}{!{\vrule width 0.5pt}}

\begin{table}[!t]
\caption{$\tool$ results across all benchmarks for equivalence refutation (ER) and disambiguation {(D-50 and D-100).} 
We first report the (average, median, maximum) number of AST nodes per benchmark---the number of AST nodes for a benchmark is the sum of AST sizes across all queries in the benchmark.
Then we show the number of benchmarks solved (with the total number of benchmarks below it) in ``\#Solved (\#Total)'' column, followed by the (average, median) running times across all solved benchmarks (recall: timeout is 1 minute). 
``\#Iterations'' shows the (average, median, maximum) number of iterations (lines~\ref{alg:conflict-driven-UA-search:while-begin}-\ref{alg:conflict-driven-UA-search:while-end} in Algorithm~\ref{alg:conflict-driven-UA-search}) across all \emph{solved} benchmarks. 
The next column reports the number of calls to $\fixconflict$. 
Finally, we present the (average, median, maximum) number of AST nodes adjusted before and after calling $\fixconflict$---that is, the number of AST nodes in $\mapastnodetounderapprox$ (see Algorithm~\ref{alg:resolveconflict}) that are mapped to different UAs in $\mapastnodetounderapprox'$ (i.e., $|\mapastnodetounderapprox' \setminus \mapastnodetounderapprox|$).} 
\vspace{-5pt}
\centering
\scriptsize
\begin{tabular}{c|rR{18pt}R{18pt}|r|rr|rrr|rrr|rrr}

\toprule

\multirow{2}{*}{} & 
\multicolumn{3}{c|}{\#Nodes per benchmark} & 
{\#Solved} & 
\multicolumn{2}{c?}{Time (sec)} & 
\multicolumn{3}{c|}{\#Iterations} & 
\multicolumn{3}{c|}{\#$\fixconflict$} & 
\multicolumn{3}{c}{\#Nodes adjusted} 

\\[1pt]

& 
\emph{avg.} & \emph{med.} & \emph{max.} & 
(\#Total) & 
\emph{avg.} & \emph{med.} & 
\emph{avg.} & \emph{med.} & \emph{max.} & 
\emph{avg.} & \emph{med.} & \emph{max.} & 
\emph{avg.} & \emph{med.} & \emph{max.} 

\\

\midrule

ER & 
9 & 8 & 33 & 
\begin{tabular}{@{}c@{}} {\raggedleft \ \!\!\; 5,497} \\ {\raggedleft (24,455)}  \end{tabular} & 
0.6 & 0.1 & 
4.5 & 3 & 139 & 
2.5 & 1 & 137 & 
6 & 5 & 26

\\ 
\midrule\midrule

D-50 & 
211 & 213 & 343 & 
\begin{tabular}{@{}c@{}} {\raggedleft 4,004} \\ {\raggedleft (4,245)}  \end{tabular} & 
6.4 & 3.7 & 
4.8 & 4 & 39 & 
1.8 & 1 & 36 & 
82 & 51 & 280

\\

\cmidrule[0.5pt]{2-16}

D-100 & 
435 & 442 & 616 & 
\begin{tabular}{@{}c@{}} {\raggedleft 2,392} \\ {\raggedleft (2,475)}  \end{tabular} & 
21.7 & 18 & 
4.9 & 4 & 23 & 
1.9 & 1 & 20 & 
126 & 101 & 380

\\

\bottomrule

\end{tabular}
\label{tab:eval:RQ1-results}

\end{table}

\newpara{Discussion.}
In our experience, the number of AST nodes added to $\mapastnodetounderapprox$ (i.e., $k$ at line~\ref{alg:conflict-driven-UA-search:add-more-nodes} in Algorithm~\ref{alg:conflict-driven-UA-search}) in each iteration has a significant impact on the overall performance. 
One extreme is to set $k = 1$; in this case, the number of iterations would be at least the number of AST nodes across $\sqlquery_i$'s (see ``\#Nodes per benchmark'' column in Table~\ref{tab:eval:RQ1-results}). This would slow down the algorithm significantly. 
On the other hand, adding all nodes in one iteration is also suboptimal (as we will show in Section~\ref{sec:eval:RQ3}). 
$\tool$'s heuristic is to add all nodes for $k$ queries (e.g., $k=50$ for D-100), which seems to achieve a good balance between (1) the number of iterations and (2) SMT solving overhead in each iteration. 
The implementation of $\splitUAspace$ also matters noticeably in practice. The two extremes (namely, returning Top UA for each $\astnode_i$ at line~\ref{alg:resolveconflict:loop-begin} of Algorithm~\ref{alg:resolveconflict}, and returning all minimal UAs for $\astnode_i$) would be slow (which we will show in Section~\ref{sec:eval:RQ3}). 
$\tool$'s heuristic is to use top for a fixed number of values in $\astnode_i$'s UA \xinyurevision{(in particular, we set 8 UA choices to $\unknownvalue$ in our experiments)}, while spelling out all permutations for the rest of the UA values. 
In our experience, how many UA values are set to top seems to impact the performance more than which values.

\subsection{RQ2: $\tool$ vs. State-of-the-Art}
\label{sec:eval:RQ2}

\evalfinding{$\tool$ outperforms all state-of-the-art techniques by a significant margin---in particular, 1.7x more benchmarks solved for equivalence refutation and 1.4x for disambiguation.}

\newpara{Baselines.}
For each application, we consider all relevant existing work---to our best knowledge---as our baselines. 
For equivalence refutation, we include 6 baselines: 
(1) $\verieql$~\cite{he2024verieql,zhaodemonstration} (SMT-based), 
(2) $\cosette$~\cite{chu2017cosette} (based on $\textsc{Rosette}$~\cite{rosette-pldi14}, with provenance-based pruning~\cite{wang2018speeding}), 
(3) $\qex$~\cite{qex} (SMT-based, with provenance-based pruning~\cite{wang2018speeding}), 
(4) $\datafiller$~\cite{datafiller-website} (fuzzing-based), 
(5) $\xdata$~\cite{chandra2015data,chandra2019automated} (mutation-based tester), and 
{(6) $\evosql$~\cite{castelein2018search} (search-based tester using random search and genetic algorithms).}
For disambiguation, we include 3 baselines:
(i) the disambiguation component (fuzzing-based) from $\cubes$~\cite{brancas2022cubes},
(ii) a modified version of $\verieql$ that encodes full semantics for multiple (not just two) queries and the disambiguation condition, 
{and (iii) $\datafiller$.}

\newpara{Setup.}
We use the same setup as RQ1 for all baselines (with 1-minute timeout).
For each application, we record the benchmarks that are solved by each tool and the corresponding solving times.

\newpara{Results.}
Figure~\ref{fig:eval:RQ2-results-solved} and Figure~\ref{fig:eval:RQ2-results-time} present our results. 
$\tool$ is a clear winner for both applications.
It disproves significantly more query pairs than the best baseline $\verieql$ ({5,497 vs. 3,177}).
Across all refuted query pairs, $\tool$'s median running time is {0.1 seconds}, which is even faster than $\verieql$ (with a median of {0.4 seconds} across its solved benchmarks). 
For disambiguation, $\tool$ also solves significantly more benchmarks than all baselines. 
It is interesting that $\tool$ solves more D-100 benchmarks than D-50, which is also observed for $\datafiller$. 
In terms of time---see Figure~\ref{fig:eval:RQ2-results-time}(b)---$\tool$ is comparable with $\verieql$.
$\datafiller$ and $\cubes$ are understandably faster (due to their fuzzing-based methods), but they solve far fewer benchmarks (see Figure~\ref{fig:eval:RQ2-results-solved}).

\begin{figure}[!t]
\centering
\begin{minipage}[b]{.48\linewidth}
\centering
\includegraphics[height=4.5cm]{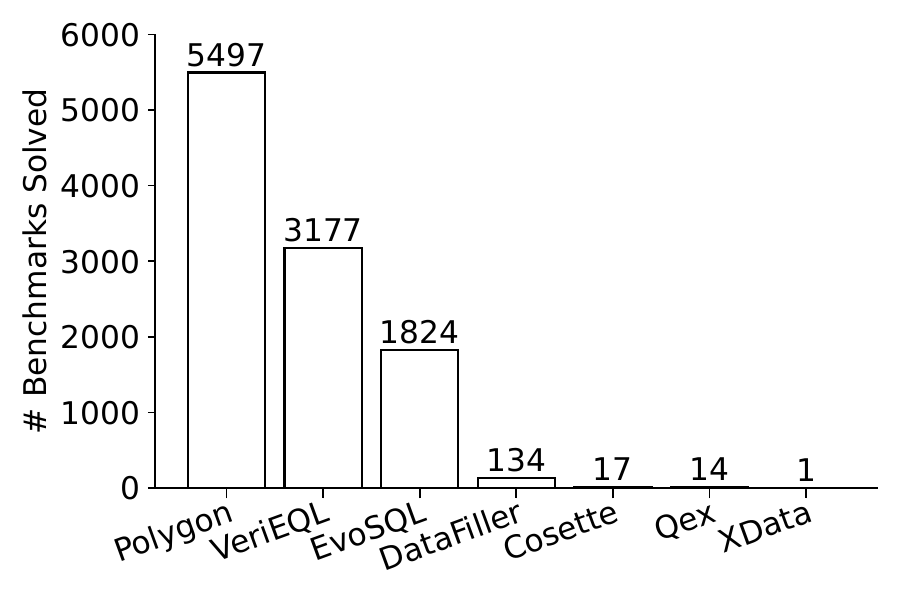}
\captionsetup{skip=3pt} 
\caption*{(a) Equivalence refutation.}
\end{minipage}
\hspace{2pt}
\begin{minipage}[b]{.48\linewidth}
\centering
\includegraphics[height=4.5cm]{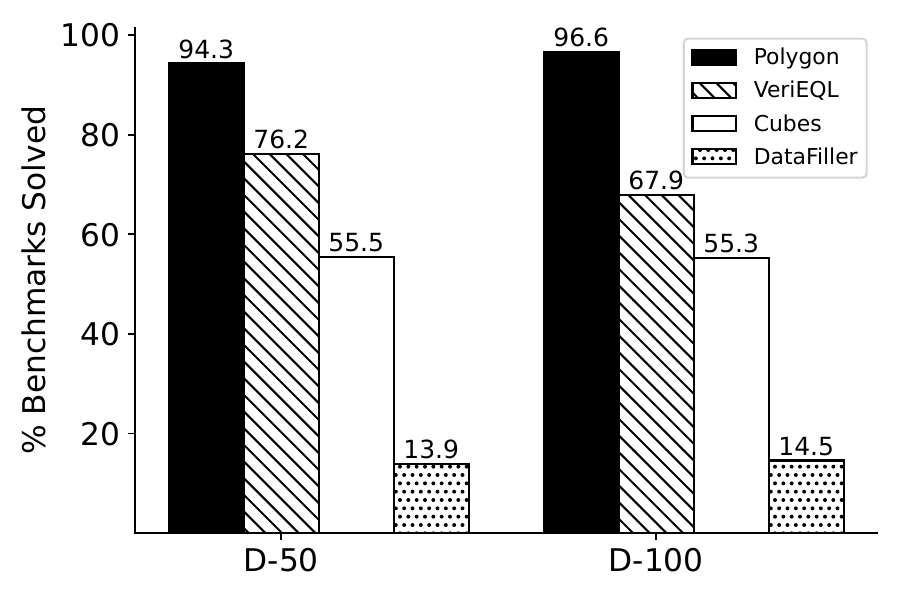}
\captionsetup{skip=3pt} 
\caption*{(b) Disambiguation.}
\end{minipage}
\vspace{-5pt}
\caption{$\tool$ vs. baselines, in terms of benchmarks solved. Note that for disambiguation, we present the \emph{percentage} of benchmarks solved. (Recall: all of our disambiguation benchmarks are solvable by construction).}
\label{fig:eval:RQ2-results-solved}
\end{figure}

\begin{figure}[!t]
\centering
\begin{minipage}[b]{.45\linewidth}
\centering
\includegraphics[height=4.4cm]{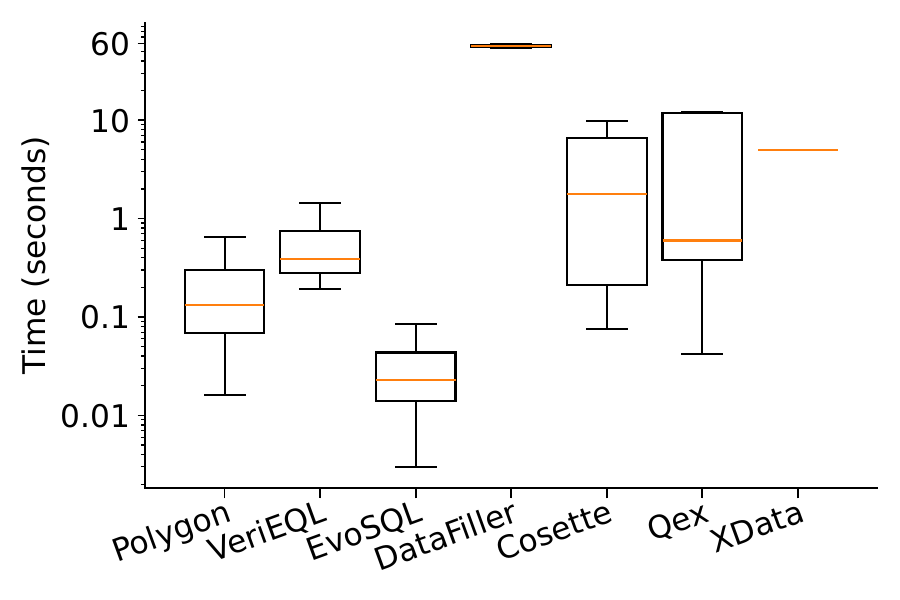}
\captionsetup{skip=3pt} 
\caption*{(a) Equivalence refutation.}
\end{minipage}
\hspace{1pt}
\begin{minipage}[b]{.53\linewidth}
\centering
\includegraphics[height=4.4cm]{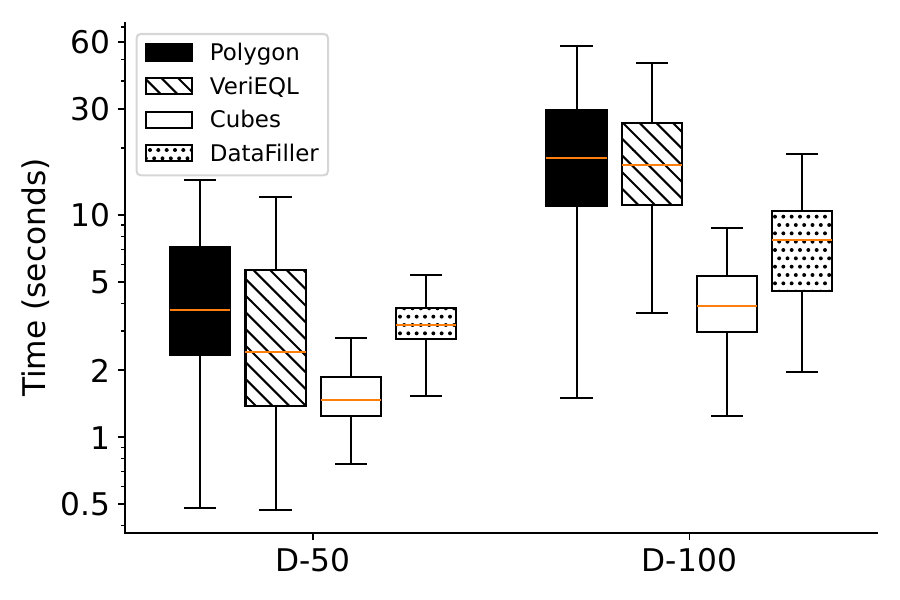}
\captionsetup{skip=3pt} 
\caption*{(b) Disambiguation.}
\end{minipage}
\vspace{-5pt}
\caption{$\tool$ vs. baselines, in terms of solving time. For each tool (including $\tool$ and all baselines), we present the quartile statistics of its solving times across all solved benchmarks.}
\label{fig:eval:RQ2-results-time}
\vspace{-10pt}
\end{figure}

\subsection{RQ3: Ablation Studies}
\label{sec:eval:RQ3}

\evalfinding{All of our design choices (including the top-level conflict-driven search algorithm architecture, the lattice structure of UAs, conflict extraction from unsatisfiability core, conflict accumulation) play an important role for $\tool$'s performance.}

\newpara{Two sets of ablations.}
The first set considers ablations that perform brute-force enumeration over UAs, with the goal of understanding the impact of the top-level architecture of our algorithm. 
Variants in the second set alter lower-level designs, and reuse the same architecture as $\tool$.

\newpara{First set of ablations.}
To be complete, these ablations must search a covering set of UAs for $\underapproxfamily_{\astnode}$, for each AST node $\astnode$. 
We create the following variants that use different covering sets. 
\begin{itemize}[leftmargin=*]
\item 
\textsc{Enum-MinUAs}, 
which considers all \emph{minimal} UAs for each $\astnode$. 
It enumerates all combinations of minimal UAs (one per node), and stops when a satisfying UA map is found.
\item 
\textsc{Enum-TopUAs}, 
which considers only \emph{Top UA} for each $\astnode$; i.e., it encodes the full semantics for $\astnode$. 
\item 
\textsc{Enum-50\%Top}, 
which considers UAs with \emph{half} of the values set to top. 
In particular, for each $\astnode$, we first set 50\% of the values (randomly picked) in $\astnode$'s UA to top. Then, we create all permutations of non-top values for the remaining half, each of which corresponds to a UA for $\astnode$.
\item 
\textsc{Enum-25\%Top} and \textsc{Enum-75\%Top}, which are constructed in the same fashion as \textsc{Enum-50\%Top} but set 25\% and 75\% (respectively, randomly selected) of the values to top. 
\end{itemize}
While still based on UAs, these ablations differ from our algorithm architecturally: they enumerate UA choices for every AST node in a brute-force manner, and return the first working combination. 
The reason we consider five ablations here is to obtain a full range of their performance data.

\newpara{Second set of ablations.}
We create the following variants. 
\begin{itemize}[leftmargin=*]

\item 
\textsc{TopUACover}, where $\splitUAspace$ is changed to return Top UA for each $\astnode_i$ at line~\ref{alg:resolveconflict:loop-begin} of Algorithm~\ref{alg:resolveconflict}.

\item 
\textsc{MinUAsCover}, where $\splitUAspace$ returns all minimal UAs for each $\astnode_i$ at line~\ref{alg:resolveconflict:loop-begin} of Algorithm~\ref{alg:resolveconflict}.

\item 
\textsc{NoNewConflicts}, which removes line~\ref{alg:resolveconflict:check-conflict} from Algorithm~\ref{alg:resolveconflict} and hence does not add new conflicts other than the one at line~\ref{alg:resolveconflict:add-conflict} (which is necessary for termination).

\item 
\textsc{NoUnsatCore}, where $\conflictASTnodes$ includes all nodes from $\mapastnodetounderapprox$ at line~\ref{alg:conflict-driven-UA-search:resolve-conflict} of Algorithm~\ref{alg:conflict-driven-UA-search}. 

\item 
\textsc{AddMinUAs}, which maps each $\astnode_i$ at line~\ref{alg:conflict-driven-UA-search:add-more-nodes} of Algorithm~\ref{alg:conflict-driven-UA-search}
to a minimal UA  (randomly selected from $\underapproxfamily_{\astnode_i}$) and also removes line~\ref{alg:conflict-driven-UA-search:obtain-model-and-update} (which is no longer necessary). 

\end{itemize}

\newpara{Setup.}
We use the same setup as in RQ2 to run all ablations and collect experimental data.

\newpara{Results.}
Figure~\ref{fig:eval:RQ3-results-solved} and Figure~\ref{fig:eval:RQ3-results-time} present our results. 
Let us begin with the first set of ablations. 
Our take-away is that the \textsc{Enum-X} ablations are significantly worse than $\tool$, 
both in terms of benchmarks solved and the solving time---for both equivalence refutation and disambiguation. 
This underscores the advantage of $\tool$'s algorithm architecture. 
On the other hand, ablations from the second set are on par with $\tool$ (except for $\textsc{MinUAsCover}$ which performs poorly) on ER and D-50: they solve slightly fewer benchmarks than $\tool$, using slightly more time. 
However, they become significantly worse on D-100 (with more queries to disambiguate). 
This highlights the importance of our various design choices---e.g., the lattice structure of UAs which allows $\splitUAspace$ to partition the search space---especially for harder problems (e.g., those involving more queries).

\begin{figure}[!t]
\centering
\begin{minipage}[b]{.35\linewidth}
\centering
\includegraphics[height=4.2cm]{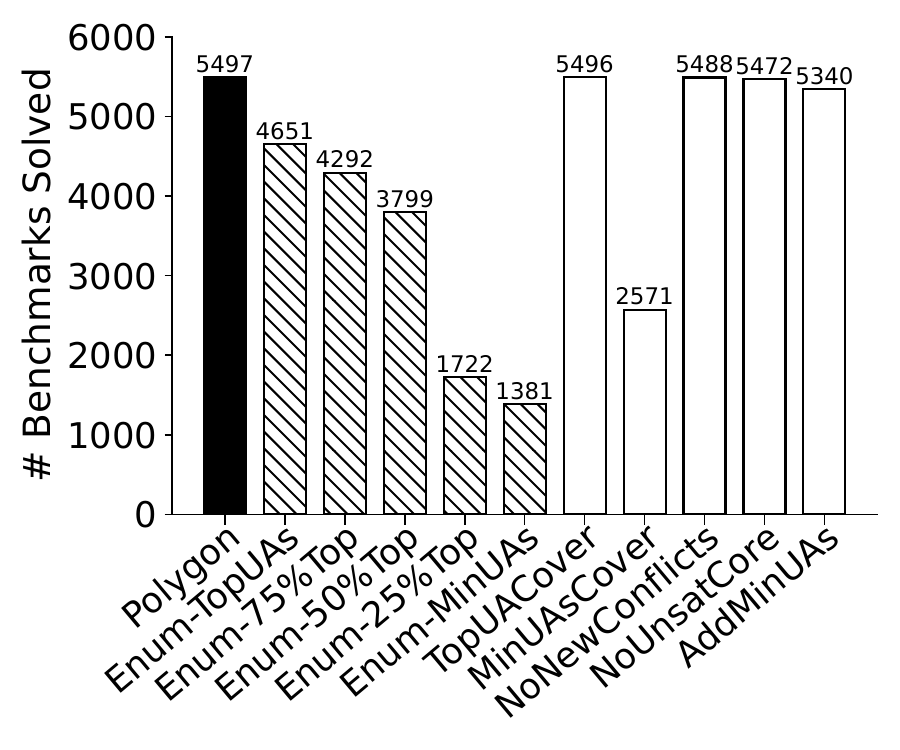}
\captionsetup{skip=3pt} 
\caption*{(a) Equivalence refutation.}
\end{minipage}
\hspace{1pt}
\begin{minipage}[b]{.63\linewidth}
\centering
\includegraphics[height=4.2cm]{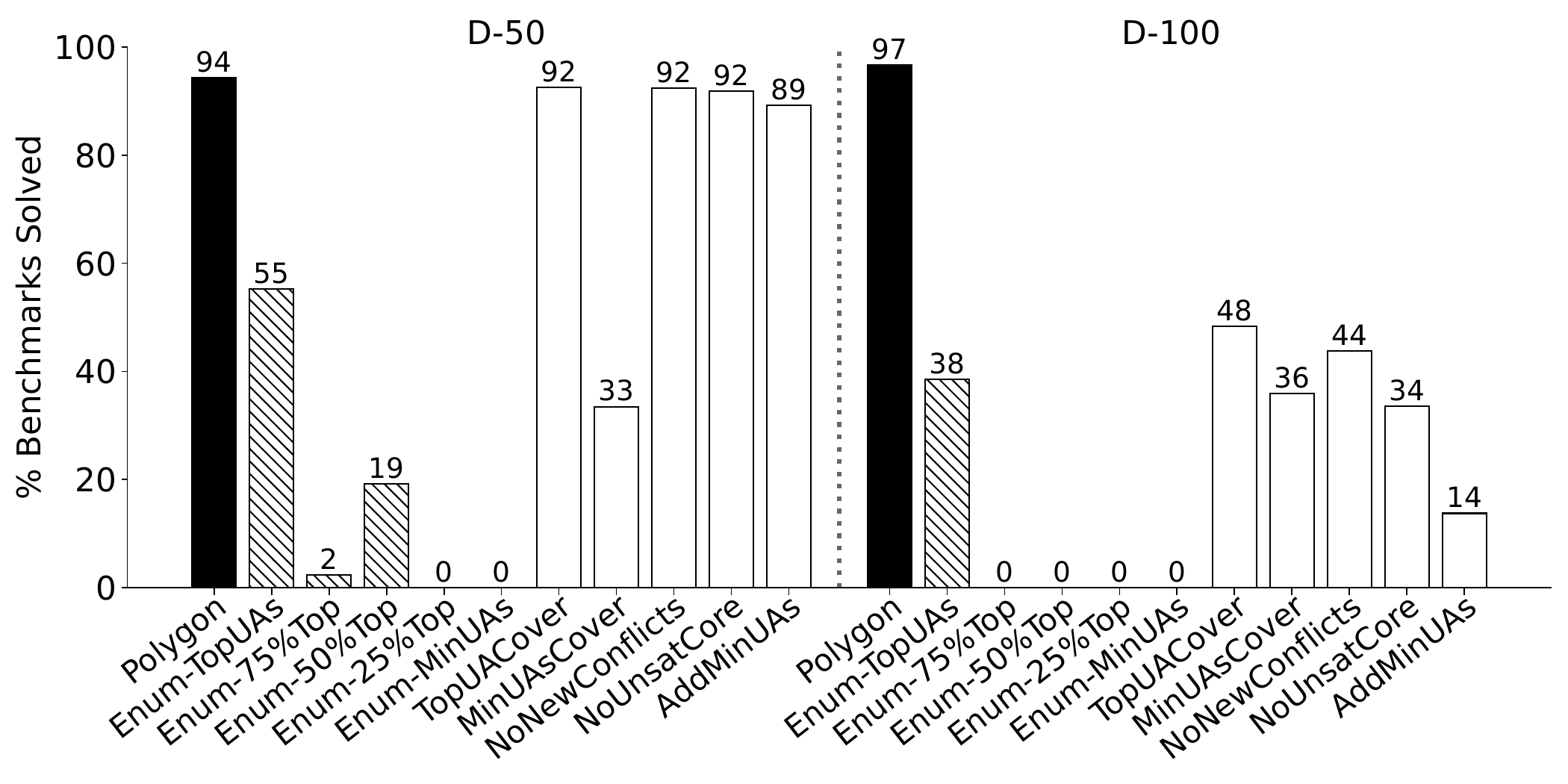}
\captionsetup{skip=3pt} 
\caption*{(b) Disambiguation.}
\end{minipage}
\vspace{-5pt}
\caption{$\tool$ vs. ablations, in terms of benchmarks solved.}
\label{fig:eval:RQ3-results-solved}
\end{figure}

\begin{figure}[!t]
\centering
\begin{minipage}[b]{.35\linewidth}
\centering
\includegraphics[height=4.2cm]{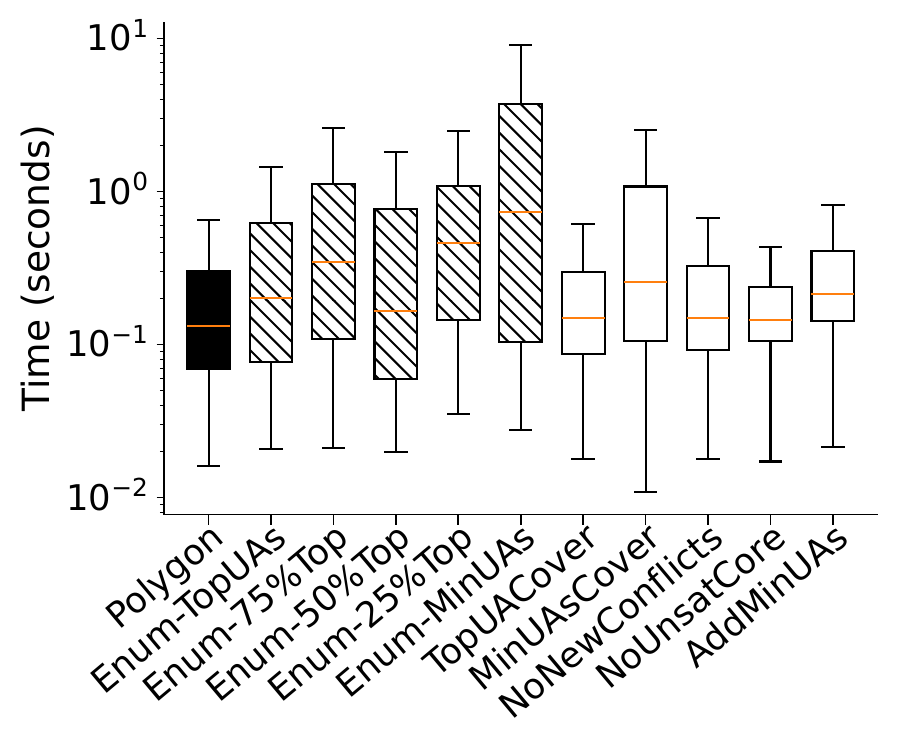}
\captionsetup{skip=3pt} 
\caption*{(a) Equivalence refutation.}
\end{minipage}
\hspace{1pt}
\begin{minipage}[b]{.63\linewidth}
\centering
\includegraphics[height=4.2cm]{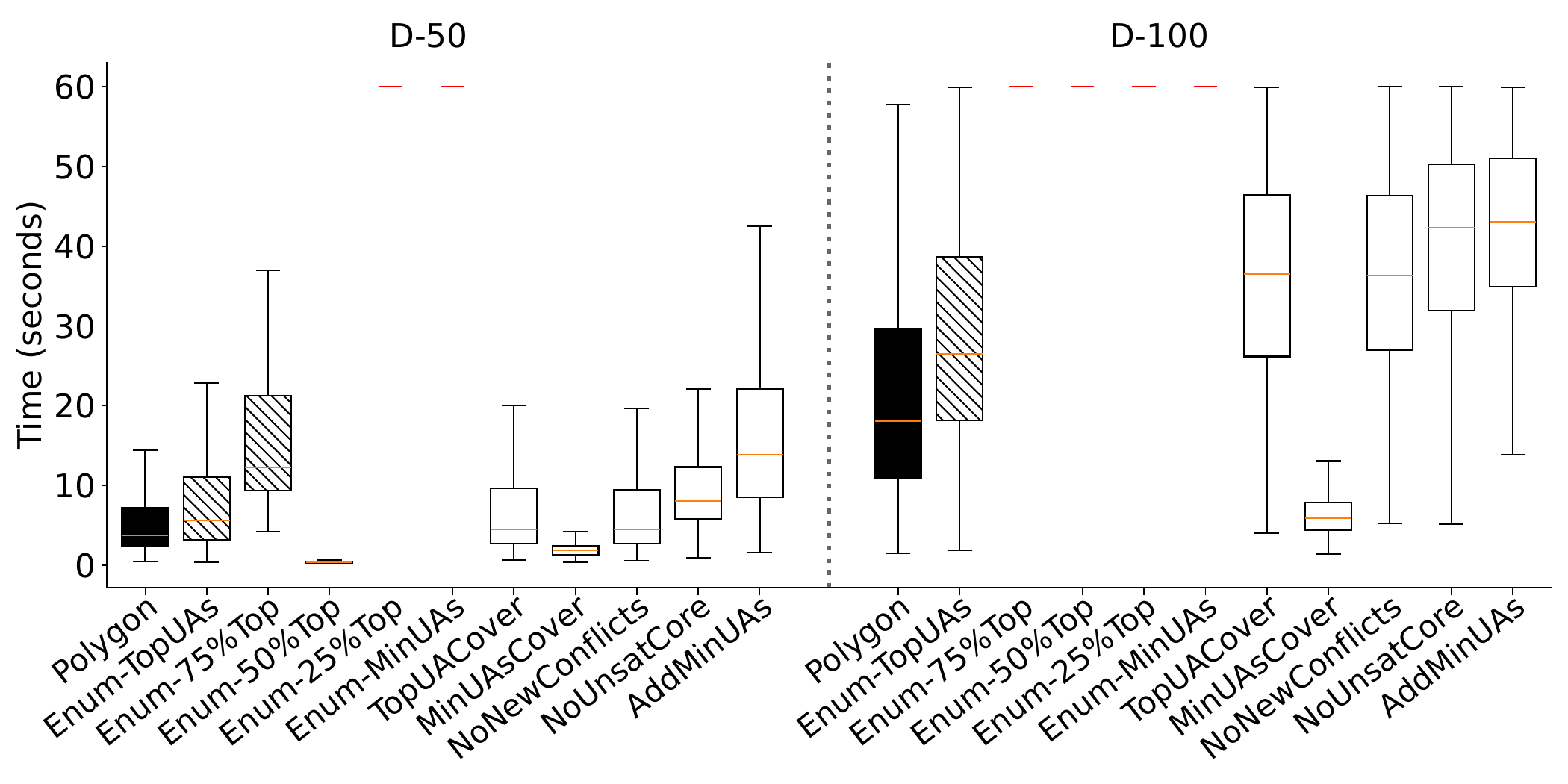}
\captionsetup{skip=3pt} 
\caption*{(b) Disambiguation.}
\end{minipage}
\vspace{-5pt}
\caption{$\tool$ vs. ablations, in terms of solving time. A red bar at the top means ``timeout on all benchmarks''.}
\label{fig:eval:RQ3-results-time}
\vspace{-10pt}
\end{figure}

%% file: sections/related.tex
\section{Related Work}\label{sec:related}

This section briefly discusses some closely related work.

\newpara{Under-approximate reasoning.}
Our work is inspired by O'Hearn's seminal work on incorrectness logic~\cite{o2019incorrectness} and a long line of works that leverage under-approximate reasoning for various tasks, such as 
proving non-termination~\cite{raad2024non}, 
detecting memory errors~\cite{le2022finding}, reasoning about concurrency~\cite{blackshear2018racerd,raad2023general}, scaling static analysis~\cite{distefano2019scaling}, dynamic symbolic execution~\cite{godefroid2007compositional},
among others~\cite{ball2005abstraction,puasuareanu2005concrete,kroening2015under,raad2020local,murray2020under,dardinier2024hyper,gorogiannis2019true}. 
$\tool$ can be viewed as a successful use of under-approximate reasoning to generate test inputs for SQL---in particular, for equivalence refutation and disambiguation. 
Building upon the literature, we contribute a compositional approach (based on SMT) to define a family of under-approximations per SQL operator, and a fast algorithm to search within this space for desired under-approximations.

\newpara{Symbolic execution.}
Our work can be viewed as a form of (backward) symbolic execution~\cite{baldoni2018survey,chandra2009snugglebug,chalupa2021backward}: it begins with the application condition $\appcond$ (i.e., an assertion) and under-approximates the semantics (akin to picking execution paths) of AST nodes top-down (i.e., backwards). 
We perform backtracking when the current UA map does not meet $\appcond$---but in a different fashion from prior work---by analyzing a subset of UAs that are in conflict. 
In particular, we utilize our lattice structure of UAs to search many UAs at the same time, which allows us to efficiently find a fix for the conflict. 
During this process, we discover and block additional conflicts, to speed up future search. 
This is related to, but different from, prior pruning techniques, such as those based on interpolation~\cite{mcmillan2010lazy,jaffar2013boosting} and   detecting inconsistent code~\cite{schwartz2015conflict}.
Our work is also related to summary-based symbolic execution~\cite{sery2011interpolation,alt2017hifrog,sery2012incremental}---especially compositional dynamic symbolic execution~\cite{anand2008demand,sen2015multise,godefroid2007compositional}---in the sense that we also utilize (under-approximate) function summaries. Our contribution is a novel compositional method to define a lattice of summaries (per SQL operator), which allows us to perform combinatorial search efficiently. 
Non-minimal UAs essentially correspond to state merging~\cite{godefroid2007compositional,kuznetsov2012efficient,porncharoenwase2022formal,lu2023grisette}, and are used locally (when analyzing conflicts or adding new AST nodes) to not overly stress the SMT solver.

\newpara{Symbolic reasoning for SQL.}
$\tool$ is especially related to prior symbolic reasoning techniques that are tailored towards SQL~\cite{wang2018speeding,he2024verieql,qex,chu2017cosette,chu2017hottsql,chu2018axiomatic,cheung2023towards} (all of which we compare with in our evaluation).
Unlike these works, $\tool$ under-approximates SQL semantics and performs reasoning while searching under-approximations. 
There is a long line of work on SQL equivalence verification~\cite{chu2017hottsql,chu2018axiomatic,green2009containment,zhou2022spes,zhou2019automated}. 
$\tool$ focuses on generating inputs to satisfy properties of query outputs,  including but not limited to checking (non-)equivalence. \xinyurevision{Alloy~\cite{torlak2007kodkod,jackson2012software,jackson2002alloy} represents another line of related work, which natively supports relational operators and hence in principle can be used to reason about SQL. 
The key distinction of our work lies in the granularity of UAs and the capability of performing search over them. In contrast, it is unclear if Alloy's ``scope'' mechanism is as flexible as our UAs. More importantly, Alloy's scope is always fixed beforehand and cannot be dynamically changed during analysis.}

\newpara{Test data generation for SQL.}
$\tool$ is also closely related to works on testing SQL queries. 
For instance, $\xdata$~\cite{chandra2019automated,chandra2015data} is a mutation-based tester to detect common SQL mistakes. $\datafiller$~\cite{datafiller-website} is a fuzzer that generates random test inputs, given the database schema. 
$\evosql$~\cite{castelein2018search} generates test data via an evolutionary search algorithm, guided by predicate coverage~\cite{tuya2010full}. 
Different from these approaches, we incorporate SQL semantics to generate satisfying inputs for a given property. 
$\tool$ is also related to property-based testing~\cite{claessen2000quickcheck,lampropoulos2019coverage,lampropoulos2018quickchick,lampropoulos2017generating,lampropoulos2017beginner,padhye2019jqf,zhou2023covering,goldstein2024property}, in the sense that we aim to generate input databases for a given property over SQL queries. 
Our generation process, however, uses SMT-based under-approximate reasoning over both the property and the queries, rather than some form of random input generation (as in many prior works) or enumeration~\cite{runciman2008smallcheck,matela2017tools}.

\newpara{Conflict-driven search.}
The idea of conflict-driven search has received success in multiple areas. 
For example, modern constraint solvers use conflict-driven clause learning to derive new clauses for faster boolean satisfiability solving~\cite{cdcl-book21}.
In program synthesis, a candidate program that fails to meet the specification can be viewed as a conflict. Various algorithms~\cite{neo-pldi18,trinity-vldb19,Wang20dynamite,mobius-oopsla23} have been proposed to generalize such a conflict to unseen programs that would fail due to the same reason. 
Our work is distinct in a few ways: 
we perform conflict-driven search over under-approximations, 
our conflict is generalized to new ones in a way that takes advantage of a predefined lattice structure of UAs, 
and we aim to generate inputs for SQL queries to meet a given property. 

%% file: sections/conc.tex
\section{Conclusion and Discussion}
\label{sec:conc}

\xinyurevision{
This paper presented a new method based on under-approximation search to perform symbolic reasoning for SQL. Our evaluation demonstrated significant performance boost over all state-of-the-art techniques, for two reasoning tasks (namely SQL equivalence refutation and disambiguation).}

While this work is largely focused on SQL, we believe the underlying principles have the potential to generalize to other languages and domains.
A fundamental assumption is: analyzing an under-approximation (UA) of a program is cheap---which we believe holds true in general. 
Then, we need to curate a family of UAs for the language. 
In the case of SQL, we were able to do this compositionally. 
We believe this is also possible, in general, for programs representable using a loop-free composition of blocks (like an AST). 
Finally, our search technique (Algorithms~\ref{alg:conflict-driven-UA-search} and~\ref{alg:resolveconflict}) does not assume SQL. 
It, however, assumes a lattice of UAs, which we believe is definable for other languages.

\xinyurevision{
To explore the full generality of this idea, one interesting future direction is to build the idea on top of the Rosette solver-aided programming language~\cite{torlak2013growing}. 
For instance, one approach is to build a Rosette-based symbolic interpreter that is parameterized with a UA. Then, given $n$ programs, it can perform symbolic execution to check against the given application condition with respect to the given UA. 
On top of this symbolic interpreter, we can implement the UA search algorithm. 
}

%% file: sections/ack.tex
\section*{Acknowledgments}
\label{sec:ack}

We would like to thank the PLDI anonymous reviewers for their insightful feedback. 
We thank Zheng Guo, Chenglong Wang, and Wenxi Wang for their feedback on earlier drafts of this work. 
We would also like to thank Danny Ding for helping with a baseline in the evaluation, Xiaomeng Xu and Yuxuan Zhu for their contributions to earlier versions of this work, and Brian Zhang for helping process some of the benchmarks. 
This research is supported by the National Science Foundation under Grant Numbers 
CCF-2210832, 
CCF-2318937,  
CCF-2236233, 
and CCF-2123654, 
as well as an NSERC Discovery Grant.

%% file: sections/data-availability-statement.tex
\section*{Artifact Availability Statement}\label{sec:das}

The artifact that implements the techniques and supports the evaluation results reported in this paper is available on Zenodo~\cite{artifact}.

%% file: sections/appendix/ua-others.tex
\section{Under-Approximations for Other Operators} \label{sec:ua-others}

\newpara{Cartesian Product.}
The UAs for $\product$ is similar to the UAs for the inner join $\ijoin_{\predicate}$, where its UA $\underapprox$ is also an $n_1 \times n_2$ matrix: $n_1$ (resp. $n_2$) is the maximum size of the first (resp. second) table, and each $\underapprox_{i, j}$ is either $\truevalue, \falsevalue$ or $\unknownvalue$.  Different from the inner join operator, a cartesian product doesn't have a join predicate $\predicate$.  So for each $\underapprox_{i, j}$, $\truevalue$ simply means the $i$th tuple $\tuple_i$ from the first table and the $j$th tuple $\tuple_j$ from the second are both present, whereas $\falsevalue$ means at least one of $\tuple_i$, $\tuple_j$ is deleted.

\newpara{Outer Joins.}
Outer joins $\ljoin_{\predicate}$, $\rjoin_{\predicate}$, and $\fjoin_{\predicate}$ have the same UA as inner join's.  Specifically, their UA $\underapprox$ is an $n_1 \times n_2$  matrix, where $\underapprox_{i, j}$ is $\truevalue, \falsevalue$ or $\unknownvalue$, and $n_1$ (resp. $n_2$) is the maximum size of the first (resp. second) table.
For each $\underapprox_{i, j}$, $\truevalue$  means the $i$th tuple $\tuple_i$ from the first table and the $j$th tuple $\tuple_j$ from the second both exist and satisfy the join condition $\predicate$; whereas $\falsevalue$ means at least one of $\tuple_i$, $\tuple_j$ is deleted, or they are both present but do not satisfy $\predicate$. 

\newpara{OrderBy.}
The UAs for $\orderby_{\expression}$ is different from all other operators. 
While its $\underapprox$ is still a vector of $n$ values and $n$ is the maximum size of the input table, each $\underapprox_i$ now is chosen from $\{ \falsevalue, \truevalue_1, \mydots, \truevalue_n, \unknownvalue \}$. 
If $\underapprox_i = \falsevalue$, it means the $i$th tuple $\tuple_i$ is not present in the input table. 
On the other hand, $\underapprox_i = \truevalue_k$ means tuple $\tuple_i$ has rank $k$ among all tuples according to expression $\expression$. 

%% file: sections/appendix/encoding.tex
\section{Encoding Full Semantics} \label{sec:all-encoding}

Figures~\ref{fig:under-semantics} and~\ref{fig:under-semantics-b} present the encoding of full semantics for all query operators from Figure ~\ref{fig:sql-syntax}.  Specifically, Figure ~\ref{fig:under-semantics} shows the encoding for $\filter$, $\proj$, $\product$, and $\ijoin$, and the other operators are formalized in ~\ref{fig:under-semantics-b}.

\input{figures/fig-under-semantics-all}

\input{figures/fig-under-semantics-all-b}

Note that we do not have rules for $\relation$, $\rename$, and $\with$.  This is because for $\relation$ and $\rename$, the output is simply the same as input; an operator $\with(\vec{\sqlquery}, \vec{\relation}, \sqlquery)$ does not require a separate rule and is handled by first encoding the queries in $\vec{\sqlquery}$ and mapping the corresponding results to $\vec{\relation}$, and the result of $\sqlquery$ is returned.


%% file: figures/fig-under-semantics-all.tex
\begin{figure}[!h]
\centering
\scriptsize
\[
\hspace{-10pt}
\arraycolsep=1pt\def\arraystretch{1}
\begin{array}{ll}

(1) & 
\irule
{
\begin{array}{l}

\queryopinput = [\tuple_1, \mydots, \tuple_n] 

\quad

\queryopoutput = [\tuple'_1, \mydots, \tuple'_n]

\quad

\underapproxvar = [ \underapproxvar_1, \mydots, \underapproxvar_n ]

\quad

\attrlist = \tableattrs(\queryopinput)

\\

\queryopencoding_{i, \truevalue} = 
( \underapproxvar_i = \truevalue ) 
\to 
\big(
\cmdExist(\tuple_i, \tuple'_i, \denot{\predicate}_{\tuple'_i}) \land \cmdCopy(\tuple_i, \tuple'_i, \attrlist) 
\big)

\quad

\queryopencoding_{i, \falsevalue} = 
( 
\underapproxvar_i = \falsevalue ) 
\to 
\cmdCondDel(\tuple_i, \tuple'_i, \denot{\predicate}_{\tuple'_i}
)

\end{array}
}
{
\encodefullsemantics 
\big( 
\filter_{\predicate}
\big) 
\ruleleadsto 
\bigland_{i = 1, \mydots, n}
\queryopencoding_{i, \truevalue} 
\land 
\queryopencoding_{i, \falsevalue}
}

\\ \\ 

(2) & 
\irule
{
\begin{array}{l}

\queryopinput = [\tuple_1, \mydots, \tuple_n] 
\quad
\queryopoutput = [\tuple'_{1}, \mydots, \tuple'_n] 
\quad
\underapproxvar = [ \underapproxvar_1, \mydots, \underapproxvar_n ] 

\\

\queryopencoding_{i, \truevalue} = 
( \underapproxvar_i = \truevalue ) 
\to 
\big( 
\cmdExist(\tuple_i, \tuple'_i, \truepredicate) 
\land
\cmdCopy(\tuple_i, \tuple'_i, \attrlist) 
\big)

\quad 

\queryopencoding_{i, \falsevalue} = 
( \underapproxvar_i = \falsevalue ) 
\to 
\cmdCondDel(\tuple_i, \tuple'_i, \truepredicate)

\end{array}
}
{
\encodefullsemantics \big( \proj_\attrlist \big)  
\ruleleadsto 
\bigland_{i = 1, \mydots, n}
\queryopencoding_{i, \truevalue} 
\land 
\queryopencoding_{i, \falsevalue}
}

\\ \\

(3) & 
\irule
{
\begin{array}{l}

\queryopinput_1 = [ \tuple_1, \mydots, \tuple_{n_1} ] 
\ \ 
\queryopinput_2 = [ \tuple'_1, \mydots, \tuple'_{n_2} ] 
\ \ 
\queryopoutput = [ \tuple\doubleprime_{1,1}, \mydots, \tuple\doubleprime_{n_1, n_2} ] 

\ \ 

\underapproxvar = \big[ [ \underapproxvar_{1, 1}, \mydots \underapproxvar_{1, n_2} ], \mydots, [ \underapproxvar_{n_1, 1 }, \mydots, \underapproxvar_{n_1, n_2}] \big]

\ \ 

\attrlist_i =  \tableattrs(\queryopinput_i) 

\\[2pt]

\queryopencoding_{i, j, \truevalue} = 
( \underapproxvar_{i,j} = \truevalue ) 
\to 
\cmdExist 
\Big( 
( \tuple_i, \tuple'_j ), \tuple\doubleprime_{i, j}, \top 
\Big) 
\land 
\cmdCopy( \tuple_i, \tuple\doubleprime_{i,j}, \attrlist_1 ) 
\land 
\cmdCopy( \tuple'_j, \tuple\doubleprime_{i,j}, \attrlist_2 )

\\ 

\queryopencoding_{i, j, \falsevalue} = 
( \underapproxvar_{i,j} = \falsevalue )  
\to 
\cmdCondDel 
\Big( 
( \tuple_i, \tuple'_j ),
\tuple\doubleprime_{i,j}, 
\top 
\Big)

\end{array}
}
{
\encodefullsemantics \big( \product \big) 
\ruleleadsto 
\bigland_{i = 1, \mydots, n_1} 
\bigland_{ j = 1, \mydots, n_2}
\queryopencoding_{i, j, \truevalue}
\land 
\queryopencoding_{i, j, \falsevalue}
}

\\ \\

(4) & 
\irule
{
\begin{array}{l}

\queryopinput_1 = [ \tuple_1, \mydots, \tuple_{n_1} ] 
\ \ 
\queryopinput_2 = [ \tuple'_1, \mydots, \tuple'_{n_2} ] 
\ \ 
\queryopoutput = [ \tuple\doubleprime_{1,1}, \mydots, \tuple\doubleprime_{n_1, n_2} ] 

\ \ 

\underapproxvar = \big[ [ \underapproxvar_{1, 1}, \mydots \underapproxvar_{1, n_2} ], \mydots, [ \underapproxvar_{n_1, 1 }, \mydots, \underapproxvar_{n_1, n_2}] \big]

\ \ 

\attrlist_i =  \tableattrs(\queryopinput_i) 

\\[2pt]

\queryopencoding_{i, j, \truevalue} = 
( \underapproxvar_{i,j} = \truevalue ) 
\to 
\cmdExist 
\Big( 
( \tuple_i, \tuple'_j ), \tuple\doubleprime_{i, j}, \denot{\predicate}_{\tuple_i, \tuple'_j} 
\Big) 
\land 
\cmdCopy( \tuple_i, \tuple\doubleprime_{i,j}, \attrlist_1 ) 
\land 
\cmdCopy( \tuple'_j, \tuple\doubleprime_{i,j}, \attrlist_2 )

\\ 

\queryopencoding_{i, j, \falsevalue} = 
( \underapproxvar_{i,j} = \falsevalue )  
\to 
\cmdCondDel 
\Big( 
( \tuple_i, \tuple'_j ),
\tuple\doubleprime_{i,j}, 
\denot{\predicate}_{\tuple_i, \tuple'_j} 
\Big)

\end{array}
}
{
\encodefullsemantics \big( \ijoin_\predicate \big) 
\ruleleadsto 
\bigland_{i = 1, \mydots, n_1} 
\bigland_{ j = 1, \mydots, n_2}
\queryopencoding_{i, j, \truevalue}
\land 
\queryopencoding_{i, j, \falsevalue}
}

\end{array}
\]
\caption{Full semantics for query operators from Figure ~\ref{fig:sql-syntax}.  An auxiliary function $\cmdNullTuples$ creates and determines the existence of a null tuple $\tuple_\omega$, based on $\underapproxvar$, where attributes in $\attrlist$ are copied from the tuple $\tuple$, while attributes in $\attrlist_\omega$ are set to $\nullv$, formally defined as $\cmdNullTuples(\tuple_\omega, \tuple, \attrlist, \attrlist_\omega, \vec{\underapproxvar}) = \cmdCopy(\tuple, \tuple_\omega, \attrlist) \land \Big(\bigwedge\limits_{a \in \attrlist_\omega} \denot{\tuple_\omega.a} = \nullv \Big) \land (\land_{\underapproxvar \in \vec{\underapproxvar}} \ \underapproxvar = \falsevalue) \leftrightarrow \neg \del(\tuple_\omega)$. }
\label{fig:under-semantics}
\end{figure}

%% file: figures/fig-under-semantics-all-b.tex
\begin{figure}[!h]
\centering
\scriptsize
\[
\hspace{-10pt}
\arraycolsep=1pt\def\arraystretch{1}
\begin{array}{ll}

\\ \\

(5) & 
\irule
{
\begin{array}{l}

\queryopinput_1 = [ \tuple_1, \mydots, \tuple_{n_1} ] 
\ \ 
\queryopinput_2 = [ \tuple'_1, \mydots, \tuple'_{n_2} ] 
\ \ 
\queryopoutput = [ \tuple\doubleprime_{1,1}, \mydots, \tuple\doubleprime_{n_1, n_2}, \tuple\doubleprime_{1, n_2+1}, \mydots, \tuple\doubleprime_{n_1, n_2+1} ] \\


\underapproxvar = \big[ [ \underapproxvar_{1, 1}, \mydots \underapproxvar_{1, n_2} ], \mydots, [ \underapproxvar_{n_1, 1 }, \mydots, \underapproxvar_{n_1, n_2}] \big]

\ \ 

\attrlist_i =  \tableattrs(\queryopinput_i) 

\\[2pt]

\queryopencoding_{i, j, \truevalue} = 
( \underapproxvar_{i,j} = \truevalue ) 
\to 
\cmdExist 
\Big( 
( \tuple_i, \tuple'_j ), \tuple\doubleprime_{i, j}, \denot{\predicate}_{\tuple_i, \tuple'_j} 
\Big) 
\land 
\cmdCopy( \tuple_i, \tuple\doubleprime_{i,j}, \attrlist_1 ) 
\land 
\cmdCopy( \tuple'_j, \tuple\doubleprime_{i,j}, \attrlist_2 )

\\

\queryopencoding_{i, j, \falsevalue} = 
( \underapproxvar_{i,j} = \falsevalue )  
\to 
\cmdCondDel 
\Big( 
( \tuple_i, \tuple'_j ),
\tuple\doubleprime_{i,j}, 
\denot{\predicate}_{\tuple_i, \tuple'_j} 
\Big)

\quad

\queryopencoding_{i, \nullv} = 

\cmdNullTuples(\tuple''_{i, n_2+1}, \tuple_i, \attrlist_1, \attrlist_2, \bigcup_{j \in [1,n_2]}\underapproxvar_{i,j} ) \\

\end{array}
}
{
\encodefullsemantics \big( \ljoin_\predicate \big) 
\ruleleadsto 
\big(
\bigland_{i = 1, \mydots, n_1} 
\bigland_{ j = 1, \mydots, n_2}
\queryopencoding_{i, j, \truevalue}
\land 
\queryopencoding_{i, j, \falsevalue}
\big)
\land
\big(
\bigland_{i = 1, \mydots, n_1} 
\queryopencoding_{i, \nullv}
\big)
}

\\ \\

(6) & 
\irule
{
\begin{array}{l}

\queryopinput_1 = [ \tuple_1, \mydots, \tuple_{n_1} ] 
\ \ 
\queryopinput_2 = [ \tuple'_1, \mydots, \tuple'_{n_2} ] 
\ \ 
\queryopoutput = [ \tuple\doubleprime_{1,1}, \mydots, \tuple\doubleprime_{n_1, n_2}, \tuple\doubleprime_{n_1 + 1, 1}, \mydots, \tuple\doubleprime_{n_1 + 1, n_2} ] \\


\underapproxvar = \big[ [ \underapproxvar_{1, 1}, \mydots \underapproxvar_{1, n_2} ], \mydots, [ \underapproxvar_{n_1, 1 }, \mydots, \underapproxvar_{n_1, n_2}] \big]

\ \ 

\attrlist_i =  \tableattrs(\queryopinput_i) 

\\[2pt]

\queryopencoding_{i, j, \truevalue} = 
( \underapproxvar_{i,j} = \truevalue ) 
\to 
\cmdExist 
\Big( 
( \tuple_i, \tuple'_j ), \tuple\doubleprime_{i, j}, \denot{\predicate}_{\tuple_i, \tuple'_j} 
\Big) 
\land 
\cmdCopy( \tuple_i, \tuple\doubleprime_{i,j}, \attrlist_1 ) 
\land 
\cmdCopy( \tuple'_j, \tuple\doubleprime_{i,j}, \attrlist_2 )

\\

\queryopencoding_{i, j, \falsevalue} = 
( \underapproxvar_{i,j} = \falsevalue )  
\to 
\cmdCondDel 
\Big( 
( \tuple_i, \tuple'_j ),
\tuple\doubleprime_{i,j}, 
\denot{\predicate}_{\tuple_i, \tuple'_j} 
\Big)

\quad

\queryopencoding_{j, \nullv} = 

\cmdNullTuples(\tuple''_{n_1+1, j}, \tuple'_j, \attrlist_2, \attrlist_1, \bigcup_{i \in [1,n_1]}\underapproxvar_{i,j} ) \\

\end{array}
}
{
\encodefullsemantics \big( \rjoin_\predicate \big) 
\ruleleadsto 
\big(
\bigland_{i = 1, \mydots, n_1} 
\bigland_{ j = 1, \mydots, n_2}
\queryopencoding_{i, j, \truevalue}
\land 
\queryopencoding_{i, j, \falsevalue}
\big)
\land
\big(
\bigland_{j = 1, \mydots, n_2} 
\queryopencoding_{j, \nullv}
\big)
}

\\ \\

(7) & 
\irule
{
\begin{array}{l}

\queryopinput_1 = [ \tuple_1, \mydots, \tuple_{n_1} ] 
\ \ 
\queryopinput_2 = [ \tuple'_1, \mydots, \tuple'_{n_2} ] 
\ \ 
\queryopoutput = [ \tuple\doubleprime_{1,1}, \mydots, \tuple\doubleprime_{n_1, n_2}, \tuple\doubleprime_{1, n_2+1}, \mydots, \tuple\doubleprime_{n_1, n_2+1}, \tuple\doubleprime_{n_1 + 1, 1}, \mydots, \tuple\doubleprime_{n_1 + 1, n_2} ] \\


\underapproxvar = \big[ [ \underapproxvar_{1, 1}, \mydots \underapproxvar_{1, n_2} ], \mydots, [ \underapproxvar_{n_1, 1 }, \mydots, \underapproxvar_{n_1, n_2}] \big]

\ \ 

\attrlist_i =  \tableattrs(\queryopinput_i) 

\\[2pt]

\queryopencoding_{i, j, \truevalue} = 
( \underapproxvar_{i,j} = \truevalue ) 
\to 
\cmdExist 
\Big( 
( \tuple_i, \tuple'_j ), \tuple\doubleprime_{i, j}, \denot{\predicate}_{\tuple_i, \tuple'_j} 
\Big) 
\land 
\cmdCopy( \tuple_i, \tuple\doubleprime_{i,j}, \attrlist_1 ) 
\land 
\cmdCopy( \tuple'_j, \tuple\doubleprime_{i,j}, \attrlist_2 )

\\

\queryopencoding_{i, j, \falsevalue} = 
( \underapproxvar_{i,j} = \falsevalue )  
\to 
\cmdCondDel 
\Big( 
( \tuple_i, \tuple'_j ),
\tuple\doubleprime_{i,j}, 
\denot{\predicate}_{\tuple_i, \tuple'_j} 
\Big)

\\[2pt]

\queryopencoding_{i, \nullv} = 

\cmdNullTuples(\tuple''_{i, n_2+1}, \tuple_i, \attrlist_1, \attrlist_2, \bigcup_{j \in [1,n_2]}\underapproxvar_{i,j} )

\quad

\queryopencoding'_{j, \nullv} = 

\cmdNullTuples(\tuple''_{n_1+1, j}, \tuple'_j, \attrlist_2, \attrlist_1, \bigcup_{i \in [1,n_1]}\underapproxvar_{i,j} )
\\

\end{array}
}
{
\encodefullsemantics \big( \fjoin_\predicate \big) 
\ruleleadsto 
\big(
\bigland_{i = 1, \mydots, n_1} 
\bigland_{ j = 1, \mydots, n_2}
\queryopencoding_{i, j, \truevalue}
\land 
\queryopencoding_{i, j, \falsevalue}
\big)

\land

\big(
\bigland_{i = 1, \mydots, n_1} 
\queryopencoding_{i, \nullv}
\big)

\land

\big(
\bigland_{j = 1, \mydots, n_2} 
\queryopencoding'_{j, \nullv}
\big)
}

\\ \\

(8) & 
\irule
{
\begin{array}{l}

\queryopinput = [\tuple_1, \mydots, \tuple_n] 

\quad

\queryopoutput = [\tuple'_{1}, \mydots, \tuple'_n] 

\quad

\underapproxvar = [ \underapproxvar_1, \mydots, \underapproxvar_n ]

\\[2pt]

\queryopencoding_{i, \falsevalue} =

( \underapproxvar_i = \falsevalue ) 

\to

\Big(

\big( 
\del( \tuple_i ) 
\lor 

\Big( 
\neg \del( \tuple_i )

\land

\biglor_{j=1, \mydots, i-1} 
\big( 
\neg \del(\tuple_j) 
\land 
\bigland_{a \in \vec{\expression}} \denot{\tuple_i.a} = \denot{\tuple_j.a} 
\land 
\group(\tuple_i) = j
\big) 

\Big) 
\big)

\land 

\del( \tuple'_i )

\Big)

\\

\queryopencoding_{i, \neg \falsevalue} = 
\neg \del(\tuple_i) 
\land 
\neg \biglor_{j=1, \mydots, i-1} 
\Big( 
\neg \del(\tuple_j) \land \bigland_{a \in \vec{\expression}}\denot{\tuple_i.a} = \denot{\tuple_j.a}
\Big) 
\land 
\group(\tuple_i) = i

\\

\queryopencoding_{i, \truevalue} = 
( \underapproxvar_i = \truevalue ) 
\to 
\Big( 
\queryopencoding_{i, \neg \falsevalue}
\land 
\neg \denot{\predicate}_{\group^{-1}(i)} 
\land 
\del(\tuple'_i)
\Big)

\quad

\queryopencoding_{i, \truevalue_{\predicate}} = 
( \underapproxvar_i = \truevalue_{\predicate} ) 
\to 
\Big( 
\queryopencoding_{i, \neg \falsevalue}
\land 
\denot{\phi}_{\group^{-1}(i)} \land \neg \del(\tuple'_i) \land \cmdCopy(\group^{-1}(i), \tuple'_i, \attrlist)
\Big)

\end{array}
}
{
\encodefullsemantics \big( \groupby_{\vec{\expression}, \attrlist, \predicate} \big) 
\ruleleadsto 
\bigland_{i = 1, \mydots, n}
\queryopencoding_{i, \truevalue} 
\land 
\queryopencoding_{i, \truevalue_{\predicate}} 
\land 
\queryopencoding_{i, \falsevalue} 
}

\\ \\

(9) & 
\irule
{
\begin{array}{l}

\queryopinput = [\tuple_1, \mydots, \tuple_n] 
\quad
\queryopoutput = [\tuple'_{1}, \mydots, \tuple'_n] 
\quad
\underapproxvar = [ \underapproxvar_1, \mydots, \underapproxvar_n ]

\\[2pt]

\queryopencoding_{i, \truevalue} = 
( \underapproxvar_i = \truevalue ) 

\to
\big(
\cmdExist(\tuple_i, \tuple'_i, \land_{j=1}^{i-1} t_i \neq t_j) \land \cmdCopy(\tuple_i, \tuple'_i, \attrlist)
\big)

\quad

\queryopencoding_{i, \falsevalue} = 

( \underapproxvar_i = \falsevalue ) 

\to 

\cmdCondDel(\tuple_i, \tuple'_i, \land_{j=1}^{i-1} t_i \neq t_j)

\end{array}
}
{
\encodefullsemantics \big( \distinct \big) 
\ruleleadsto 
\bigland_{i = 1, \mydots, n}
\queryopencoding_{i, \truevalue} 
\land 
\queryopencoding_{i, \falsevalue} 
}

\\ \\

(10) & 
\irule
{
\begin{array}{l}

\queryopinput_1 = [\tuple_1, \mydots, \tuple_{n_1}] 
\quad
\queryopinput_2 = [\tuple'_1, \mydots, \tuple_{n_2}] 
\quad
\queryopoutput = [\tuple\doubleprime_{1}, \mydots, \tuple\doubleprime_{n_1+n_2}] 
\quad
\underapproxvar = [ \underapproxvar_1, \mydots, \underapproxvar_{n_1+n_2} ]

\\[2pt]

\queryopencoding_{i, \truevalue} = 
( \underapproxvar_i = \truevalue ) 

\to

\big(
\cmdCopy(\tuple_{i}, \tuple''_i, \attrlist_1) \land \cmdExist(\tuple_i, \tuple''_i, \top)
\big)

\quad

\queryopencoding_{i, \falsevalue} = 

( \underapproxvar_i = \falsevalue ) 

\to 

\cmdCondDel(\tuple_i, \tuple''_i, \top)

\\

\queryopencoding'_{i, \truevalue} = 
( \underapproxvar_i = \truevalue ) 

\to

\big(
(\cmdCopy(\tuple'_{i-{n_1}}, \tuple''_i, \attrlist_2) \land \cmdExist(\tuple'_{i-{n_1}}, \tuple''_i, \top))
\big)

\quad

\queryopencoding'_{i, \falsevalue} = 

( \underapproxvar_i = \falsevalue ) 

\to 

\cmdCondDel(\tuple'_{i-{n_1}}, \tuple''_i, \top)

\end{array}
}
{
\encodefullsemantics \big( \unionall \big) 
\ruleleadsto 

\big(
\bigland_{i = 1, \mydots, n}
\queryopencoding_{i, \truevalue} 
\land 
\queryopencoding_{i, \falsevalue} 
\big)

\land

\big(
\bigland_{i = n_1+1, \mydots, n_1 + n_2}
\queryopencoding'_{i, \truevalue} 
\land 
\queryopencoding'_{i, \falsevalue} 
\big)
}

\\ \\

(11) & 
\irule
{
\begin{array}{l}

\queryopinput = [\tuple_1, \mydots, \tuple_n] 

\quad

\queryopoutput = [\tuple'_{1}, \mydots, \tuple'_n] 

\quad

\underapproxvar = [ \underapproxvar_1, \mydots, \underapproxvar_n ]

\quad

\attrlist = \tableattrs(\queryopinput)

\\

\queryopencoding_{i, \falsevalue} = 
( \underapproxvar_i = \falsevalue  ) 
\to 
\del(\tuple_i)

\quad

\queryopencoding_{i, \truevalue_k} 
= 
( \underapproxvar_i = \truevalue_k ) 
\to 

\big(

\sum_{j=1}^{n}
\indicator 
\Big( 
j \neq i \land \denot{\expression}_{\tuple_j} < \denot{\expression}_{\tuple_i} 
\Big) 
= 
k 

\land 

\cmdCopy 
\Big( 
\tuple_i, \tuple'_{ k + \sum_{j=1}^{i-1}\indicator(\tuple_j = \tuple_i)}, \attrlist 
\Big)

\big)

\end{array}
}
{
\encodefullsemantics \big( \orderby_\expression \big) 
\ruleleadsto 
\big(
\bigland_{i = 1, \mydots, n}
\bigland_{k = 1, \mydots, n}
\queryopencoding_{i, \falsevalue}
\land 
\queryopencoding_{i, \truevalue_k} 
\big)

\land

\big(
\bigland_{i = \sum_{i=1}^n\indicator(\neg\del(\tuple_i)) + 1, \mydots, n}
\del(\tuple_i)
\big)
}



\end{array}
\]
\caption{Full semantics for query operators from Figure ~\ref{fig:sql-syntax} (cont.)}
\label{fig:under-semantics-b}
\end{figure}

%% file: sections/appendix/proof.tex
\section{Proofs} \label{sec:proof}

We provide proofs for all theorems in this section.

\textsc{Theorem}~\ref{thm:semantics-correctness}. 
Suppose $\encodeuasemantics(\queryop, \underapprox)$ yields an SMT formula $\formula$, for query operator $\queryop$ and UA $\underapprox \in \underapproxfamily_{\queryop}$.
For any model $\smtmodel$ of $\formula$, the corresponding inputs $\smtmodel(\vec{x})$ and output $\smtmodel(y)$ are consistent with the precise semantics of $\queryop$; that is, $\denot{\queryop}_{\smtmodel(\vec{x})} = \smtmodel(y)$.
\begin{proof}
Prove by case analysis of $\queryop$.
\begin{itemize}[leftmargin=*]

\item $\queryop = \proj_\attrlist$. \\
Suppose $\encodeuasemantics(\queryop, \underapprox)$ is satisfiable, let us consider a model $\smtmodel$ of the result formula, i.e., $\smtmodel \models \encodeuasemantics(\queryop, \underapprox)$.
By the definition $\encodeuasemantics(\queryop, \underapprox) = \encodechoice(\underapprox) \land \encodefullsemantics(\queryop)$, it holds that $\smtmodel \models \encodechoice(\underapprox)$ and $\smtmodel \models \encodefullsemantics(\queryop)$.
Consider the definition of $\encodefullsemantics(\queryop)$ in Figure~\ref{fig:under-semantics}, $\smtmodel \models \bigland_{i = 1, \mydots, n} \queryopencoding_{i, \truevalue} \land  \queryopencoding_{i, \falsevalue}$ where
\[
\queryopencoding_{i, \truevalue} = ( \underapproxvar_i = \truevalue ) \to \big( \cmdExist(\tuple_i, \tuple'_i, \truepredicate) \land \cmdCopy(\tuple_i, \tuple'_i, \attrlist) \big)
\]
and
\[
\queryopencoding_{i, \falsevalue} = ( \underapproxvar_i = \falsevalue ) \to \cmdCondDel(\tuple_i, \tuple'_i, \truepredicate)
\]
Thus, $\smtmodel \models \queryopencoding_{i, \truevalue}$ and $\smtmodel \models \queryopencoding_{i, \falsevalue}$.
Next, let us consider two possible values of $\smtmodel(\underapproxvar_i)$.
    \begin{enumerate}
    \item $\smtmodel(\underapproxvar_i) = \truevalue$.
    Since  $\smtmodel \models \queryopencoding_{i, \truevalue}$, it holds that $\smtmodel \models \cmdExist(\tuple_i, \tuple'_i, \truepredicate) \land \cmdCopy(\tuple_i, \tuple'_i, \attrlist)$. By the definition of $\cmdExist$ and $\cmdCopy$, we know that (a) $\smtmodel(\tuple_i)$ exists in the input $\smtmodel(\queryopinput)$, (b) $\smtmodel(\tuple'_i)$ exists in the result $\smtmodel(\queryopoutput)$, and (c) $\forall a \in \attrlist.~\denot{\smtmodel(\tuple_i).a} = \denot{\smtmodel(\tuple'_i).a}$. By the semantics of $\proj$, $\denot{\proj_\attrlist}_{[\smtmodel(\tuple_i)]} = [\smtmodel(\tuple'_i)]$.
    
    \item $\smtmodel(\underapproxvar_i) = \falsevalue$. Since $\smtmodel \models \queryopencoding_{i, \falsevalue}$, it holds that $\smtmodel \models \cmdCondDel(\tuple_i, \tuple'_i, \truepredicate)$. By the definition of $\cmdCondDel$, we know that (a) $\smtmodel(\tuple_i)$ does not exist in input $\smtmodel(\queryopinput)$, and (b) $\smtmodel(\tuple'_i)$ does not exist in output $\smtmodel(\queryopoutput)$.
    \end{enumerate}
Observe that for each tuple $\tuple_i$ ($1 \leq i \leq n$) in the input $\queryopinput$: if $\smtmodel(\underapproxvar_i) = \truevalue$, then $\denot{\proj_\attrlist}_{[\smtmodel(\tuple_i)]} = [\smtmodel(\tuple'_i)]$; if $\smtmodel(\underapproxvar_i) = \falsevalue$, then $\smtmodel(\tuple_i)$ and $\smtmodel(\tuple'_i)$ do not exist in the input or output. Therefore, by the semantics of $\proj$, $\denot{\proj_\attrlist}_{\smtmodel(\queryopinput)} = \smtmodel(\queryopoutput)$.

\item $\queryop = \filter_\predicate.$ \\
Suppose $\encodeuasemantics(\queryop, \underapprox)$ is satisfiable, let us consider a model $\smtmodel$ of the result formula, i.e., $\smtmodel \models \encodeuasemantics(\queryop, \underapprox)$.
By the definition $\encodeuasemantics(\queryop, \underapprox) = \encodechoice(\underapprox) \land \encodefullsemantics(\queryop)$, it holds that $\smtmodel \models \encodechoice(\underapprox)$ and $\smtmodel \models \encodefullsemantics(\queryop)$.
Consider the definition of $\encodefullsemantics(\queryop)$ in Figure~\ref{fig:under-semantics}, $\smtmodel \models \bigland_{i = 1, \mydots, n} \queryopencoding_{i, \truevalue} \land  \queryopencoding_{i, \falsevalue}$ where
\[
\queryopencoding_{i, \truevalue} = ( \underapproxvar_i = \truevalue ) \to \big(
\cmdExist(\tuple_i, \tuple'_i, \denot{\predicate}_{\tuple'_i}) \land \cmdCopy(\tuple_i, \tuple'_i, \attrlist) \big)
\]
and
\[
\queryopencoding_{i, \falsevalue} = ( \underapproxvar_i = \falsevalue ) \to \cmdCondDel(\tuple_i, \tuple'_i, \denot{\predicate}_{\tuple'_i})
\]
Thus, $\smtmodel \models \queryopencoding_{i, \truevalue}$ and $\smtmodel \models \queryopencoding_{i, \falsevalue}$.
Next, let us consider two possible values of $\smtmodel(\underapproxvar_i)$.
    \begin{enumerate}
    \item $\smtmodel(\underapproxvar_i) = \truevalue$.
    Since  $\smtmodel \models \queryopencoding_{i, \truevalue}$, it holds that $\smtmodel \models \cmdExist(\tuple_i, \tuple'_i, \truepredicate) \land \cmdCopy(\tuple_i, \tuple'_i, \attrlist)$. By the definition of $\cmdExist$ and $\cmdCopy$, we know that (a) $\smtmodel(\tuple_i)$ exists in the input $\smtmodel(\queryopinput)$ and $\denot{\predicate}_{\tuple_i} = \top$, (b) $\smtmodel(\tuple'_i)$ exists in the result $\smtmodel(\queryopoutput)$, and (c) $\forall a \in \attrlist.~\denot{\smtmodel(\tuple_i).a} = \denot{\smtmodel(\tuple'_i).a}$. By the semantics of $\filter$, $\denot{\filter_\predicate}_{[\smtmodel(\tuple_i)]} = [\smtmodel(\tuple'_i)]$.
    
    \item $\smtmodel(\underapproxvar_i) = \falsevalue$. Since $\smtmodel \models \queryopencoding_{i, \falsevalue}$, it holds that $\smtmodel \models \cmdCondDel(\tuple_i, \tuple'_i, \truepredicate)$. By the definition of $\cmdCondDel$, we know that (a) $\smtmodel(\tuple_i)$ does not exist in input $\smtmodel(\queryopinput)$ or $\smtmodel(\tuple_i)$ exists in $\smtmodel(\queryopinput)$ but $\denot{\predicate}_{\tuple_i} = \bot$, and (b) $\smtmodel(\tuple'_i)$ does not exist in output $\smtmodel(\queryopoutput)$.  By the semantics of $\filter$, $\denot{\filter_\predicate}_{[\smtmodel(\tuple_i)]} = []$.
    \end{enumerate}
Observe that for each tuple $\tuple_i$ ($1 \leq i \leq n$) in the input $\queryopinput$: if $\smtmodel(\underapproxvar_i) = \truevalue$, then $\denot{\filter_\predicate}_{[\smtmodel(\tuple_i)]} = [\smtmodel(\tuple'_i)]$; if $\smtmodel(\underapproxvar_i) = \falsevalue$, then $\smtmodel(\tuple'_i)$ do not exist does not exist in the output and $\smtmodel(\tuple_i)$ either does not exist in the input or does not satisfy the filter predicate $\predicate$. Therefore, by the semantics of $\filter$, $\denot{\filter_\predicate}_{\smtmodel(\queryopinput)} = \smtmodel(\queryopoutput)$.

\item $\queryop = \product$. \\
Suppose $\encodeuasemantics(\queryop, \underapprox)$ is satisfiable, let us consider a model $\smtmodel$ of the result formula, i.e., $\smtmodel \models \encodeuasemantics(\queryop, \underapprox)$.
By the definition $\encodeuasemantics(\queryop, \underapprox) = \encodechoice(\underapprox) \land \encodefullsemantics(\queryop)$, it holds that $\smtmodel \models \encodechoice(\underapprox)$ and $\smtmodel \models \encodefullsemantics(\queryop)$.
Consider the definition of $\encodefullsemantics(\queryop)$ in Figure~\ref{fig:under-semantics}, $\smtmodel \models \bigland_{i = 1, \mydots, n_1} \bigland_{ j = 1, \mydots, n_2} \queryopencoding_{i, j, \truevalue} \land \queryopencoding_{i, j, \falsevalue}$ where
\[
\queryopencoding_{i, j, \truevalue} = 
( \underapproxvar_{i,j} = \truevalue ) 
\to 
\cmdExist 
\Big( 
( \tuple_i, \tuple'_j ), \tuple\doubleprime_{i, j}, \top 
\Big) 
\land 
\cmdCopy( \tuple_i, \tuple\doubleprime_{i,j}, \attrlist_1 ) 
\land 
\cmdCopy( \tuple'_j, \tuple\doubleprime_{i,j}, \attrlist_2 )
\]
and
\[
\queryopencoding_{i, j, \falsevalue} = 
( \underapproxvar_{i,j} = \falsevalue )  
\to 
\cmdCondDel 
\Big( 
( \tuple_i, \tuple'_j ),
\tuple\doubleprime_{i,j}, 
\top 
\Big)
\]
Thus, $\smtmodel \models \queryopencoding_{i, j, \truevalue}$ and $\smtmodel \models \queryopencoding_{i, j, \falsevalue}$.
Next, let us consider two possible values of $\smtmodel(\underapproxvar_{i,j})$.
    \begin{enumerate}
    \item $\smtmodel(\underapproxvar_{i,j}) = \truevalue$.
    Since  $\smtmodel \models \queryopencoding_{i, j, \truevalue}$, it holds that $\smtmodel \models \cmdExist \Big( ( \tuple_i, \tuple'_j ), \tuple\doubleprime_{i, j}, \top \Big) \land \cmdCopy( \tuple_i, \tuple\doubleprime_{i,j}, \attrlist_1 ) \land \cmdCopy( \tuple'_j, \tuple\doubleprime_{i,j}, \attrlist_2 )$. By the definition of $\cmdExist$ and $\cmdCopy$, we know that (a) $\smtmodel(\tuple_i)$ and $\smtmodel(\tuple'_j)$ exists in the input $\smtmodel(\queryopinput)$, (b) $\smtmodel(\tuple''_{i,j})$ exists in the result $\smtmodel(\queryopoutput)$, and (c) $\forall a \in \attrlist_1.~\denot{\smtmodel(\tuple''_{i,j}).a} = \denot{\smtmodel(\tuple_i).a}$ and $\forall a \in \attrlist_2.~\denot{\smtmodel(\tuple''_{i,j}).a} = \denot{\smtmodel(\tuple'_j).a}$. By the semantics of $\product$, $\denot{\product}_{[\smtmodel(\tuple_i), \smtmodel(\tuple'_j)]} = [\texttt{merge}(\smtmodel(\tuple_i), \smtmodel(\tuple'_j))]$.
    
    \item $\smtmodel(\underapproxvar_{i,j}) = \falsevalue$. Since $\smtmodel \models \queryopencoding_{i, j, \truevalue}$, it holds that $\smtmodel \models \cmdCondDel \Big( ( \tuple_i, \tuple'_j ), \tuple\doubleprime_{i,j}, \top \Big)$. By the definition of $\cmdCondDel$, we know that (a) $\smtmodel(\tuple_i)$ or $\smtmodel(\tuple'_j)$ does not exist in input $\smtmodel(\queryopinput)$, and (b) $\smtmodel(\tuple''_{i,j})$ does not exist in output $\smtmodel(\queryopoutput)$.
    \end{enumerate}
Observe that for each tuple pair $\tuple_i$ and $\tuple'_j$ ($1 \leq i \leq n_1$, $1 \leq i \leq n_2$) in the input $\queryopinput$: if $\smtmodel(\underapproxvar_{i,j}) = \truevalue$, then $\denot{\product}_{[\smtmodel(\tuple_i), \smtmodel(\tuple'_j)]} = [\texttt{merge}(\smtmodel(\tuple_i), \smtmodel(\tuple'_j))]$; if $\smtmodel(\underapproxvar_{i,j}) = \falsevalue$, then $\smtmodel(\tuple_i)$ or $\smtmodel(\tuple'_j)$ does not exist in the input, and $\smtmodel(\tuple''_{i,j})$ does not exist in the output. Therefore, by the semantics of $\product$, $\denot{\product}_{\smtmodel(\queryopinput)} = \smtmodel(\queryopoutput)$.

\item $\queryop = \ijoin_\predicate$. \\
Suppose $\encodeuasemantics(\queryop, \underapprox)$ is satisfiable, let us consider a model $\smtmodel$ of the result formula, i.e., $\smtmodel \models \encodeuasemantics(\queryop, \underapprox)$.
By the definition $\encodeuasemantics(\queryop, \underapprox) = \encodechoice(\underapprox) \land \encodefullsemantics(\queryop)$, it holds that $\smtmodel \models \encodechoice(\underapprox)$ and $\smtmodel \models \encodefullsemantics(\queryop)$.
Consider the definition of $\encodefullsemantics(\queryop)$ in Figure~\ref{fig:under-semantics}, $\smtmodel \models \bigland_{i = 1, \mydots, n_1} \bigland_{ j = 1, \mydots, n_2} \queryopencoding_{i, j, \truevalue} \land \queryopencoding_{i, j, \falsevalue}$ where
\[
\queryopencoding_{i, j, \truevalue} = 
( \underapproxvar_{i,j} = \truevalue ) 
\to 
\cmdExist 
\Big( 
( \tuple_i, \tuple'_j ), \tuple\doubleprime_{i, j}, \denot{\predicate}_{\tuple_i, \tuple'_j} 
\Big) 
\land 
\cmdCopy( \tuple_i, \tuple\doubleprime_{i,j}, \attrlist_1 ) 
\land 
\cmdCopy( \tuple'_j, \tuple\doubleprime_{i,j}, \attrlist_2 )
\]
and
\[
\queryopencoding_{i, j, \falsevalue} = 
( \underapproxvar_{i,j} = \falsevalue )  
\to 
\cmdCondDel 
\Big( 
( \tuple_i, \tuple'_j ),
\tuple\doubleprime_{i,j}, 
\denot{\predicate}_{\tuple_i, \tuple'_j} 
\Big)
\]
Thus, $\smtmodel \models \queryopencoding_{i, j, \truevalue}$ and $\smtmodel \models \queryopencoding_{i, j, \falsevalue}$.
Next, let us consider two possible values of $\smtmodel(\underapproxvar_{i,j})$.
    \begin{enumerate}
    \item $\smtmodel(\underapproxvar_{i,j}) = \truevalue$.
    Since  $\smtmodel \models \queryopencoding_{i, j, \truevalue}$, it holds that $\smtmodel \models \cmdExist \Big( ( \tuple_i, \tuple'_j ), \tuple\doubleprime_{i, j},\denot{\predicate}_{\tuple_i, \tuple'_j} \Big) \land \cmdCopy( \tuple_i, \tuple\doubleprime_{i,j}, \attrlist_1 ) \land \cmdCopy( \tuple'_j, \tuple\doubleprime_{i,j}, \attrlist_2 )$. By the definition of $\cmdExist$ and $\cmdCopy$, we know that (a) $\smtmodel(\tuple_i)$ and $\smtmodel(\tuple'_j)$ exists in the input $\smtmodel(\queryopinput)$, (b) $\smtmodel(\tuple_i)$ and $\smtmodel(\tuple'_j)$ satisfy the join predicate $\predicate$, (c) $\smtmodel(\tuple''_{i,j})$ exists in the result $\smtmodel(\queryopoutput)$, and (d) $\forall a \in \attrlist_1.~\denot{\smtmodel(\tuple''_{i,j}).a} = \denot{\smtmodel(\tuple_i).a}$ and $\forall a \in \attrlist_2.~\denot{\smtmodel(\tuple''_{i,j}).a} = \denot{\smtmodel(\tuple'_j).a}$. By the semantics of $\ijoin$, $\denot{\ijoin_\predicate}_{[\smtmodel(\tuple_i),  \smtmodel(\tuple'_j)]} = [\texttt{merge}(\smtmodel(\tuple_i), \smtmodel(\tuple'_j))]$.
    
    \item $\smtmodel(\underapproxvar_{i,j}) = \falsevalue$. Since $\smtmodel \models \queryopencoding_{i, j, \truevalue}$, it holds that $\smtmodel \models \cmdCondDel \Big( ( \tuple_i, \tuple'_j ), \tuple\doubleprime_{i,j}, \denot{\predicate}_{\tuple_i, \tuple'_j} \Big)$. By the definition of $\cmdCondDel$, we know that (a) $\smtmodel(\tuple_i)$ or $\smtmodel(\tuple'_j)$ does not exist in input $\smtmodel(\queryopinput)$, or $\smtmodel(\tuple_i)$ and $\smtmodel(\tuple'_j)$ exist but do not satisfy the join predicate $\predicate$, and (b) $\smtmodel(\tuple''_{i,j})$ does not exist in output $\smtmodel(\queryopoutput)$.
    \end{enumerate}
Observe that for each tuple pair $\tuple_i$ and $\tuple'_j$ ($1 \leq i \leq n_1$, $1 \leq i \leq n_2$) in the input $\queryopinput$: if $\smtmodel(\underapproxvar_{i,j}) = \truevalue$, then $\denot{\ijoin_\predicate}_{[\smtmodel(\tuple_i), \smtmodel(\tuple'_j)]} = [\texttt{merge}(\smtmodel(\tuple_i), \smtmodel(\tuple'_j))]$; if $\smtmodel(\underapproxvar_{i,j}) = \falsevalue$, then $\smtmodel(\tuple_i)$ or $\smtmodel(\tuple'_j)$ does not exist in input $\smtmodel(\queryopinput)$, or $\smtmodel(\tuple_i)$ and $\smtmodel(\tuple'_j)$ exist but do not satisfy the join predicate $\predicate$, and $\smtmodel(\tuple''_{i,j})$ does not exist in the output. Therefore, by the semantics of $\ijoin$, $\denot{\ijoin_\predicate}_{\smtmodel(\queryopinput)} = \smtmodel(\queryopoutput)$.

\item $\queryop = \ljoin_\predicate$. \\
Suppose $\encodeuasemantics(\queryop, \underapprox)$ is satisfiable, let us consider a model $\smtmodel$ of the result formula, i.e., $\smtmodel \models \encodeuasemantics(\queryop, \underapprox)$.
By the definition $\encodeuasemantics(\queryop, \underapprox) = \encodechoice(\underapprox) \land \encodefullsemantics(\queryop)$, it holds that $\smtmodel \models \encodechoice(\underapprox)$ and $\smtmodel \models \encodefullsemantics(\queryop)$.
Consider the definition of $\encodefullsemantics(\queryop)$ in Figure~\ref{fig:under-semantics}, $\smtmodel \models \bigland_{i = 1, \mydots, n_1} \bigland_{ j = 1, \mydots, n_2} \queryopencoding_{i, j, \truevalue} \land \queryopencoding_{i, j, \falsevalue}$ $\land \big( \bigland_{i = 1, \mydots, n_1}  \queryopencoding_{i, \nullv} \big)$ where
\[
\queryopencoding_{i, j, \truevalue} = 
( \underapproxvar_{i,j} = \truevalue ) 
\to 
\cmdExist 
\Big( 
( \tuple_i, \tuple'_j ), \tuple\doubleprime_{i, j}, \denot{\predicate}_{\tuple_i, \tuple'_j} 
\Big) 
\land 
\cmdCopy( \tuple_i, \tuple\doubleprime_{i,j}, \attrlist_1 ) 
\land 
\cmdCopy( \tuple'_j, \tuple\doubleprime_{i,j}, \attrlist_2 )
\]
and
\[
\queryopencoding_{i, j, \falsevalue} = 
( \underapproxvar_{i,j} = \falsevalue )  
\to 
\cmdCondDel 
\Big( 
( \tuple_i, \tuple'_j ),
\tuple\doubleprime_{i,j}, 
\denot{\predicate}_{\tuple_i, \tuple'_j} 
\Big)
\]
and
\[
\queryopencoding_{i, \nullv} = 
\cmdNullTuples(\tuple''_{i, n_2+1}, \tuple_i, \attrlist_1, \attrlist_2, \bigcup_{j \in [1,n_2]}\underapproxvar_{i,j} )
\]
Thus, $\smtmodel \models \queryopencoding_{i, j, \truevalue}$ and $\smtmodel \models \queryopencoding_{i, j, \falsevalue}$.
Next, let us first consider the two possible values of $\smtmodel(\underapproxvar_{i,j})$.
    \begin{enumerate}
    \item $\smtmodel(\underapproxvar_{i,j}) = \truevalue$.
    Since  $\smtmodel \models \queryopencoding_{i, j, \truevalue}$, it holds that $\smtmodel \models \cmdExist \Big( ( \tuple_i, \tuple'_j ), \tuple\doubleprime_{i, j},\denot{\predicate}_{\tuple_i, \tuple'_j} \Big) \land \cmdCopy( \tuple_i, \tuple\doubleprime_{i,j}, \attrlist_1 ) \land \cmdCopy( \tuple'_j, \tuple\doubleprime_{i,j}, \attrlist_2 )$. By the definition of $\cmdExist$ and $\cmdCopy$, we know that (a) $\smtmodel(\tuple_i)$ and $\smtmodel(\tuple'_j)$ exists in the input $\smtmodel(\queryopinput)$, (b) $\smtmodel(\tuple_i)$ and $\smtmodel(\tuple'_j)$ satisfy the join predicate $\predicate$, (c) $\smtmodel(\tuple''_{i,j})$ exists in the result $\smtmodel(\queryopoutput)$, and (d) $\forall a \in \attrlist_1.~\denot{\smtmodel(\tuple''_{i,j}).a} = \denot{\smtmodel(\tuple_i).a}$ and $\forall a \in \attrlist_2.~\denot{\smtmodel(\tuple''_{i,j}).a} = \denot{\smtmodel(\tuple'_j).a}$.
    
    \item $\smtmodel(\underapproxvar_{i,j}) = \falsevalue$. Since $\smtmodel \models \queryopencoding_{i, j, \truevalue}$, it holds that $\smtmodel \models \cmdCondDel \Big( ( \tuple_i, \tuple'_j ), \tuple\doubleprime_{i,j}, \denot{\predicate}_{\tuple_i, \tuple'_j} \Big)$. By the definition of $\cmdCondDel$, we know that (a) $\smtmodel(\tuple_i)$ or $\smtmodel(\tuple'_j)$ does not exist in input $\smtmodel(\queryopinput)$, or $\smtmodel(\tuple_i)$ and $\smtmodel(\tuple'_j)$ exist but do not satisfy the join predicate $\predicate$, and (b) $\smtmodel(\tuple''_{i,j})$ does not exist in output $\smtmodel(\queryopoutput)$.
    \end{enumerate}
As $\smtmodel \models \queryopencoding_{i, \nullv}$, it also holds that $\smtmodel \models \cmdNullTuples(\tuple''_{i, n_2+1}, \tuple_i, \attrlist_1, \attrlist_2, \bigcup_{j \in [1,n_2]}\underapproxvar_{i,j} )$.  Then, by the definition of $\cmdNullTuples$, if $\bigland_{j=1, \mydots, n_2} \smtmodel(\underapproxvar_{i,j}) = \falsevalue$, we know that (a) all tuples $\smtmodel(\tuple''_{i, j})$ where $1 \leq j \leq n_2$ do not exist in the output $\smtmodel(\queryopoutput)$, and therefore (b) $\smtmodel(\tuple''_{i, n_2 + 1})$ exists in $\smtmodel(\queryopoutput)$, and (c) $\forall a \in \attrlist_1.~\denot{\smtmodel(\tuple''_{i,j}).a} = \denot{\smtmodel(\tuple_i).a}$ and $\forall a \in \attrlist_2.~\denot{\smtmodel(\tuple''_{i,j}).a} = \nullv$.  Otherwise, if $\biglor_{j=1, \mydots, n_2} \smtmodel(\underapproxvar_{i,j}) = \truevalue$, we have that (a) at least one tuple $\smtmodel(\tuple''_{i, j})$ where $1 \leq j \leq n_2$ exists in the output $\smtmodel(\queryopoutput)$, and therefore (b) $\smtmodel(\tuple''_{i, n_2 + 1})$ does not exist in $\smtmodel(\queryopoutput)$.

Therefore, by the semantics of $\ljoin$, $\denot{\ljoin_\predicate}_{\smtmodel(\queryopinput)} = \smtmodel(\queryopoutput)$.

\item $\queryop = \rjoin_\predicate$. \\
Suppose $\encodeuasemantics(\queryop, \underapprox)$ is satisfiable, let us consider a model $\smtmodel$ of the result formula, i.e., $\smtmodel \models \encodeuasemantics(\queryop, \underapprox)$.
By the definition $\encodeuasemantics(\queryop, \underapprox) = \encodechoice(\underapprox) \land \encodefullsemantics(\queryop)$, it holds that $\smtmodel \models \encodechoice(\underapprox)$ and $\smtmodel \models \encodefullsemantics(\queryop)$.
Consider the definition of $\encodefullsemantics(\queryop)$ in Figure~\ref{fig:under-semantics}, $\smtmodel \models \bigland_{i = 1, \mydots, n_1} \bigland_{ j = 1, \mydots, n_2} \queryopencoding_{i, j, \truevalue} \land \queryopencoding_{i, j, \falsevalue}$ $\land \big( \bigland_{j = 1, \mydots, n_2}  \queryopencoding_{j, \nullv} \big)$ where
\[
\queryopencoding_{i, j, \truevalue} = 
( \underapproxvar_{i,j} = \truevalue ) 
\to 
\cmdExist 
\Big( 
( \tuple_i, \tuple'_j ), \tuple\doubleprime_{i, j}, \denot{\predicate}_{\tuple_i, \tuple'_j} 
\Big) 
\land 
\cmdCopy( \tuple_i, \tuple\doubleprime_{i,j}, \attrlist_1 ) 
\land 
\cmdCopy( \tuple'_j, \tuple\doubleprime_{i,j}, \attrlist_2 )
\]
and
\[
\queryopencoding_{i, j, \falsevalue} = 
( \underapproxvar_{i,j} = \falsevalue )  
\to 
\cmdCondDel 
\Big( 
( \tuple_i, \tuple'_j ),
\tuple\doubleprime_{i,j}, 
\denot{\predicate}_{\tuple_i, \tuple'_j} 
\Big)
\]
and
\[
\queryopencoding_{j, \nullv} = 
\cmdNullTuples(\tuple''_{n_1+1, j}, \tuple'_j, \attrlist_2, \attrlist_1, \bigcup_{i \in [1,n_1]}\underapproxvar_{i,j} )
\]
Thus, $\smtmodel \models \queryopencoding_{i, j, \truevalue}$ and $\smtmodel \models \queryopencoding_{i, j, \falsevalue}$.
Next, let us first consider the two possible values of $\smtmodel(\underapproxvar_{i,j})$.
    \begin{enumerate}
    \item $\smtmodel(\underapproxvar_{i,j}) = \truevalue$.
    Since  $\smtmodel \models \queryopencoding_{i, j, \truevalue}$, it holds that $\smtmodel \models \cmdExist \Big( ( \tuple_i, \tuple'_j ), \tuple\doubleprime_{i, j},\denot{\predicate}_{\tuple_i, \tuple'_j} \Big) \land \cmdCopy( \tuple_i, \tuple\doubleprime_{i,j}, \attrlist_1 ) \land \cmdCopy( \tuple'_j, \tuple\doubleprime_{i,j}, \attrlist_2 )$. By the definition of $\cmdExist$ and $\cmdCopy$, we know that (a) $\smtmodel(\tuple_i)$ and $\smtmodel(\tuple'_j)$ exists in the input $\smtmodel(\queryopinput)$, (b) $\smtmodel(\tuple_i)$ and $\smtmodel(\tuple'_j)$ satisfy the join predicate $\predicate$, (c) $\smtmodel(\tuple''_{i,j})$ exists in the result $\smtmodel(\queryopoutput)$, and (d) $\forall a \in \attrlist_1.~\denot{\smtmodel(\tuple''_{i,j}).a} = \denot{\smtmodel(\tuple_i).a}$ and $\forall a \in \attrlist_2.~\denot{\smtmodel(\tuple''_{i,j}).a} = \denot{\smtmodel(\tuple'_j).a}$.
    
    \item $\smtmodel(\underapproxvar_{i,j}) = \falsevalue$. Since $\smtmodel \models \queryopencoding_{i, j, \truevalue}$, it holds that $\smtmodel \models \cmdCondDel \Big( ( \tuple_i, \tuple'_j ), \tuple\doubleprime_{i,j}, \denot{\predicate}_{\tuple_i, \tuple'_j} \Big)$. By the definition of $\cmdCondDel$, we know that (a) $\smtmodel(\tuple_i)$ or $\smtmodel(\tuple'_j)$ does not exist in input $\smtmodel(\queryopinput)$, or $\smtmodel(\tuple_i)$ and $\smtmodel(\tuple'_j)$ exist but do not satisfy the join predicate $\predicate$, and (b) $\smtmodel(\tuple''_{i,j})$ does not exist in output $\smtmodel(\queryopoutput)$.
    \end{enumerate}
As $\smtmodel \models \queryopencoding_{j, \nullv}$, it also holds that $\smtmodel \models \cmdNullTuples(\tuple''_{n_1+1, j}, \tuple'_j, \attrlist_2, \attrlist_1, \bigcup_{i \in [1,n_1]}\underapproxvar_{i,j} )$.  Then, by the definition of $\cmdNullTuples$, if $\bigland_{i=1, \mydots, n_1} \smtmodel(\underapproxvar_{i,j}) = \falsevalue$, we know that (a) all tuples $\smtmodel(\tuple''_{i, j})$ where $1 \leq i \leq n_1$ do not exist in the output $\smtmodel(\queryopoutput)$, and therefore (b) $\smtmodel(\tuple''_{n_1+1, j})$ exists in $\smtmodel(\queryopoutput)$, and (c) $\forall a \in \attrlist_1.~\denot{\smtmodel(\tuple''_{i,j}).a} = \nullv$ and $\forall a \in \attrlist_2.~\denot{\smtmodel(\tuple''_{i,j}).a} = \denot{\smtmodel(\tuple'_j).a}$.  Otherwise, if $\biglor_{i=1, \mydots, n_1} \smtmodel(\underapproxvar_{i,j}) = \truevalue$, we have that (a) at least one tuple $\smtmodel(\tuple''_{i, j})$ where $1 \leq i \leq n_1$ exists in the output $\smtmodel(\queryopoutput)$, and therefore (b) $\smtmodel(\tuple''_{n_1+1, j})$ does not exist in $\smtmodel(\queryopoutput)$.

Therefore, by the semantics of $\rjoin$, $\denot{\rjoin_\predicate}_{\smtmodel(\queryopinput)} = \smtmodel(\queryopoutput)$.

\item $\queryop = \fjoin_\predicate$. \\
Suppose $\encodeuasemantics(\queryop, \underapprox)$ is satisfiable, let us consider a model $\smtmodel$ of the result formula, i.e., $\smtmodel \models \encodeuasemantics(\queryop, \underapprox)$.
By the definition $\encodeuasemantics(\queryop, \underapprox) = \encodechoice(\underapprox) \land \encodefullsemantics(\queryop)$, it holds that $\smtmodel \models \encodechoice(\underapprox)$ and $\smtmodel \models \encodefullsemantics(\queryop)$.
Consider the definition of $\encodefullsemantics(\queryop)$ in Figure~\ref{fig:under-semantics}, $\smtmodel \models \bigland_{i = 1, \mydots, n_1} \bigland_{ j = 1, \mydots, n_2} \queryopencoding_{i, j, \truevalue} \land \queryopencoding_{i, j, \falsevalue}$ $\land \big( \bigland_{i = 1, \mydots, n_1}  \queryopencoding_{i, \nullv} \big) \land \big(\bigland_{j = 1, \mydots, n_2} \queryopencoding'_{j, \nullv} \big)$ where
\[
\queryopencoding_{i, j, \truevalue} = 
( \underapproxvar_{i,j} = \truevalue ) 
\to 
\cmdExist 
\Big( 
( \tuple_i, \tuple'_j ), \tuple\doubleprime_{i, j}, \denot{\predicate}_{\tuple_i, \tuple'_j} 
\Big) 
\land 
\cmdCopy( \tuple_i, \tuple\doubleprime_{i,j}, \attrlist_1 ) 
\land 
\cmdCopy( \tuple'_j, \tuple\doubleprime_{i,j}, \attrlist_2 )
\]
and
\[
\queryopencoding_{i, j, \falsevalue} = 
( \underapproxvar_{i,j} = \falsevalue )  
\to 
\cmdCondDel 
\Big( 
( \tuple_i, \tuple'_j ),
\tuple\doubleprime_{i,j}, 
\denot{\predicate}_{\tuple_i, \tuple'_j} 
\Big)
\]
and
\begin{align*}
\queryopencoding_{i, \nullv} = 
\cmdNullTuples(\tuple''_{i, n_2+1}, \tuple_i, \attrlist_1, \attrlist_2, \bigcup_{j \in [1,n_2]}\underapproxvar_{i,j} ) \\
\queryopencoding'_{j, \nullv} = 
\cmdNullTuples(\tuple''_{n_1+1, j}, \tuple'_j, \attrlist_2, \attrlist_1, \bigcup_{i \in [1,n_1]}\underapproxvar_{i,j} )
\end{align*}
Thus, $\smtmodel \models \queryopencoding_{i, j, \truevalue}$ and $\smtmodel \models \queryopencoding_{i, j, \falsevalue}$.
Next, let us first consider the two possible values of $\smtmodel(\underapproxvar_{i,j})$.
    \begin{enumerate}
    \item $\smtmodel(\underapproxvar_{i,j}) = \truevalue$.
    Since  $\smtmodel \models \queryopencoding_{i, j, \truevalue}$, it holds that $\smtmodel \models \cmdExist \Big( ( \tuple_i, \tuple'_j ), \tuple\doubleprime_{i, j},\denot{\predicate}_{\tuple_i, \tuple'_j} \Big) \land \cmdCopy( \tuple_i, \tuple\doubleprime_{i,j}, \attrlist_1 ) \land \cmdCopy( \tuple'_j, \tuple\doubleprime_{i,j}, \attrlist_2 )$. By the definition of $\cmdExist$ and $\cmdCopy$, we know that (a) $\smtmodel(\tuple_i)$ and $\smtmodel(\tuple'_j)$ exists in the input $\smtmodel(\queryopinput)$, (b) $\smtmodel(\tuple_i)$ and $\smtmodel(\tuple'_j)$ satisfy the join predicate $\predicate$, (c) $\smtmodel(\tuple''_{i,j})$ exists in the result $\smtmodel(\queryopoutput)$, and (d) $\forall a \in \attrlist_1.~\denot{\smtmodel(\tuple''_{i,j}).a} = \denot{\smtmodel(\tuple_i).a}$ and $\forall a \in \attrlist_2.~\denot{\smtmodel(\tuple''_{i,j}).a} = \denot{\smtmodel(\tuple'_j).a}$.
    
    \item $\smtmodel(\underapproxvar_{i,j}) = \falsevalue$. Since $\smtmodel \models \queryopencoding_{i, j, \truevalue}$, it holds that $\smtmodel \models \cmdCondDel \Big( ( \tuple_i, \tuple'_j ), \tuple\doubleprime_{i,j}, \denot{\predicate}_{\tuple_i, \tuple'_j} \Big)$. By the definition of $\cmdCondDel$, we know that (a) $\smtmodel(\tuple_i)$ or $\smtmodel(\tuple'_j)$ does not exist in input $\smtmodel(\queryopinput)$, or $\smtmodel(\tuple_i)$ and $\smtmodel(\tuple'_j)$ exist but do not satisfy the join predicate $\predicate$, and (b) $\smtmodel(\tuple''_{i,j})$ does not exist in output $\smtmodel(\queryopoutput)$.
    \end{enumerate}
As $\smtmodel \models \queryopencoding_{i, \nullv}$, it also holds that $\smtmodel \models \cmdNullTuples(\tuple''_{i, n_2+1}, \tuple_i, \attrlist_1, \attrlist_2, \bigcup_{j \in [1,n_2]}\underapproxvar_{i,j} )$.  Then, by the definition of $\cmdNullTuples$, if $\bigland_{j=1, \mydots, n_2} \smtmodel(\underapproxvar_{i,j}) = \falsevalue$, we know that (a) all tuples $\smtmodel(\tuple''_{i, j})$ where $1 \leq j \leq n_2$ do not exist in the output $\smtmodel(\queryopoutput)$, and therefore (b) $\smtmodel(\tuple''_{i, n_2 + 1})$ exists in $\smtmodel(\queryopoutput)$, and (c) $\forall a \in \attrlist_1.~\denot{\smtmodel(\tuple''_{i,j}).a} = \denot{\smtmodel(\tuple_i).a}$ and $\forall a \in \attrlist_2.~\denot{\smtmodel(\tuple''_{i,j}).a} = \nullv$.  Otherwise, if $\biglor_{j=1, \mydots, n_2} \smtmodel(\underapproxvar_{i,j}) = \truevalue$, we have that (a) at least one tuple $\smtmodel(\tuple''_{i, j})$ where $1 \leq j \leq n_2$ exists in the output $\smtmodel(\queryopoutput)$, and therefore (b) $\smtmodel(\tuple''_{i, n_2 + 1})$ does not exist in $\smtmodel(\queryopoutput)$.

Also since $\smtmodel \models \queryopencoding'_{j, \nullv}$, it also holds that $\smtmodel \models \cmdNullTuples(\tuple''_{n_1+1, j}, \tuple'_j, \attrlist_2, \attrlist_1, \bigcup_{i \in [1,n_1]}\underapproxvar_{i,j} )$.  Then, by the definition of $\cmdNullTuples$, if $\bigland_{i=1, \mydots, n_1} \smtmodel(\underapproxvar_{i,j}) = \falsevalue$, we know that (a) all tuples $\smtmodel(\tuple''_{i, j})$ where $1 \leq i \leq n_1$ do not exist in the output $\smtmodel(\queryopoutput)$, and therefore (b) $\smtmodel(\tuple''_{n_1+1, j})$ exists in $\smtmodel(\queryopoutput)$, and (c) $\forall a \in \attrlist_1.~\denot{\smtmodel(\tuple''_{i,j}).a} = \nullv$ and $\forall a \in \attrlist_2.~\denot{\smtmodel(\tuple''_{i,j}).a} = \denot{\smtmodel(\tuple'_j).a}$.  Otherwise, if $\biglor_{i=1, \mydots, n_1} \smtmodel(\underapproxvar_{i,j}) = \truevalue$, we have that (a) at least one tuple $\smtmodel(\tuple''_{i, j})$ where $1 \leq i \leq n_1$ exists in the output $\smtmodel(\queryopoutput)$, and therefore (b) $\smtmodel(\tuple''_{n_1+1, j})$ does not exist in $\smtmodel(\queryopoutput)$.

Therefore, by the semantics of $\fjoin$, $\denot{\fjoin_\predicate}_{\smtmodel(\queryopinput)} = \smtmodel(\queryopoutput)$.

\item $\queryop = \groupby_{\vec{\expression}, \attrlist, \predicate}$. \\
Suppose $\encodeuasemantics(\queryop, \underapprox)$ is satisfiable, let us consider a model $\smtmodel$ of the result formula, i.e., $\smtmodel \models \encodeuasemantics(\queryop, \underapprox)$.
By the definition $\encodeuasemantics(\queryop, \underapprox) = \encodechoice(\underapprox) \land \encodefullsemantics(\queryop)$, it holds that $\smtmodel \models \encodechoice(\underapprox)$ and $\smtmodel \models \encodefullsemantics(\queryop)$.
Consider the definition of $\encodefullsemantics(\queryop)$ in Figure~\ref{fig:under-semantics}, $\smtmodel \models \bigland_{i = 1, \mydots, n} \queryopencoding_{i, \truevalue} \land \queryopencoding_{i, \truevalue_{\predicate}} \land \queryopencoding_{i, \falsevalue} $ where
{
\scriptsize
\[
\queryopencoding_{i, \falsevalue} = 
( \underapproxvar_i = \falsevalue ) 
\to 
\Big( 
\big( 
\del( \tuple_i ) 
\lor 
\Big( 
\neg \del( \tuple_i ) 
\land 
\biglor_{j=1, \mydots, i-1} 
\big( 
\neg \del(\tuple_j) 
\land 
\bigland_{a \in \vec{\expression}} \denot{\tuple_i.a} = \denot{\tuple_j.a} 
\land 
\group(\tuple_i) = j
\big) 
\Big) 
\big)
\land 
\del( \tuple'_i )
\Big)
\]
}
and
\[
\queryopencoding_{i, \neg \falsevalue} = 
\neg \del(\tuple_i) 
\land 
\neg \biglor_{j=1, \mydots, i-1} 
\Big( 
\neg \del(\tuple_j) \land \bigland_{a \in \vec{\expression}}\denot{\tuple_i.a} = \denot{\tuple_j.a}
\Big) 
\land 
\group(\tuple_i) = i 
\]
and
\[
\queryopencoding_{i, \truevalue} = 
( \underapproxvar_i = \truevalue ) 
\to 
\Big( 
\queryopencoding_{i, \neg \falsevalue}
\land 
\neg \denot{\predicate}_{\group^{-1}(i)} 
\land 
\del(\tuple'_i)
\Big)
\]

and
\[
\queryopencoding_{i, \truevalue_{\predicate}} = 
( \underapproxvar_i = \truevalue_{\predicate} ) 
\to 
\Big( 
\queryopencoding_{i, \neg \falsevalue}
\land 
\denot{\phi}_{\group^{-1}(i)} \land \neg \del(\tuple'_i) \land \cmdCopy(\group^{-1}(i), \tuple'_i, \attrlist)
\Big)
\]
Thus, $\smtmodel \models \queryopencoding_{i, \truevalue}$, $\smtmodel \models \queryopencoding_{i, \truevalue_{\predicate}}$, and $\smtmodel \models \queryopencoding_{i, \falsevalue}$.
Next, let us first consider the three possible values of $\smtmodel(\underapproxvar_i)$.
    \begin{enumerate}
    \item $\smtmodel(\underapproxvar_i) = \falsevalue$.
    Since  $\smtmodel \models \queryopencoding_{i, \falsevalue}$, it holds that $\smtmodel \models \Big( \big( \del( \tuple_i ) \lor \Big( \neg \del( \tuple_i ) \land \biglor_{j=1, \mydots, i-1} \big( \neg \del(\tuple_j) \land \bigland_{a \in \vec{\expression}} \denot{\tuple_i.a} = \denot{\tuple_j.a} \land \group(\tuple_i) = j \big) \Big) \big) \land \del( \tuple'_i ) \Big)$.  Thus we have in this case (a) $\smtmodel(\tuple_i)$ does not exist in the input $\smtmodel(\queryopinput)$ or $\smtmodel(\tuple_i)$ exists but $\forall a \in \vec{E}.~\smtmodel(t_i.a) = \smtmodel(t_j.a)$ where $1 \leq j \le i - 1$, and (b) $\smtmodel(\tuple'_{i})$ does not exist in the result $\smtmodel(\queryopoutput)$.
    
    \item $\smtmodel(\underapproxvar_i) = \truevalue$.
    Since  $\smtmodel \models \queryopencoding_{i, \falsevalue}$, it holds that $\smtmodel \models \Big( \queryopencoding_{i, \neg \falsevalue} \land \neg \denot{\predicate}_{\group^{-1}(i)}  \land \del(\tuple'_i) \Big)$.  We know that (a) $\smtmodel(\tuple_i)$ exists in the input $\smtmodel(\queryopinput)$, (b) $\forall 1 \leq j \le i - 1.~\exists a \in \vec{E}.~\smtmodel(t_i.a) \neq \smtmodel(t_j.a)$, (c) $\denot{\phi}_{\group^{-1}(i)} = \bot$, and (d) $\smtmodel(\tuple'_{i})$ does not exist in the result $\smtmodel(\queryopoutput)$.

    \item $\smtmodel(\underapproxvar_i) = \truevalue_\predicate$.
    Since  $\smtmodel \models \queryopencoding_{i, \falsevalue}$, it holds that $\smtmodel \models \Big( \queryopencoding_{i, \neg \falsevalue} \land \neg \denot{\predicate}_{\group^{-1}(i)}  \land \del(\tuple'_i) \Big)$.  We know that (a) $\smtmodel(\tuple_i)$ exists in the input $\smtmodel(\queryopinput)$, (b) $\forall 1 \leq j \le i - 1.~\exists a \in \vec{E}.~\smtmodel(t_i.a) \neq \smtmodel(t_j.a)$, (c) $\denot{\phi}_{\group^{-1}(i)} = \top$, and (d) $\smtmodel(\tuple'_{i})$ exists in the result $\smtmodel(\queryopoutput)$.
    \end{enumerate}
Therefore, by the semantics of $\groupby$, $\denot{\groupby_{\vec{\expression}, \attrlist, \predicate}}_{\smtmodel(\queryopinput)} = \smtmodel(\queryopoutput)$.

\item $\queryop = \distinct.$ \\
Suppose $\encodeuasemantics(\queryop, \underapprox)$ is satisfiable, let us consider a model $\smtmodel$ of the result formula, i.e., $\smtmodel \models \encodeuasemantics(\queryop, \underapprox)$.
By the definition $\encodeuasemantics(\queryop, \underapprox) = \encodechoice(\underapprox) \land \encodefullsemantics(\queryop)$, it holds that $\smtmodel \models \encodechoice(\underapprox)$ and $\smtmodel \models \encodefullsemantics(\queryop)$.
Consider the definition of $\encodefullsemantics(\queryop)$ in Figure~\ref{fig:under-semantics}, $\smtmodel \models \bigland_{i = 1, \mydots, n} \queryopencoding_{i, \truevalue} \land  \queryopencoding_{i, \falsevalue}$ where
\[
\queryopencoding_{i, \truevalue} = 
( \underapproxvar_i = \truevalue ) 
\to
\big(
\cmdExist(\tuple_i, \tuple'_i, \land_{j=1}^{i-1} t_i \neq t_j) \land \cmdCopy(\tuple_i, \tuple'_i, \attrlist)
\big)
\]
and
\[
\queryopencoding_{i, \falsevalue} = 
( \underapproxvar_i = \falsevalue ) 
\to 
\cmdCondDel(\tuple_i, \tuple'_i, \land_{j=1}^{i-1} t_i \neq t_j)
\]
Thus, $\smtmodel \models \queryopencoding_{i, \truevalue}$ and $\smtmodel \models \queryopencoding_{i, \falsevalue}$.
Next, let us consider two possible values of $\smtmodel(\underapproxvar_i)$.
    \begin{enumerate}
    \item $\smtmodel(\underapproxvar_i) = \truevalue$.
    Since  $\smtmodel \models \queryopencoding_{i, \truevalue}$, it holds that $\smtmodel \models \cmdExist(\tuple_i, \tuple'_i, \land_{j=1}^{i-1} t_i \neq t_j) \land \cmdCopy(\tuple_i, \tuple'_i, \attrlist)$. By the definition of $\cmdExist$ and $\cmdCopy$, we know that (a) $\smtmodel(\tuple_i)$ exists in the input $\smtmodel(\queryopinput)$ and $\land_{j=1}^{i-1} t_i \neq t_j$, (b) $\smtmodel(\tuple'_i)$ exists in the result $\smtmodel(\queryopoutput)$, and (c) $\forall a \in \attrlist.~\denot{\smtmodel(\tuple_i).a} = \denot{\smtmodel(\tuple'_i).a}$.
    
    \item $\smtmodel(\underapproxvar_i) = \falsevalue$. Since $\smtmodel \models \queryopencoding_{i, \falsevalue}$, it holds that $\smtmodel \models \cmdCondDel(\tuple_i, \tuple'_i, \land_{j=1}^{i-1} t_i \neq t_j)$. By the definition of $\cmdCondDel$, we know that (a) $\smtmodel(\tuple_i)$ does not exist in input $\smtmodel(\queryopinput)$ or $\smtmodel(\tuple_i)$ exists in $\smtmodel(\queryopinput)$ but $\lor_{j=1}^{i-1} t_i = t_j$, and (b) $\smtmodel(\tuple'_i)$ does not exist in output $\smtmodel(\queryopoutput)$.
    \end{enumerate}
Therefore, by the semantics of $\distinct$, $\denot{\distinct}_{\smtmodel(\queryopinput)} = \smtmodel(\queryopoutput)$.

\item $\queryop = \unionall.$ \\
Suppose $\encodeuasemantics(\queryop, \underapprox)$ is satisfiable, let us consider a model $\smtmodel$ of the result formula, i.e., $\smtmodel \models \encodeuasemantics(\queryop, \underapprox)$.
By the definition $\encodeuasemantics(\queryop, \underapprox) = \encodechoice(\underapprox) \land \encodefullsemantics(\queryop)$, it holds that $\smtmodel \models \encodechoice(\underapprox)$ and $\smtmodel \models \encodefullsemantics(\queryop)$.
Consider the definition of $\encodefullsemantics(\queryop)$ in Figure~\ref{fig:under-semantics}, $\smtmodel \models \big( \bigland_{i = 1, \mydots, n} \queryopencoding_{i, \truevalue} \land \queryopencoding_{i, \falsevalue} \big) \land \big( \bigland_{i = n_1+1, \mydots, n_1 + n_2}$  $\queryopencoding'_{i, \truevalue} \land \queryopencoding'_{i, \falsevalue} \big)$ where
\begin{align*}
&    \queryopencoding_{i, \truevalue} = ( \underapproxvar_i = \truevalue ) \to
\big(
\cmdCopy(\tuple_{i}, \tuple''_i, \attrlist_1) \land \cmdExist(\tuple_i, \tuple''_i, \top)
\big) \\
& \queryopencoding_{i, \falsevalue} = 
( \underapproxvar_i = \falsevalue ) 
\to 
\cmdCondDel(\tuple_i, \tuple''_i, \top)
\end{align*}
and
\begin{align*}
& \queryopencoding'_{i, \truevalue} = 
( \underapproxvar_i = \truevalue ) 
\to
\big(
(\cmdCopy(\tuple'_{i-{n_1}}, \tuple''_i, \attrlist_2) \land \cmdExist(\tuple'_{i-{n_1}}, \tuple''_i, \top))
\big) \\
& \queryopencoding'_{i, \falsevalue} = 
( \underapproxvar_i = \falsevalue ) 
\to 
\cmdCondDel(\tuple'_{i-{n_1}}, \tuple''_i, \top)
\end{align*}
Thus, $\smtmodel \models \queryopencoding_{i, \truevalue}$, $\smtmodel \models \queryopencoding_{i, \falsevalue}$, $\smtmodel \models \queryopencoding'_{i, \truevalue}$, and $\smtmodel \models \queryopencoding'_{i, \falsevalue}$.
Observe that $\queryopencoding$ and $\queryopencoding'$ are symmetric.  
Next, let us consider two possible values of $\smtmodel(\underapproxvar_i)$.
    \begin{enumerate}
    \item $\smtmodel(\underapproxvar_i) = \truevalue$.
    Since  $\smtmodel \models \queryopencoding_{i, \truevalue}$, it holds that $\smtmodel \models \cmdCopy(\tuple_{i}, \tuple''_i, \attrlist_1) \land \cmdExist(\tuple_i, \tuple''_i, \top)$. By the definition of $\cmdExist$ and $\cmdCopy$, we know that (a) $\smtmodel(\tuple_i)$ exists in the input $\smtmodel(\queryopinput)$, (b) $\smtmodel(\tuple'_i)$ exists in the result $\smtmodel(\queryopoutput)$, and (c) $\forall a \in \attrlist_1.~\denot{\smtmodel(\tuple_i).a} = \denot{\smtmodel(\tuple''_i).a}$.
    
    \item $\smtmodel(\underapproxvar_i) = \falsevalue$. Since $\smtmodel \models \queryopencoding_{i, \falsevalue}$, it holds that $\smtmodel \models \cmdCondDel(\tuple_i, \tuple''_i, \top)$. By the definition of $\cmdCondDel$, we know that (a) $\smtmodel(\tuple_i)$ does not exist in input $\smtmodel(\queryopinput)$, and (b) $\smtmodel(\tuple'_i)$ does not exist in output $\smtmodel(\queryopoutput)$.
    \end{enumerate}
Therefore, by the semantics of $\unionall$, $\denot{\unionall}_{\smtmodel(\queryopinput)} = \smtmodel(\queryopoutput)$.

\item $\queryop = \orderby_\expression.$ \\
Suppose $\encodeuasemantics(\queryop, \underapprox)$ is satisfiable, let us consider a model $\smtmodel$ of the result formula, i.e., $\smtmodel \models \encodeuasemantics(\queryop, \underapprox)$.
By the definition $\encodeuasemantics(\queryop, \underapprox) = \encodechoice(\underapprox) \land \encodefullsemantics(\queryop)$, it holds that $\smtmodel \models \encodechoice(\underapprox)$ and $\smtmodel \models \encodefullsemantics(\queryop)$.
Consider the definition of $\encodefullsemantics(\queryop)$ in Figure~\ref{fig:under-semantics}, $\smtmodel \models \big(
\bigland_{i = 1, \mydots, n}
\bigland_{k = 1, \mydots, n}
\queryopencoding_{i, \falsevalue}
\land 
\queryopencoding_{i, \truevalue_k} 
\big)
\land
\big(
\bigland_{i = \sum_{i=1}^n\indicator(\neg\del(\tuple_i)) + 1, \mydots, n}
\del(\tuple_i)
\big)$ where
\[
\queryopencoding_{i, \falsevalue} = 
( \underapproxvar_i = \falsevalue  ) 
\to 
\del(\tuple_i) 
\]
and
\[
\queryopencoding_{i, \truevalue_k} 
= 
( \underapproxvar_i = \truevalue_k ) 
\to 
\big(
\sum_{j=1}^{n}
\indicator 
\Big( 
j \neq i \land \denot{\expression}_{\tuple_j} < \denot{\expression}_{\tuple_i} 
\Big) 
= 
k 
\land 
\cmdCopy 
\Big( 
\tuple_i, \tuple'_{ k + \sum_{j=1}^{i-1}\indicator(\tuple_j = \tuple_i)}, \attrlist 
\Big)
\big)
\]
Thus, $\smtmodel \models \queryopencoding_{i, \falsevalue}$ and $\smtmodel \models \queryopencoding_{i, \truevalue_k}$.
Next, let us consider two possible values of $\smtmodel(\underapproxvar_i)$.
    \begin{enumerate}
    \item $\smtmodel(\underapproxvar_i) = \truevalue$.
    Since  $\smtmodel \models \queryopencoding_{i, \truevalue}$, it holds that $\smtmodel \models \big(
\sum_{j=1}^{n}
\indicator 
\Big( 
j \neq i \land \denot{\expression}_{\tuple_j} < \denot{\expression}_{\tuple_i} 
\Big) 
= 
k 
\land 
\cmdCopy 
\Big( 
\tuple_i, \tuple'_{ k + \sum_{j=1}^{i-1}\indicator(\tuple_j = \tuple_i)}, \attrlist 
\Big)
\big)$. By the definition of $\cmdCopy$, we know that (a) $\smtmodel(\tuple_i)$ exists in the input $\smtmodel(\queryopinput)$, (b) $\smtmodel(\tuple'_{ k + \sum_{j=1}^{i-1}\indicator(\tuple_j = \tuple_i)})$ exists in the result $\smtmodel(\queryopoutput)$, and (c) $\forall a \in \attrlist.~\denot{\smtmodel(\tuple_i).a} = \denot{\smtmodel(\tuple'_{ k + \sum_{j=1}^{i-1}\indicator(\tuple_j = \tuple_i)}).a}$.
    
    \item $\smtmodel(\underapproxvar_i) = \falsevalue$. Since $\smtmodel \models \queryopencoding_{i, \falsevalue}$, it holds that $\smtmodel \models \del(\tuple_i)$. Then we know that $\smtmodel(\tuple_i)$ does not exist in input $\smtmodel(\queryopinput)$.
    \end{enumerate}
Therefore, by the semantics of $\orderby$, $\denot{\orderby}_{\smtmodel(\queryopinput)} = \smtmodel(\queryopoutput)$. 
\end{itemize}
\end{proof}

\begin{corollary}[Correctness of UA semantics for queries] \label{cor:ua-semantics-correcteness}
Given a query $\sqlquery$ and a UA map $\mapastnodetounderapprox$ that maps each AST node of P to a UA choice, let formula $\formula$ be the result of $\encodeuasemantics(\mapastnodetounderapprox)$.  If $\formula$ is satisfiable, then for any $\smtmodel$ of $\formula$, the corresponding inputs $\smtmodel(\vec{\queryopinput})$ and output $\smtmodel(y)$ are consistent with the precise semantics of $\sqlquery$, i.e., $\denot{\sqlquery}_{\smtmodel(\vec{\queryopinput})} = \smtmodel(\queryopoutput)$.
\end{corollary}
\begin{proof}
By the definition of $\encodeuasemantics$, the resulting formula is a conjunction of the formulas returned by $\encodeuasemantics(\astnode, \underapprox)$ where $(\astnode \mapsto \underapprox) \in \mapastnodetounderapprox$ and the inputs $\queryopinput_1, \mydots, \queryopinput_l$ to $\astnode$ are the outputs $\queryopoutput^{\astnode_1}, \mydots, \queryopoutput^{\astnode_l}$ produced by its children AST nodes $\astnode_1, \mydots, \astnode_l$.  Suppose $\encodeuasemantics(\mapastnodetounderapprox)$ is satisfiable, let us consider a model $\smtmodel$ of the result formula, i.e., $\smtmodel \models \encodeuasemantics(\mapastnodetounderapprox)$.  By Theorem ~\ref{thm:semantics-correctness}, we know that (1) $\denot{\astnode_i}_{\smtmodel(\vec{\queryopinput^{\astnode_i}})} = \smtmodel(\queryopoutput^{\astnode_i})$ holds for $1 \leq i \leq l$, and (2) $\denot{op(\astnode)}_{\smtmodel(\vec{\queryopinput})} = \smtmodel(\queryopoutput)$. 
 Therefore, it also holds that $\denot{\sqlquery}_{\smtmodel(\vec{\queryopinput})} = \smtmodel(\queryopoutput)$.
\end{proof}

\textsc{Theorem}~\ref{thm:semantics-refinement}.
Given query operator $\queryop$, 
and two UAs $\underapprox \in \underapproxfamily_{\queryop}$ and $\underapprox' \in \underapproxfamily_{\queryop}$ where  $\underapprox'$ refines $\underapprox$ (i.e., $\underapprox' \refines \underapprox$), we have  $\encodeuasemantics(\queryop, \underapprox) \Rightarrow \encodeuasemantics(\queryop, \underapprox')$.
Intuitively, the set of $\queryop$'s reachable outputs for $\underapprox'$ should be a subset of that for $\underapprox$; therefore, the UA semantics encoding for $\underapprox'$ is entailed by that for $\underapprox$.
\begin{proof}
Prove by case analysis of $\queryop$.
\begin{itemize}[leftmargin=*]
\item $\queryop = \proj, \filter, \groupby, \distinct, \unionall$, or $\orderby$. \\
By the definition of $\encodeuasemantics$, it holds that $\encodeuasemantics(\queryop, \underapprox) = \encodechoice(\underapprox) \land \encodefullsemantics(\queryop)$ and $\encodeuasemantics(\queryop, \underapprox') = \encodechoice(\underapprox') \land \encodefullsemantics(\queryop)$.  Then, we will need to show that $\encodechoice(\underapprox) \Rightarrow$ $\encodechoice(\underapprox')$.

Since $\encodechoice(\underapprox) = \bigland_{\underapprox_i \in \underapprox} \encodeuavalue( \underapprox_i )$ and by the definition of $\encodeuavalue$, it holds that (1) $\encodeuavalue(\underapprox'_i) = ( \underapproxvar'_i = \underapprox'_i)$ if $\underapprox'_i \neq \unknownvalue$, and $\encodeuavalue(\underapprox'_i) = \bigvee_{c \in \dom(\underapprox'_i) \setminus \set{\unknown}} \underapproxvar'_i = c$ if $\underapprox'_i = \unknownvalue$, and (2) $\encodeuavalue(\underapprox_i) = ( \underapproxvar_i = \underapprox_i)$ if $\underapprox_i \neq \unknownvalue$, and $\encodeuavalue(\underapprox_i) = \bigvee_{c \in \dom(\underapprox_i) \setminus \set{\unknown}} \underapproxvar_i = c$ if $\underapprox_i = \unknownvalue$.

Also by the value subsumption relation, if $\underapprox \subsumes \underapprox'$, then for all $\underapprox_i \neq \unknownvalue$, $\underapprox_i = \underapprox'_i$. Therefore, for each $\underapprox_i$ and $\underapprox'_i$, the disjunctive clause encoded by $\encodeuavalue( \underapprox'_i )$ is a subset of and is weaker than that encoded by $\encodeuavalue( \underapprox_i )$.  It holds that $\encodechoice(\underapprox) \Rightarrow$ $\encodechoice(\underapprox')$.  Therefore, it is proved that $\encodeuasemantics(\queryop, \underapprox) \Rightarrow \encodeuasemantics(\queryop, \underapprox')$.

\item $\queryop = \product, \ijoin, \ljoin, \rjoin$, or $\fjoin$. \\
By the definition of $\encodeuasemantics$, it holds that $\encodeuasemantics(\queryop, \underapprox) = \encodechoice(\underapprox) \land \encodefullsemantics(\queryop)$ and $\encodeuasemantics(\queryop, \underapprox') = \encodechoice(\underapprox') \land \encodefullsemantics(\queryop)$.

Then, we will need to show that $\encodechoice(\underapprox) \Rightarrow$ $\encodechoice(\underapprox')$.

By the definition $\encodechoice(\underapprox) = \bigland_{\underapprox_{i,j} \in \underapprox} \encodeuavalue( \underapprox_{i,j} )$ and by the definition of $\encodeuavalue$, it holds that (1) $\encodeuavalue(\underapprox'_{i,j}) = ( \underapproxvar'_{i,j} = \underapprox'_{i,j})$ if $\underapprox'_{i,j} \neq \unknownvalue$, and $\encodeuavalue$ $(\underapprox'_{i,j}) =$ $\bigvee_{c \in \dom(\underapprox'_{i,j}) \setminus \set{\unknown}} \underapproxvar'_{i,j} = c$ if $\underapprox'_{i,j} = \unknownvalue$; (2) $\encodeuavalue(\underapprox_{i,j}) = ( \underapproxvar_{i,j} = \underapprox_{i,j})$ if $\underapprox_{i,j} \neq \unknownvalue$, and $\encodeuavalue$ $(\underapprox_{i,j}) =$ $\bigvee_{c \in \dom(\underapprox_{i,j}) \setminus \set{\unknown}} \underapproxvar_{i,j} = c$ if $\underapprox_{i,j} = \unknownvalue$.

Also by the value subsumption relation, if $\underapprox \subsumes \underapprox'$, then for all $\underapprox_{i,j} \neq \unknownvalue$, $\underapprox_{i,j} = \underapprox'_{i,j}$. Therefore, for each $\underapprox_{i,j}$ and $\underapprox'_{i,j}$, the disjunctive clause encoded by $\encodeuavalue( \underapprox'_{i,j} )$ is a subset of and is weaker than that encoded by $\encodeuavalue( \underapprox_{i,j} )$.  It holds that $\encodechoice(\underapprox) \Rightarrow$ $\encodechoice(\underapprox')$.  Therefore, it is proved that $\encodeuasemantics(\queryop, \underapprox) \Rightarrow \encodeuasemantics(\queryop, \underapprox')$.
\end{itemize}
\end{proof}

\begin{corollary}[Refinement of UA semantics for queries] \label{cor:ua-semantics-refinement}
Given a query $\sqlquery$ and a UA map $\mapastnodetounderapprox$ that maps each AST node of P to a UA choice, let formula $\formula$ be the result of $\encodeuasemantics(\mapastnodetounderapprox)$.  For a UA map $\mapastnodetounderapprox'$ where for each node $v \in \dom(\mapastnodetounderapprox')$, $\mapastnodetounderapprox(v) \subsumes \mapastnodetounderapprox'(v)$, it holds that $\encodeuasemantics(\mapastnodetounderapprox) \Rightarrow \encodeuasemantics(\mapastnodetounderapprox')$.
\end{corollary}
\begin{proof}
By the definition of $\encodeuasemantics$, the resulting formula is a conjunction of the formulas returned by $\encodeuasemantics(\astnode, \underapprox)$ where $(\astnode \mapsto \underapprox) \in \mapastnodetounderapprox$ and the inputs $\queryopinput_1, \mydots, \queryopinput_l$ to $\astnode$ are the outputs $\queryopoutput^{\astnode_1}, \mydots, \queryopoutput^{\astnode_l}$ produced by its children AST nodes $\astnode_1, \mydots, \astnode_l$.  By Theorem ~\ref{thm:semantics-correctness}, we know that (1) $\encodeuasemantics(op(\astnode_i), \mapastnodetounderapprox(\astnode_i)) \Rightarrow \encodeuasemantics(op(\astnode_i), \mapastnodetounderapprox'(\astnode_i))$ holds for $1 \leq i \leq l$, and (2) $\encodeuasemantics(op(\astnode), \mapastnodetounderapprox(\astnode_i)) \Rightarrow \encodeuasemantics(op(\astnode), \mapastnodetounderapprox'(\astnode_i))$.  Therefore it also holds that $\encodeuasemantics(\mapastnodetounderapprox) \Rightarrow \encodeuasemantics(\mapastnodetounderapprox')$.
\end{proof}

\begin{lemma}[Soundness of Conflict-Driven UA Search] \label{lem:soundness-ua-search}
Given queries $\sqlquery_1, \mydots, \sqlquery_n$ and an application condition $\appcond$, if the procedure $\conflictdrivenUAsearch(\sqlquery_1, \mydots, \sqlquery_n, \appcond)$ returns a UA map $\mapastnodetounderapprox$ and $\mapastnodetounderapprox \neq null$, then $\encodeuasemantics(\mapastnodetounderapprox) \land \appcond$ is satisfiable.
\end{lemma}
\begin{proof}
Suppose at line ~\ref{alg:conflict-driven-UA-search:condition} of Algorithm ~\ref{alg:conflict-driven-UA-search}, we have a UA map $\mapastnodetounderapprox^*$.
If $\conflictdrivenUAsearch$ $(\sqlquery_1, \mydots, \sqlquery_n, \appcond)$ returns a UA map $\mapastnodetounderapprox$ and $\mapastnodetounderapprox \neq null$, by line ~\ref{alg:conflict-driven-UA-search:sat-case} of Algorithm ~\ref{alg:conflict-driven-UA-search}, the formula $\Phi = \encodeuasemantics(\mapastnodetounderapprox^*) \land \appcond$ is satisfiable.
Also by line ~\ref{alg:conflict-driven-UA-search:obtain-model-and-update} of Algorithm ~\ref{alg:conflict-driven-UA-search}, for each AST node $\astnode \in \dom(\mapastnodetounderapprox^*)$, it holds that $v \in \dom(\mapastnodetounderapprox)$ and $\mapastnodetounderapprox^*(v) \subsumes \mapastnodetounderapprox(v)$.  Thus, by Corollary ~\ref{cor:ua-semantics-refinement}, we have $\encodeuasemantics(\mapastnodetounderapprox^*) \Rightarrow \encodeuasemantics(\mapastnodetounderapprox)$.
Since $\Phi = \encodeuasemantics(\mapastnodetounderapprox^*) \land \appcond$ is satisfiable, $\encodeuasemantics(\mapastnodetounderapprox) \land \appcond$ is also satisfiable.
\end{proof}

\textsc{Theorem}~\ref{thm:soundness}.
Given queries $\sqlquery_1, \mydots, \sqlquery_n$ and application condition $\appcond$, if $\geninput$ returns an input database $\inputdatabase$, then we have $\appcond[ \queryoutput_1 \mapsto \sqlquery_1(\inputdatabase), \mydots, \queryoutput_n \mapsto \sqlquery_n(\inputdatabase) ]$ is true. 
\begin{proof}
By line ~\ref{alg:top-level:call-search} of Algorithm~\ref{alg:top-level} and Lemma~\ref{lem:soundness-ua-search}, if $\conflictdrivenUAsearch(\sqlquery_1, \mydots, \sqlquery_n, \appcond)$ returns a UA map $\mapastnodetounderapprox$ and $\mapastnodetounderapprox \neq null$, then the formula $\Phi = \encodeuasemantics(\mapastnodetounderapprox) \land \appcond$ is satisfiable.
Then for any model $\smtmodel$ of $\Phi$, by Corollary~\ref{cor:ua-semantics-correcteness} and soundness of $\buildinputdatabase$, the database $\inputdatabase$ is consistent with $\smtmodel(\vec{x})$ and $\denot{\sqlquery}_{\smtmodel(\vec{x})} = \smtmodel(y)$ where $\smtmodel(y) \models \appcond$.  Therefore, if $\geninput(\sqlquery_1, \mydots, \sqlquery_n, \appcond)$ returns a database $\inputdatabase$, then $\inputdatabase$ satisfies $\appcond$. 
\end{proof}

\begin{lemma} \label{lem:omega-all-incorrect-ua}
Given a UA map $\mapastnodetounderapprox$, a subset of AST nodes $\conflictASTnodes$, an application condition $\appcond$, and a set of conflicts $\conflicts$, let $(\mapastnodetounderapprox', \conflicts')$ be the result of $\fixconflict(\mapastnodetounderapprox, \conflictASTnodes, \appcond, \conflicts)$. If any $\mapastnodetounderapprox^* \in \conflicts$ makes $\encodeuasemantics(\mapastnodetounderapprox^*) \land \appcond$ unsatisfiable, then any $\mapastnodetounderapprox \in \conflicts'$ also makes $\encodeuasemantics(\mapastnodetounderapprox) \land \appcond$ unsatisfiable.
\end{lemma}
\begin{proof}
Since by line~\ref{alg:conflict-driven-UA-search:resolve-conflict} of Algorithm~\ref{alg:conflict-driven-UA-search}, $\conflictASTnodes$ is a subset of nodes in $\mapastnodetounderapprox$ such that $\encodeuasemantics(\conflictASTnodes) \land \appcond$ is unsatisfiable, and by soundness of $\extractconflict$, $\encodeuasemantics(\mapastnodetounderapprox \mapprojection \conflictASTnodes) \land \appcond$ is also unsatisfiable.
Also by line ~\ref{alg:resolveconflict:encode-semantics-at-V} and line ~\ref{alg:resolveconflict:check-conflict} of Algorithm ~\ref{alg:resolveconflict}, we know that if $\encodingall_{\conflictASTnodes}$ is unsatisfiable, then $\encode(\mapastnodetounderapprox_{\conflictASTnodes}) \land  \appcond$ is unsatisfiable.
Thus, if $\mapastnodetounderapprox^* \in \conflicts$ such that $\encodeuasemantics(\mapastnodetounderapprox^*) \land \appcond$ is unsatisfiable, then it also holds any $\mapastnodetounderapprox \in \conflicts'$ also makes $\encodeuasemantics(\mapastnodetounderapprox) \land \appcond$ unsatisfiable.
\end{proof}

\begin{lemma} \label{lem:omega-progress}
Given a UA map $\mapastnodetounderapprox$, a subset of AST nodes $\conflictASTnodes$, an application condition $\appcond$, and a set of conflicts $\conflicts$, let $(\mapastnodetounderapprox', \conflicts')$ be the result of $\fixconflict(\mapastnodetounderapprox, \conflictASTnodes, \appcond, \conflicts)$. $\conflicts'$ is a proper superset of $\conflicts$, i.e., $\conflicts' \supset \conflicts$.
\end{lemma}
\begin{proof}
By line ~\ref{alg:resolveconflict:add-conflict} of Algorithm~\ref{alg:resolveconflict}, $\conflicts' = \conflicts \ \cup \ \{  \mapastnodetounderapprox \mapprojection \conflictASTnodes  \}$.  Also by line ~\ref{alg:resolveconflict:check-conflict} of Algorithm~\ref{alg:resolveconflict}, $\conflicts' = \conflicts' \cup \{ \mapastnodetounderapprox_{\conflictASTnodes} \}$ when $\encodingall_{\conflictASTnodes}$ is not satisfiable.  Then evidently, the result $\conflicts'$ will contain at least one additional element $\mapastnodetounderapprox \mapprojection \conflictASTnodes$ and $\conflicts'$ also contains all the elements of $\conflicts$.  Therefore, it holds that $\conflicts' \supset \conflicts$.
\end{proof}
\begin{lemma} \label{lem:resolve-conflict-exhaust}
Given a UA map $\mapastnodetounderapprox$, a subset of AST nodes $\conflictASTnodes$, an application condition $\appcond$, and a set of conflicts $\conflicts$, let $(\mapastnodetounderapprox', \conflicts')$ be the result of $\fixconflict(\mapastnodetounderapprox, \conflictASTnodes, \appcond, \conflicts)$. All possible UA map $
\mapastnodetounderapprox^* =
\big\{ \astnode_i \mapsto \underapprox_i \ | \ \underapprox_i \in \splitUAspace(\underapproxfamily_{\astnode_i}), \astnode_i \in \conflictASTnodes \big\}
\uplus
\big\{  \astnode \mapsto \topUA(\underapproxfamily_{\astnode}) \ | \  \astnode \in \dom(\mapastnodetounderapprox ) \setminus \conflictASTnodes \big\}
$ collectively cover the space of $\big\{ \underapprox^* \in \underapproxfamily_{\astnode^*} \ | \ \astnode^* \in \mapastnodetounderapprox^* \big\} \setminus \conflicts$, and if for any $\mapastnodetounderapprox^*$, $\encodeuasemantics(\mapastnodetounderapprox^*) \land \appcond$ is unsatisfiable, then $\mapastnodetounderapprox' = null$.
\end{lemma}
\begin{proof}
By line ~\ref{alg:resolveconflict:loop-begin} of Algorithm~\ref{alg:resolveconflict}, the loop enumerates each possible covering set $S \subseteq \underapproxfamily_{\astnode_i}$ returned by $\splitUAspace(\underapproxfamily_{\astnode_i})$. 
Thus, the UAs represented by all $S$ collectively cover the entire space of $\underapproxfamily_{\astnode_i}$ for any AST node $\astnode_i \in \conflictASTnodes$.  Also, for an AST node $\astnode \in \dom(\mapastnodetounderapprox ) \setminus \conflictASTnodes$, $\mapastnodetounderapprox^*(\astnode) = \topUA(\underapproxfamily_\astnode)$.
Therefore, all possible $\mapastnodetounderapprox^*$ cover the UAs of $\big\{ \underapprox^* \in \underapproxfamily_{\astnode^*} \ | \ \astnode^* \in \mapastnodetounderapprox^* \big\} \setminus \conflicts$.

We know that upon the termination of the loop without invoking line~\ref{alg:resolveconflict:return-M'} of Algorithm~\ref{alg:resolveconflict}, it holds that there does not exist an $\mapastnodetounderapprox^*$ such that $\encodeuasemantics(\mapastnodetounderapprox^*) \land \appcond$ is satisfiable.  Since by line ~\ref{alg:resolveconflict:loop-begin} of Algorithm~\ref{alg:resolveconflict} and the property of $\splitUAspace$, there are finitely possible covering sets.  Then it holds that if no $
\mapastnodetounderapprox^* =
\big\{ \astnode_i \mapsto \underapprox_i \ | \ \underapprox_i \in \splitUAspace(\underapproxfamily_{\astnode_i}), \astnode_i \in \conflictASTnodes \big\}
\uplus
\big\{  \astnode \mapsto \topUA(\underapproxfamily_{\astnode}) \ | \  \astnode \in \dom(\mapastnodetounderapprox ) \setminus \conflictASTnodes \big\}
$ makes $\encodeuasemantics(\mapastnodetounderapprox^*) \land \appcond$ satisfiable, the Algorithm~\ref{alg:resolveconflict} terminates with $\mapastnodetounderapprox' = null$.
\end{proof}

\begin{lemma}[Completeness of Conflict-Driven UA Search]\label{lem:completeness-ua-search}
Given queries $\sqlquery_1, \mydots, \sqlquery_n$ and an application condition $\appcond$, $\conflictdrivenUAsearch(\sqlquery_1, \mydots, \sqlquery_n, \appcond)$ terminates, and if it returns $null$, then for any UA map $\mapastnodetounderapprox$ where all AST nodes have minimal UAs, $\encodeuasemantics(\mapastnodetounderapprox) \land \appcond$ is unsatisfiable.
\end{lemma}
\begin{proof}
By Lemma ~\ref{lem:omega-progress}, the conflicts $\conflicts$ strictly grows with additional conflicts each time the procedure $\fixconflict$ is invoked at line ~\ref{alg:conflict-driven-UA-search:resolve-conflict} of Algorithm~\ref{alg:conflict-driven-UA-search}.  Since $\conflicts$ has a finite upper bound, which is the search space of the UA family of all AST nodes (a finite space), it therefore can be proved that $\conflictdrivenUAsearch$ is terminating as the procedure $\fixconflict$ may be called by a finite number of times.

By line ~\ref{alg:conflict-driven-UA-search:resolve-conflict}, line ~\ref{alg:conflict-driven-UA-search:return-null} of Algorithm~\ref{alg:conflict-driven-UA-search}, and Lemma ~\ref{lem:resolve-conflict-exhaust}, if $\mapastnodetounderapprox = null$, then for any $M = \big\{ \underapprox^* \in \underapproxfamily_{\astnode^*} \ | \ \astnode^* \in \mapastnodetounderapprox^* \big\} \setminus \conflicts$, $\encodeuasemantics(\mapastnodetounderapprox) \land \appcond$ is unsatisfiable.  Also by Lemma ~\ref{lem:omega-all-incorrect-ua}, any $\mapastnodetounderapprox' \in \conflicts$ makes $\encodeuasemantics(\mapastnodetounderapprox') \land \appcond$ unsatisfiable.  Thus it holds that if $\mapastnodetounderapprox = null$, then for any UA map $\mapastnodetounderapprox$ where all AST nodes have minimal UAs, $\encodeuasemantics(\mapastnodetounderapprox) \land \appcond$ is unsatisfiable.
\end{proof}

\textsc{Theorem}~\ref{thm:completeness}.
Given queries $\sqlquery_1, \mydots, \sqlquery_n$ and application condition $\appcond$, if $\geninput$ returns $null$, then there does not exist an input $\inputdatabase$ (with respect to the semantics presented in Figure~\ref{fig:full-semantics}) for which $\appcond[ \queryoutput_1 \mapsto \sqlquery_1(\inputdatabase), \mydots, \queryoutput_n \mapsto \sqlquery_n(\inputdatabase)]$ is true. 
\begin{proof}
By line ~\ref{alg:top-level:call-search}, line ~\ref{alg:top-level:null} of Algorithm ~\ref{alg:top-level}, and Lemma~\ref{lem:completeness-ua-search}, if $\conflictdrivenUAsearch$ returns $null$, then it has for any UA map $\mapastnodetounderapprox$ where all AST nodes have minimal UAs, $\encodeuasemantics(\mapastnodetounderapprox) \land \appcond$ is unsatisfiable.  Therefore, it holds that if $\geninput(\sqlquery_1, \mydots, \sqlquery_n, \appcond)$ returns $null$, then there does not exist a database $\inputdatabase$ satisfies $\appcond$.
\end{proof}

%% file: sections/appendix/disambiguation-cond.tex
\section{Encoding Disambiguation Condition} \label{sec:disambiguation-encoding}
In this section, we present the detailed encoding of the disambiguation condition.

Given $\sqlquery_1, \mydots, \sqlquery_n$, the disambiguation condition $\appcond$ that evenly splits $n$ queries into $m$ clusters is encoded as a formula $\Phi$ as follows.

\begin{align*}
\Phi & = \bigland_{i = 1, \mydots, n} \Big( \land_{j = 1}^m (\querybelongstocluster(n, m) \to \denot{P_i} = R_m) \land \sum_{j=1}^{m} \indicator(\querybelongstocluster(n, m)) = 1 \Big) \\
& \land \bigland_{i = 1, \mydots, m} \Big( \sum_{j=1}^{n} \indicator(\querybelongstocluster(n, m)) = \frac{n}{m} \land \land_{i'=1}^{m} (i \neq i' \to R_i \neq R'_i) \Big)
\end{align*}
where $R_1, \mydots, R_m$ are fresh table variables, $b$ is an uninterpreted function that takes as input two integers $n, m$ and returns a boolean indicating if $\sqlquery_n$ belongs to the group $m$.  In our evaluation, we used $n= 50, 100$ for D-50 and D-100, respectively, and $m$ is always 2.


